\RequirePackage{amsmath}
\documentclass{llncs}
\usepackage{amsmath}
\usepackage{amssymb}
\usepackage{url}
\usepackage[bookmarks,bookmarksdepth=2,pdfusetitle]{hyperref}
\usepackage{graphicx,multirow,wasysym}
\usepackage{lmodern}
\usepackage{amsfonts,amsmath}
\usepackage{paralist}
\usepackage{enumerate}
\usepackage{bbm}
\usepackage{tikz}
\usepackage{verbatim}
\usepackage[shortlabels]{enumitem}
\usepackage{comment}
\usepackage{booktabs}
\usepackage{lipsum}
\usepackage{tabularx}
\usepackage{xmpmulti}
\usepackage{multicol}
\usepackage{float}
\floatstyle{boxed}
\restylefloat{figure}
\begin{document}

%\newcommand{\0}{{\mathbf{0}}}

%%%%%%%%%%%%%%%%%%%%%%%%%%%%%%%%%%%%%%%%%%%%%%%%%%%%%%%%%%%%%%%%%
%defined by yanhong for fdgs

% only used in this paper
%\newcommand{\msFDGS}{\mathsf{FDGS}}
%\newcommand{\msGM}{\mathsf{GM}}
%\newcommand{\msTM}{\mathsf{TM}}
%\newcommand{\msInfo}{\mathsf{Info}}
%%%%%%%%%%%%%%%%%%%%%%%%%%%%%%%%%%%%%%%%%%%%%%%%%%%%%%%%%%%%%%%%
% adapted to use in any other paper

\newcommand{\mbbN}{\mathbb{N}}
\newcommand{\mcA}{\mathcal{A}}
\newcommand{\mcB}{\mathcal{B}}
\newcommand{\mcS}{\mathcal{S}}
\newcommand{\mcD}{\mathcal{D}}

\newcommand{\mb}{\mathbf}
\newcommand{\ms}{\mathsf}
\newcommand{\mbb}{\mathbb}
\newcommand{\mc}{\mathcal}

%%%%%%%%%%%%%%%%%%%%%%%%%%%%%%%%%%%%%%%%%%%%%%%%%%%%%%%%%%%%
%\pagestyle{plain}
%\mainmatter
\title{Lattice-Based Group Signatures: \\ Achieving Full Dynamicity (and Deniability) with Ease}
\author{San Ling, Khoa Nguyen, Huaxiong Wang, Yanhong Xu}
\institute{
Division of Mathematical Sciences, \\
School of Physical and Mathematical Sciences,\\
Nanyang Technological University, Singapore.\\
\texttt{\{lingsan,khoantt,hxwang,xu0014ng\}@ntu.edu.sg}
}
\maketitle

\begin{abstract}
Lattice-based group signature is an active research topic in recent years. Since the pioneering work by Gordon, Katz and Vaikuntanathan (Asiacrypt 2010), eight other schemes have been proposed, providing various improvements in terms of security, efficiency and functionality. However, most of the existing constructions work only in the static setting where the group population is fixed at the setup phase. The only two exceptions are the schemes by Langlois et al. (PKC 2014) that handles user revocations (but new users cannot join), and by Libert et al. (Asiacrypt 2016) which addresses the orthogonal problem of dynamic user enrollments (but users cannot be revoked).

~~~In this work, we provide the first lattice-based group signature that offers full dynamicity (i.e., users have the flexibility in joining and leaving the group), and thus, resolve a prominent open problem posed by previous works. Moreover, we achieve this non-trivial feat in a relatively simple manner.  Starting with Libert et al.'s fully static construction (Eurocrypt 2016) - which is arguably the most efficient lattice-based group signature to date, we introduce  simple-but-insightful tweaks that allow to upgrade it directly into the fully dynamic setting. More startlingly, our scheme even produces slightly shorter signatures than the former, thanks to an adaptation of a technique proposed by Ling et al. (PKC 2013), allowing to prove inequalities in zero-knowledge.
%Our design approach consists of upgrading Libert et al.'s static construction (EUROCRYPT 2016) - which is arguably the most efficient lattice-based group signature to date - into the fully dynamic setting. Somewhat surprisingly, our scheme produces slightly shorter signatures than the former, thanks to a new technique for proving inequality in zero-knowledge without relying on any inequality check.
The scheme satisfies the strong security requirements of Bootle et al.'s model (ACNS 2016), under the Short Integer Solution (\textsf{SIS}) and the Learning With Errors (\textsf{LWE}) assumptions.

~~~Furthermore, we demonstrate how to equip the obtained group signature scheme with the deniability functionality in a simple way. This attractive functionality, put forward by Ishida et al. (CANS  2016), enables the tracing authority to provide an evidence that a given user is \emph{not} the owner of a signature in question. In the process, we design a zero-knowledge protocol for proving that a given \textsf{LWE} ciphertext does not decrypt to a particular message.


\end{abstract}

\section{Introduction}
Group signature, introduced by Chaum and van Heyst  \cite{CV91},  is a fundamental anonymity primitive which  allows members of a group to sign messages on behalf of the whole group. Yet, users are kept accountable for the signatures they issue since a tracing authority
can identify them should the need arise. There have been numerous works on group signatures in the last quarter-century.

Ateniese et al.~\cite{ACJT00} proposed the first scalable instantiation meeting the security properties that can be intuitively expected
from the primitive, although clean security notions were not available yet at that time. Bellare et al.~\cite{BMW03} filled this gap by
providing strong security notions for static groups, in which the group population is fixed at the setup phase.
Subsequently,
 Kiayias and Yung~\cite{KY06} and Bellare et al.~\cite{BSZ05} established the models capturing the partially dynamic setting, where users are able to join the group at any time, but once they have done so, they cannot leave the group.
Sakai et al.~\cite{SSEH+12} strengthened these models by suggesting the notion of opening soundness, guaranteeing that a valid signature only traces to one user.
Efficient schemes satisfying these models have been proposed in the random
oracle model~\cite{NguyenS04,DP06}  and in the standard model \cite{Gro07,LPY15}.

One essential functionality of group signatures is the support for membership revocation.
Enabling this feature in an efficient manner is quite challenging, since one has to ensure that revoked users are no longer able to sign messages and the workloads of other parties (managers, non-revoked users, verifiers) do not significantly increase in the meantime.
%For instance, misbehaving members who issue signatures on controversial documents should be revoked from the group. If the scheme does not have a suitable mechanism to handle these cases, then the whole system has to be re-initialized, which is obviously an undesirable solution in practice.
Several different approaches have been suggested~\cite{BressonS01,CL02a,TX03,Nguyen05,BS04} to address this problem, and notable pairing-based constructions supporting both dynamic joining and efficient revocation were given in~\cite{NakanishiFHF09,LPY12a,LPY12b}. Very recently, Bootle et al.~\cite{BCCGG16} pointed out a few shortcomings of previous models, and put forward stringent security notions for fully dynamic group signatures. They also demonstrated a construction satisfying these notions based on the decisional Diffie-Hellman (\textsf{DDH}) assumption, following a generic transformation from a secure accountable ring signature scheme~\cite{BCCGGP15}.

Another interesting functionality of group signatures that we consider in this work is \emph{deniability}, suggested by Ishida et al.~\cite{IEHST16}. It allows the tracing authority to provide a digital evidence that a given group user did \emph{not} generate a signature in question. In~\cite{IEHST16}, Ishida et al. discussed various situations in which such functionality helps to protect the privacy of users. For instance, suppose that the entrace/exit control of  a building is implemented using a group signature scheme, and the police wants to check whether a suspect was in this building at the time of a crime. If the police asks the tracing authority to reveal the signers of all signatures generated during that time period, then this will seriously violate the privacy of innocent users. In such situation, deniable group signatures make it possible to prove that the signer of a given signature is not the suspect, while still keeping the signer's identity secret. As shown by Ishida et al.~\cite{IEHST16}, the main technical challenge towards realizing the deniability functionality consists of constructing a zero-knowledge proof/argument that a given ciphertext does not decrypt to a particular message. Such a mechanism is non-trivial to realize in general, but Ishida et al.~\cite{IEHST16} managed to achieve it from pairing-based assumptions.

For the time being, existing group signature schemes offering full dynamicity and/or deniability all rely on number-theoretic assumptions which are vulnerable to quantum attacks~\cite{Shor97}. To avoid putting all eggs in one basket, it is thus encouraging to consider instantiations based on alternative, post-quantum foundations, e.g., lattice assumptions. In view of this, let us now look at the topic of lattice-based group signatures.
\smallskip

\noindent
{\sc Lattice-based group signatures.}
Lattice-based cryptography has been an exciting research area since the seminal works of Regev~\cite{Regev05} and Gentry et al.~\cite{GPV08}. %Lattices not only allow to build powerful primitives (e.g.,~\cite{Gentry09}) that have no feasible instantiations in conventional number-theoretic cryptography, but they also provide several advantages over the latter, such as conjectured resistance against quantum adversaries and faster arithmetic operations.
Along with other primitives, lattice-based group signature has received noticeable attention in recent years.
The first scheme was introduced by Gordon, Katz and Vaikuntanathan~\cite{GKV10} whose solution produced signature size linear in
the number of group users $N$. Camenisch et al.~\cite{CNR12} then extended \cite{GKV10} to achieve anonymity in the strongest sense.  Later, Laguillaumie et al. \cite{LLLS13-Asiacrypt}  put forward the first scheme with the signature size logarithmic in $N$, at the cost of relatively large parameters. Simpler and more efficient solutions with $\mathcal{O}(\log N)$ signature size were subsequently given by Nguyen et al.~\cite{NZZ15} and Ling et al.~\cite{LNW15}.  Libert et al.~\cite{LLNW16} obtained substantial efficiency
improvements via a construction based on Merkle trees which eliminates the need for GPV trapdoors \cite{GPV08}.
More recently, a scheme supporting message-dependent opening (\textsf{MDO}) feature~\cite{SEHK+12} was proposed in~\cite{LMN16}. All the schemes mentioned above are designed for static groups.

The only two known lattice-based group signatures that have certain dynamic features were proposed by Langlois et al.~\cite{LLNW14-PKC} and Libert et al.~\cite{LLMNW16-dgs}. The former is a scheme with verifier-local revocation (\textsf{VLR})~\cite{BS04}, which means that only the verifiers need to download the up-to-date group information. The latter addresses the orthogonal problem of dynamic user enrollments (but users cannot be revoked). To achieve those partially dynamic functionalities, both of the proposals have to incorporate relatively complicated mechanisms\footnote{Langlois et al. considered users' ``tokens'' as functions of Bonsai signatures~\cite{CHKP10} and associated them with a sophisticated revocation technique, while Libert et al. used a variant of Boyen's signature~\cite{Boy10} to sign users' public keys. Both underlying signature schemes require long keys and lattice trapdoors.} and both heavily rely on lattice trapdoors.

The discussed above situation is somewhat unsatisfactory,  given that the full dynamicity feature is highly desirable in most applications of group signatures (e.g., protecting the privacy of commuters in public transportation), and it has been achieved based on number-theoretic assumptions. %In terms of deniability, the work by Ishida et al.~\cite{IEHST16} is actually the only group signature scheme supporting this functionality.
This motivates us to work on fully dynamic group signatures from lattices. %and it would be better to design a fully dynamic group signature scheme with deniability from lattices.
Besides, considering that the journey to partial dynamicity in previous works~\cite{LLNW14-PKC,LLMNW16-dgs} was shown not easy, we ask ourselves an inspiring question: Can we achieve full dynamicity with ease? Furthermore, given that deniability is an attractive and non-trivial functionality~\cite{IEHST16}, can we also realize it with ease in the context of lattice-based group signatures? At the end of the day, it is good to solve an open research question, but it would be even better and more exciting to do this in a simple way. To make it possible, we will likely need some new and insightful ideas.

\smallskip

\noindent
{\sc Our Results and Techniques. }
Our main contribution is the first fully dynamic group signature from lattices.
The scheme satisfies the strong security requirements put forward by Bootle et al.~\cite{BCCGG16}, under the Short Integer Solution (\textsf{SIS}) and the Learning With Errors (\textsf{LWE}) assumptions. As in all previous lattice-based group signatures, our scheme is analyzed in the random oracle model. Additionally, we incorporate the deniability functionality suggested by Ishida et al.~\cite{IEHST16} into Bootle et al.'s model, and then demonstrate how to make our fully dynamic group signature scheme deniable. %Note that this extended model is carefully designed based on the model of the deniable group signatures by Ishida et al.~\cite{IEHST16}.

\begin{table}[!htbp]
\centering
\setlength{\tabcolsep}{2pt}
\begin{tabular}{|c|c|c|c|c|c|c|c|c|}
  \hline
  \rule{0pt}{2.8ex}
  Scheme &
 \begin{tabular}{c}
    {\small Sig.} \\
    {\small size} \\
  \end{tabular}
& \begin{tabular}{c}
    {\small Group} \\
    {\small PK size} \\
  \end{tabular}
& \begin{tabular}{c}
    {\small Signer's} \\
    {\small SK size} \\
  \end{tabular}
&\begin{tabular}{c}
    {\small Trap-} \\
    {\small door?} \\
  \end{tabular}
&
{\small Model}
   & \begin{tabular}{c}
    \hspace*{-2.5pt}\small{Extra info} \\
    \small {per epoch}\hspace*{-2.5pt} \\
  \end{tabular}
  &\begin{tabular}{c}
    \hspace*{-2.5pt}\small{Denia-} \\
    \hspace*{-2.5pt}\small {bility}  \\
  \end{tabular}
  %& {\small Deni-ty}
  %\begin{tabular}{c}
 %   {\small Assumption} \\
 %   {\small $(\gamma$ \hspace*{-1pt}{\small in} \hspace*{-1pt}$\mathsf{SIVP}_\gamma)$} \\
 % \end{tabular}
 \\
  \hline
  \rule{0pt}{2.8ex}
  {GKV}~\cite{GKV10} & $\widetilde{\mathcal{O}}(\lambda^2\hspace*{-2pt}\cdot\hspace*{-2pt}N)$ & $\widetilde{\mathcal{O}}(\lambda^2\hspace*{-2pt}\cdot\hspace*{-2pt}N)$ & $\widetilde{\mathcal{O}}(\lambda^2)$ & yes & \textsf{static} & \textsf{NA} & ?\\
  \hline
  \rule{0pt}{2.8ex}
  CNR~\cite{CNR12} & $\widetilde{\mathcal{O}}(\lambda^2\hspace*{-2pt}\cdot\hspace*{-2pt}N)$ & $\widetilde{\mathcal{O}}(\lambda^2)$  & $\widetilde{\mathcal{O}}(\lambda^2)$ & yes & \textsf{static} &   \textsf{NA} & ?\\
  \hline
  \rule{0pt}{2.8ex}
  {\small LLLS}~\cite{LLLS13-Asiacrypt} & $\widetilde{\mathcal{O}}(\lambda\hspace*{-2pt}\cdot\hspace*{-2pt}\ell)$& ${\mathcal{O}}(\lambda^2\hspace*{-2pt}\cdot\hspace*{-2pt}\ell)$  & $\widetilde{\mathcal{O}}(\lambda^2)$ & yes & \textsf{static} &  \textsf{NA}& ? \\
   \hline
  \rule{0pt}{2.8ex}
  {\small LLNW}~\cite{LLNW14-PKC} & $\widetilde{\mathcal{O}}(\lambda\hspace*{-2pt}\cdot\hspace*{-2pt}\ell)$ & $\widetilde{\mathcal{O}}(\lambda^2\hspace*{-2pt}\cdot\hspace*{-2pt}\ell)$ & $\widetilde{\mathcal{O}}(\lambda\hspace*{-2pt}\cdot\hspace*{-2pt} \ell)$ & yes & \textsf{VLR} &
  \begin{tabular}{c}
    \small{Sign:} \textbf{no} \\
    \hline
     \rule{0pt}{2.8ex}
    \hspace*{-2.8pt}\small{Ver:} \hspace*{-1.5pt}$\widetilde{\mathcal{O}}(\lambda)\hspace*{-1.65pt}\cdot \hspace*{-1.65pt}R$ \\
  \end{tabular} & ?
 \\
  \hline
  \rule{0pt}{2.8ex}
  NZZ~\cite{NZZ15} & $\widetilde{\mathcal{O}}(\lambda\hspace*{-2pt}+\hspace*{-2pt}\ell^2)$ & $\widetilde{\mathcal{O}}(\lambda^2\hspace*{-2pt}\cdot\hspace*{-2pt}\ell^2)$ & $\widetilde{\mathcal{O}}(\lambda^2)$ & yes & \textsf{static} &  \textsf{NA} & ?\\
  \hline
  \rule{0pt}{2.8ex}
  LNW~\cite{LNW15} & $\widetilde{\mathcal{O}}(\lambda\hspace*{-2pt}\cdot\hspace*{-2pt}\ell)$ & $\widetilde{\mathcal{O}}(\lambda^2\hspace*{-2pt}\cdot\hspace*{-2pt}\ell)$ & $\widetilde{\mathcal{O}}(\lambda)$ & yes & \textsf{static} &  \textsf{NA} & $\star$\\
  \hline
  \rule{0pt}{2.8ex}
  {\small LLNW}~\cite{LLNW16} & $\widetilde{\mathcal{O}}(\lambda\hspace*{-2pt}\cdot\hspace*{-2pt}\ell)$ &  $\widetilde{\mathcal{O}}(\lambda^2 + \lambda\cdot \ell)$ & $\widetilde{\mathcal{O}}(\lambda\cdot \ell)$ & \textbf{FREE} & \textsf{static} & \textsf{NA} & $\star$\\
  \hline
  \rule{0pt}{3.2ex}
  {\small LLM+}~\cite{LLMNW16-dgs} & $\widetilde{\mathcal{O}}(\lambda\hspace*{-2pt}\cdot\hspace*{-2pt}\ell)$ & $\widetilde{\mathcal{O}}(\lambda^2\hspace*{-2pt}\cdot\hspace*{-2pt}\ell)$  & $\widetilde{\mathcal{O}}(\lambda%\hspace*{-2pt}\cdot\hspace*{-2pt} \ell
  )$ & yes &   \begin{tabular}{c}
       {\small\textsf{partially}} \\
        {\small \textsf{dynamic}}
     \end{tabular}
  %\begin{minipage}[t]{0.1\textwidth}
        %{\small \hspace*{2.6pt}\textsf{partially}} \\
        %{\small \hspace*{2.6pt}\textsf{dynamic}}
    % \end{minipage}
  &  \textsf{NA} & $\star$\\

  \hline
  \rule{0pt}{2.8ex}
  LMN~\cite{LMN16} & $\widetilde{\mathcal{O}}(\lambda\hspace*{-2pt}\cdot\hspace*{-2pt}\ell)$ & $\widetilde{\mathcal{O}}(\lambda^2\hspace*{-2pt}\cdot\hspace*{-2pt}\ell)$ & $\widetilde{\mathcal{O}}(\lambda)$ & yes & \textsf{MDO} &  \textsf{NA} & ?\\
  \hline
  \rule{0pt}{3.2ex}
  Ours & $\widetilde{\mathcal{O}}(\lambda\hspace*{-2pt}\cdot\hspace*{-2pt}\ell)$ &  $\widetilde{\mathcal{O}}(\lambda^2 + \lambda\cdot \ell)$ & $\widetilde{\mathcal{O}}(\lambda) + \ell$ & \textbf{FREE} &
  \begin{tabular}{c}
       {\small\textsf{fully}} \\
        {\small \textsf{dynamic}}
     \end{tabular}
   &
     \begin{tabular}{c}
       {\small Sign:} \\ $\widetilde{\mathcal{O}}(\lambda\hspace*{-2pt}\cdot\hspace*{-2pt} \ell)$ \\
       \hline
        \rule{0pt}{2.8ex}
       {\small Ver:} $\widetilde{\mathcal{O}}(\lambda)$
     \end{tabular} & \textbf{yes}
     \\
  \hline
\end{tabular}
\vspace*{0.26cm}
\caption{Comparison of known lattice-based group signatures, in terms of efficiency and functionality. The comparison is done based on two governing parameters: security parameter $\lambda$ and the maximum expected number of group users $N= 2^\ell$.
%The public key of the scheme from~\cite{NZZ-PKC15} contains a constant number of matrices in $\mathbb{Z}_q^{n \times m}$, but the security proof requires $q > N$, and thus, each matrix has size $nm\log q > n^2\cdot \ell^2$ bits.
As for the scheme from~\cite{LLNW14-PKC}, $R$ denotes the number of revoked users at the epoch in question. In the last column, ``$\star$'' means that the respective scheme can directly achieve the deniability feature via our technique, while ``?'' indicates that we do not know how to do so for the respective scheme.}
\label{table-comparison}
\end{table}

For a security parameter parameter $\lambda$ and  maximum expected number of group users~$N$, our scheme features signature size $\widetilde{\mathcal{O}}(\lambda\hspace*{-2pt}\cdot\hspace*{-2pt}\log N)$ and group public key size
$\widetilde{\mathcal{O}}(\lambda^2 + \lambda\cdot \log N)$. The user's secret key has bit-size $\widetilde{\mathcal{O}}(\lambda) + \log N$. At each epoch when the group information is updated, the verifiers only need to download an extra $\widetilde{\mathcal{O}}(\lambda)$ bits in order to perform verification of signatures\footnote{We remark that in the \textsf{DDH}-based instantiation from~\cite{BCCGG16} which relies on the accountable ring signature from~\cite{BCCGGP15}, the verifiers have to periodically download public keys of active signers. Our scheme overcomes this issue, thanks to the use of an updatable accumulator constructed in Section~\ref{section:Merkle-trees}.}, while each active signer only has to download $\widetilde{\mathcal{O}}(\lambda\hspace*{-2pt}\cdot\hspace*{-2pt}\log N)$ bits.
In Table~\ref{table-comparison}, we give a detailed comparison of our scheme with known lattice-based group signatures, in terms of efficiency and functionality. The full dynamicity feature is achieved with a very reasonable cost and without having to rely on lattice trapdoors. Somewhat surprisingly, our scheme even produces shorter signatures than the scheme from~\cite{LLNW16} - which is arguably the most efficient lattice-based group signature known to date. Furthermore, these results are obtained in a relatively simple manner, thanks to three main ideas/techniques discussed below.

Our starting point is the scheme~\cite{LLNW16}, which works in the static setting.
Instead of relying on trapdoor-based ordinary signatures as in prior works, the LLNW scheme employs on a \textsf{SIS}-based Merkle tree accumulator. For a group of $N= 2^\ell$ users, the manager chooses uniformly random vectors $\mathbf{x}_0, \ldots, \mathbf{x}_{N-1}$;  hashes them to $\mathbf{p}_0, \ldots, \mathbf{p}_{N-1}$, respectively; builds a tree on top of these hash values; and publishes the tree root $\mathbf{u}$. The signing key of user $i$ consists of $\mathbf{x}_i$ and the witness for the fact that $\mathbf{p}_i$ was accumulated in $\mathbf{u}$. When issuing signatures, the user proves knowledge of a valid pair $(\mathbf{x}_i, \mathbf{p}_i)$ and of the tree path from~$\mathbf{p}_i$ to~$\mathbf{u}$. The user also has to encrypt the binary representation $\mathsf{bin}(i)$ of his identity~$i$, and prove that the ciphertext is well-formed. The encryption layer is also lattice-trapdoor-free, since it utilizes the Naor-Yung double-encryption paradigm~\cite{NY90} with Regev's \textsf{LWE}-based encryption scheme~\cite{Regev05}.
To upgrade the LLNW scheme directly into a fully dynamic group signature, we now let the user compute the pair $(\mathbf{x}_i, \mathbf{p}_i)$ on his own (for enabling non-frameability), and we employ the following three ideas/techniques.

First, we add a dynamic ingredient into the static Merkle tree accumulator from~\cite{LLNW16}. To this end, we equip it with an efficient updating algorithm with complexity $\mathcal{O}(\log N)$: to change an accumulated value, we simply update the values at the corresponding leaf and along its path to the root.

Second, we create a simple rule to handle user enrollment and revocation efficiently (i.e., without resetting the whole tree). Specifically, we use the updating algorithm to set up the system so that: (i)- If a user has not joined the group or has been revoked, the value at the leaf associated with him is set as~$\mathbf{0}$; (ii)-  When a user joins the group, that value is set as his public key $\mathbf{p}_i$. Our setup guarantees that only active users (i.e., who has joined and has not been revoked at the given epoch) have their \emph{non-zero} public keys accumulated into the updated root. This rule effectively separates active users who can sign from those who cannot: when signing messages, the user proceeds as in the LLNW scheme, and is asked to additionally prove in zero-knowledge that $\mathbf{p}_i \neq \mathbf{0}$. In other words, the seemingly big gap between being fully static and being fully dynamic has been reduced to a small difference!

Third, the arising question now is how to additionally prove the inequality $\mathbf{p}_i \neq \mathbf{0}$ in the framework of the Stern-like~\cite{Ste96} protocol from~\cite{LLNW16}. One would naturally expect that this extra job could be done without losing too much in terms of efficiency. Here, the surprising and somewhat unexpected fact is that we can actually do it while \emph{gaining} efficiency, thanks to the following simple idea. Recall that, in~\cite{LLNW16}, to prove knowledge of $\mathbf{p}_i \in \{0,1\}^{nk}$, an extension technique from~\cite{LNSW13} is employed, in which $\mathbf{p}_i$ is extended into a vector of dimension $2nk$. We note that, the authors of~\cite{LNSW13} also suggested a slightly modified version of their technique, that allows to simultaneously prove that $\mathbf{p}_i \in \{0,1\}^{nk}$ and $\mathbf{p}_i$ is non-zero while working only with dimension $2nk-1$. This intriguing tweak enables us to obtain a zero-knowledge protocol with slightly lower communication cost, and hence, group signatures with slightly smaller size than in~\cite{LLNW16}.

%change for deniability
Another problem we have to overcome is that the fully dynamic setting requires a proof of correct opening,
which boils down to proving correct decryption for Regev's encryption scheme~\cite{Regev05}. It involves modular linear equations with bounded-norm secrets, and can be handled using Stern-like techniques from~\cite{LNSW13,LLMNW16-dgs}.
%To this end, we provide a Stern-like zero-knowledge argument of correct decryption for Regev's encryption scheme~\cite{Regev05}.

Now, to equip the obtained fully dynamic group signature scheme with the deniability functionality, we first incorporate Ishida et al.'s notion of deniability~\cite{IEHST16} into Bootle et al.'s model~\cite{BCCGG16}.  Then we describe how to make our scheme deniable, via a zero-knowledge argument of knowledge that a given Regev ciphertext~\cite{Regev05} does not decrypt to a given message. The latter reduces to proving knowledge of a \emph{non-zero} vector $\mathbf{b}\in\{-1,0,1\}^{\ell}$ satisfying a modular linear equation, where $\mathbf{b}$ is the difference between the result of decryption and the given message. This in turn can be handled by developing the above method for proving inequality in the framework of Stern's protocol.

%Furthermore, this argument system can be easily combined with the tweak described above to realize a non-trivial mechanism: proving that a given ciphertext does \emph{not} decrypt to a given message. This mechanism is essentially what is needed to enable \emph{deniability} functionality, put forward by Ishida et al.~\cite{IEHST16}, which is helpful in protecting users' privacy in many real-life scenario. Moreover, this technique is also applicable for  several other existing lattice-based group signatures~\cite{LNW15,LLNW16,LLMNW16-dgs}.

To summarize, we solve the prominent open question of achieving full dynamicity in the field of lattice-based group signatures. The scheme can be easily extended to a deniable group signature. Moreover, our solution is simple and comparatively efficient. Our results, while not yielding truly practical schemes, would certainly help to enrich the field and bring lattice-based group signatures one step closer to practice.

\smallskip

\noindent
{\sc Remarks. } This article is the full, extended version of~\cite{LNWX17}, published in the proceedings of ACNS 2017. In comparison with~\cite{LNWX17}, the present version additionally introduces the deniability feature to the obtained fully dynamic group signature scheme, and provides details of all the security proofs.

\smallskip

\noindent
{\sc Organization.} In Section~\ref{section:prelim}, we recall some background on fully dynamic group signatures and lattice-based cryptography. Section~\ref{section:Merkle-trees} develops an updatable Merkle tree accumulator. Our main scheme is constructed and analyzed in Section~\ref{section:main-scheme}. In Section~\ref{section:deniability}, we demonstrate how to achieve deniability. %Due to space restriction,
Some details are deferred to  the appendix.

\section{Preliminaries}\label{section:prelim}
\subsection{Notations}

In this paper, all vectors are column vectors. The  concatenation two column vectors $\mathbf{x}\in\mathbb{R}^m$ and $\mathbf{y}\in\mathbb{R}^k$ is denoted by $(\mathbf{x}\|\mathbf{y})\in\mathbb{R}^{m+k}$ for simplicity. % instead of $(\mathbf{x}^T\|\mathbf{y}^T)^T$.
Let $\|\cdot\|$ and $\|\cdot\|_{\infty}$  be the Euclidean ($\ell_2$) norm and infinity ($\ell_{\infty}$) norm of a vector respectively.
For $a\in\mathbb{R}$, $\log a$ denotes logarithm of $a$ with base~$2$.
%For a vector $\mathbf{v}\in\mathbb{R}^m$, let $\mathbf{v}[i]$ be the $i$-th bit of $\mathbf{v}$.
Denote $[m]$ as the set $\{1,2,\cdots, m\}$.
\subsection{Fully Dynamic Group Signatures}\label{subsection:FDGS-definitions}
We recall the definition and security notions of fully dynamic group signatures ($\ms{FDGS}$) presented by Bootle et al.~\cite{BCCGG16}. An $\ms{FDGS}$ scheme involves the following entities: a group manager~$\ms{GM}$ that determines who can join the group, a tracing manager~$\ms{TM}$ who can open signatures, and a set of users who are potential group members. Users can join/leave the group at the discretion of $\ms{GM}$. We assume $\ms{GM}$ will publish some information~$\ms{info}_{\tau}$ regularly associated with a distinct index $\tau$ (referred as epoch hereafter). Without loss of generality, assume there is one-to-one correspondence between information and the associated epoch.
The information describes changes of the group, e.g., current group members or members that are excluded from the group. We assume the published group information is authentic. By comparing current group information with previous one, it allows any party to identify revoked users at current epoch. For simplicity assume $\tau_1< \tau_2$ if $\ms{info}_{\tau_1}$ is published before $\ms{info}_{\tau_2}$, i.e., the epoch preserves the order in which the corresponding group information was published.
In existing models, the keys of authorities are supposed to be generated honestly; while in~\cite{BCCGG16}, Bootle et al. consider a stronger model where the keys of authorities can be maliciously generated by the adversary.

\medskip
\noindent{\bf{Syntax of fully dynamic group signatures.}}
%\subsubsection{Syntax of Fully Dynamic Group Signatures}
An $\ms{FDGS}$ scheme is a tuple of following polynomial-time algorithms.

\begin{description}
        %\rightarrow \ms{pp}
\item $\ms{GSetup}(\lambda)\rightarrow \ms{pp}$. On input security parameter $\lambda$, this algorithm generates public parameters~$\ms{pp}$.% and makes them public to all parties.
      %It will also initialize the registration table $\mb{reg}$.
      %\item $\ms{GKgen_{GM}(pp)}$: Given the public parameters~$\ms{pp}$, group manager generates his secret key pair $(\ms{mpk},\ms{msk})$ and initializes the group information~$\ms{Info}_{\tau_{\init}}$ and registration table~$\mb{reg}$. Output~$\ms{mpk},\ms{Info}_{\tau_{\init}}$.
      %\item $\ms{GKgen_{TM}(pp)$: Given the public parameters $\ms{pp}$, tracing manager generates his secret key pair~ $(\ms{tpk},\ms{tsk})$. Output~$\ms{tpk}$. set $\ms{gpk}=(\ms{pp},\ms{mpk},\ms{tpk})$.
      %    \rightarrow(\ms{tpk},\ms{tsk})  \rightarrow(\ms{Info}_{\tau_{\ms{init}}},\mb{reg},\ms{mpk},\ms{msk})
\item $\langle\ms{GKgen}_{\ms{GM}}(\ms{pp}),\ms{GKgen}_{\ms{TM}}(\ms{pp})\rangle$. This is an interactive protocol between the group manager $\ms{GM}$ and the tracing manager $\ms{TM}$. If it completes successfully,  algorithm $\ms{GKgen}_{\ms{GM}}$ outputs a manager key pair $(\ms{mpk},\ms{msk})$. Meanwhile, $\ms{GM}$ initializes the group information $\ms{info}$ and the registration table $\mb{reg}$. The algorithm $\ms{GKgen}_{\ms{TM}}$ outputs a tracing key pair $(\ms{tpk},\ms{tsk})$. Set group public key $\ms{gpk}=(\ms{pp},\ms{mpk},\ms{tpk})$.   %Given the public parameters~$\ms{pp}$, group manager runs $\ms{GKgen}_{\ms{GM}}(\ms{pp})$ to obtain his secret key pair $(\ms{mpk},\ms{msk})$ and initializes the group information~$\ms{Info}_{\tau_{\ms{init}}}$ and registration table~$\mb{reg}$. While tracing manager runs $\ms{GKgen}_{\ms{TM}}(\ms{pp})$ to generate his secret key pair~$(\ms{tpk},\ms{tsk})$. Set $\ms{gpk}=(\ms{pp},\ms{mpk},\ms{tpk})$.

        %This is an interactive protocol between $\ms{GM}$ and $\ms{TM}$. Assume the communication between them is authentic and private. $\ms{TM}$ generates a key pair from an indistinguishable $\ms{CCA}$ secure public key encryption scheme, denoted as $(\ms{tpk},\ms{tsk})$ while $\ms{GM}$ generates generates $(\ms{mpk},\ms{msk})$and initializes group information $\ms{Info}$ and registration table $\mb{reg}$. If interaction completes successfully, set $\ms{gpk}$ to be $(\ms{pp},\ms{mpk},\ms{tpk})$ and make it available to all parties.
       %The output of $\ms{TM}$ is $(\ms{tpk},\ms{tsk})$, where $\ms{tsk}$ is for tracing. Set the $\ms{gpk}$ to be $(\ms{pp},\ms{mpk},\ms{tpk})$. In our construction, we can extract the set of all user public keys, denoted as $R$, from the $\ms{Info}$.
       % \rightarrow\langle(\ms{Info},\ms{reg},c,\ms{mpk},\ms{msk});(\ms{tpk},\ms{tsk})\rangle,
    % $\rightarrow(\ms{upk},\ms{usk})$:
    \smallskip
\item $\ms{UKgen}(\ms{pp})\rightarrow(\ms{upk},\ms{usk})$. Given public parameters~$\ms{pp}$, this algorithm generates a user key pair $(\ms{upk},\ms{usk})$. % and makes the user public key available to all other parties. We assume anyone can get an authentic copy of it.
\smallskip
\item  $\langle\ms{Join}(\ms{info}_{\tau_{\rm{current}}},\ms{gpk},\ms{upk},\ms{usk});\ms{Issue}(\ms{info}_{\tau_{\rm{current}}},\ms{msk},\ms{upk})\rangle$. This is an interactive algorithm run by a user and $\ms{GM}$. If it completes successfully, this user becomes a group member with an identifier $\ms{uid}$ and the $\ms{Join}$ algorithm stores secret group signing key~$\ms{gsk}[\ms{uid}]$ while $\ms{Issue}$ algorithm stores registration information in the table $\mathbf{reg}$ with index $\ms{uid}$.
    %register the user in $\mb{reg}$ with index $c$, i.e., $\mb{reg}[c]$.
    %$\rightarrow\langle \ms{gsk}_i\rangle$
    %$\rightarrow \ms{Info}_{\tau_1}$
    \smallskip
\item $\ms{GUpdate}(\ms{gpk},\ms{msk},\ms{info}_{\tau_{\rm{current}}},\mc{S},\mb{reg})\rightarrow \ms{info}_{\tau_{\rm{new}}}$. This is an algorithm run by $\ms{GM}$ to update group information while advancing the epoch. Given $\ms{gpk}$,$\ms{msk}$, $\ms{info}_{\tau_{\rm{current}}}$, registration table $\mb{reg}$, a set $\mc{S}$ of active users to be removed from the group, $\ms{GM}$ computes new group information $\ms{info}_{\tau_{\rm{new}}}$ and may update the table $\mb{reg}$. If there is no change to the group, $\ms{GM}$ outputs $\bot$. %This algorithm aborts if there exists a user in $\mc{S}$ who is not active at epoch~$\tau_{\rm{current}}$.
        %\item $\ms{GUpdate}(\ms{gpk},\ms{Info}_{\tau_0},\ms{reg},\ms{msk},\mc{S})\rightarrow \ms{Info}_{\tau_1}$: This is an algorithm run by $\ms{GM}$ to update group information. Given $\ms{gpk}$, group information $\ms{Info}_{\tau_0}$ at epoch $\tau_0$, registration table $\ms{reg}$, new users' identity set $\mc{S}_1$, and revoked users' identity set $\mc{S}_2$, $\ms{GM}$ then computes new group information, denoted as $\ms{Info}_{\tau_1}$ and makes it available to the public. If there is no changes to the group, $\ms{GM}$ outputs $\bot$. This algorithm aborts if there exists a user in $\mc{S}_1\cup\mc{S}_2$ not enrolled or a user in $\mc{S}_1\cap\mc{S}_2$ or a user in $\mc{S}_1$ but being active in the group or a user in $\mc{S}_2$ but being non-active in the group.
        %\rightarrow \Sigma
\item $\ms{Sign}(\ms{gpk},\ms{gsk}[\ms{uid}],\ms{info}_{\tau},M)\rightarrow \Sigma$.
%On inputs $\ms{gpk}$, $\ms{info}_{\tau}$, $\ms{gsk}[\ms{uid}]$ and message $M$,
This algorithm outputs a group signature $\Sigma$ on message $M$ by user~$\ms{uid}$. It outputs $\bot$ if this user is inactive at epoch~$\tau$.
        %If the user $i$ is not active at current group, then this algorithm outputs abort.
        %Denote $N_{\tau}$ to be total number of users at time $\tau$. In our construction, $N_{\tau}$ can be known from $\ms{Info}_{\tau}$. Let $l_{\tau}=\lceil \log N_{\tau}\rceil$.
        %\rightarrow 0/1
\item  $\ms{Verify}(\ms{gpk},\ms{info}_{\tau},M,\Sigma)\rightarrow 0/1$. This algorithm checks the validity of the signature $\Sigma$ on message $M$ at epoch~$\tau$.
%This is a deterministic algorithm that on inputs $\ms{gpk}$, $\ms{info}_{\tau}$ and a message-signature pair $(M,\Sigma)$ it checks whether this message-signature pair is valid at epoch $\tau$ or not.% and output $1$ if yes and $0$ otherwise.
    %\rightarrow(\mb{p},\pi_2)
\item $\ms{Trace}(\ms{gpk},\ms{tsk},\ms{info}_{\tau},\mb{reg},M,\Sigma)\rightarrow(\ms{uid},\Pi_{\ms{trace}})$. This is an algorithm run by~$\ms{TM}$.
%On inputs $\ms{gpk}$, $\ms{tsk}$, $\ms{info}_{\tau}$, $\mb{reg}$, and message-signature pair $(M,\Sigma)$,
Given the inputs, $\ms{TM}$ returns an identity $\ms{uid}$ of a group member who signed the message and a proof indicating this tracing result or $\bot$ if it fails to trace to a group member.
    %    If the algorithm cannot trace to a group member, $\ms{TM}$ outputs $(0,\bot)$.    This algorithm can be run by any party.
\item $\ms{Judge}(\ms{gpk},\ms{uid},\ms{info}_{\tau},\Pi_{\ms{trace}},M,\Sigma)\rightarrow 0/1$. This algorithm checks the validity of $\Pi_{\ms{trace}}$ outputted by the $\ms{Trace}$ algorithm.
%On inputs $\ms{gpk}$, $\ms{Info}_{\tau}$ at epoch $\tau$, message-signature pair $(M,\Sigma)$ and the tracing result $(\ms{uid},\Pi_{\ms{trace}})$, this algorithm checks the validity of $\Pi_{\ms{trace}}$ and outputs $1$ if it is valid and $0$ otherwise.
    %is a correct proof or not. This algorithm will deal with $(0,\bot)$ case differently since there is no proof here.
%\item $\ms{IsActive}(\ms{Info}_{\tau},\mb{reg},i)\rightarrow 0/1$: It outputs $1$ if user $i$ is active at time $\tau$ and $0$ otherwise. %Remember that $R_{\tau}$ is the  non-revoked user public key set.
 \end{description}

\medskip
\noindent{\bf{Correctness and security of fully dynamic group signatures.}}
%\subsubsection{Correctness and Security of Fully Dynamic Group Signatures}
As put forward by Bootle et al.~\cite{BCCGG16}, an $\ms{FDGS}$ scheme must satisfy \emph{correctness}, \emph{anonymity}, \emph{non-frameability}, \emph{traceability} and \emph{tracing soundness}.
\smallskip

\noindent
\emph{Correctness} demands that a signature generated by an honest and active user is always accepted by algorithm $\mathsf{Verify}$, and that algorithm $\ms{Trace}$ can always identify that user, as well as produces a proof accepted by algorithm $\ms{Judge}$.
%given the message-signature pair $(M,\Sigma)$ at epoch $\tau$, the $\ms{Trace}$ algorithm will return correct signer identity and proof $\Pi_{\ms{trace}}$ which will be accepted by the $\ms{Judge}$ algorithm.
\smallskip

\noindent\emph{Anonymity} requires that it is infeasible for any PPT adversary to distinguish signatures generated by two active users of its choice at the chosen epoch, even if it can corrupt any user, can choose the key of $\mathsf{GM}$, and is given access to the $\ms{Trace}$ oracle.
%any PPT adversary without tracing secret key $\ms{tsk}$ cannot recover the identity of a signature. It suffices that the adversary cannot distinguish signatures of a message from one of two users of his choice even though he may learns the user secret keys and private group signing keys of these two users.
\smallskip

\noindent{\emph{Non-Frameability} makes sure that the adversary cannot generate a valid signature that traces to an honest user even if it can corrupt all other users, and can choose keys of both managers.
\smallskip

\noindent\emph{Traceability} ensures that the adversary cannot produce a valid signature that cannot be traced to an active user at the chosen epoch, even if it can corrupt any user and can choose the key of $\mathsf{TM}$.

\smallskip

\noindent\emph{Tracing Soundness} requires that it is infeasible to produce a valid signature that traces to two different users, even if all group users and both managers are fully controlled by the adversary.
\smallskip

%Formal definitions of correctness and security requirements are deferred to Appendix \ref{appendix-security-fdgs-definition}.
We will define these requirements using a set of experiments, in which the adversary has access to a set of oracles.
The oracles use the following global variables: $\ms{HUL}$ is a list of users whose keys are generated honestly. $\ms{BUL}$ is a list of users whose secret signing keys are revealed to the adversary; $\ms{CUL}$ is a list of users whose public keys are chosen by the adversary; $\ms{SL}$ is a set of signatures that are generated by the oracle $\ms{Sign}$; $\ms{CL}$ is a set of signatures that are generated by the oracle $\ms{Chal}_b$. %The details of these oracles are given in Fig.~\ref{fig-oracle}.

\begin{description}
\item $\ms{AddU}(\ms{uid})$. This oracle adds an honest user $\ms{uid}$ to the group at current epoch. It returns $\ms{upk}[\ms{uid}]$ and adds $\ms{uid}$ to the list $\ms{HUL}$.
%This oracle runs the $\ms{UKgen}$ algorithm and adds the $upk$ to the list $\ms{HUL}$. It then runs the $<\ms{Join},\ms{Issue}>$ algorithm at current epoch. It returns user public key $\ms{upk}$ and identity $\ms{uid}$.
\item $\ms{CrptU}(\ms{uid},\ms{pk})$. This oracle allows the adversary to create a new user by choosing the key $\ms{upk}[\ms{uid}]$ as $\ms{pk}$. It then adds $\ms{uid}$ to the list $\ms{CUL}$. It returns~$\bot$ if $\ms{uid}\in \ms{HUL}\cup\ms{CUL}$. This actually prepares for calling the oracle $\ms{SndToGM}$.
%This oracle allows the adversary to add a user to the group by choosing the key $\ms{upk}$. It then adds $\ms{upk}$ to the list $\ms{CUL}$. It returns~$\bot$ if $\ms{upk}\in \ms{HUL}\cup\ms{CUL}$. %set $\mb{upk}[\ms{uid}]$ to be any user public key $\ms{pk}$ chosen by the adversary. This oracle also initializes the the state of the group manager in preparation for calling $\ms{SndToGM}$ oracle.
\item  $\ms{SndToGM}(\ms{uid},M_{\rm{in}})$. This oracle engages in the $\langle\ms{Join},\ms{Issue}\rangle$ protocol between the adversary, who corrupts the user $\ms{uid}$ in the list $\ms{CUL}$, and the honest $\ms{Issue}$-executing group manager.
% completely wrong the below description
%This oracle acts as the group manager at current epoch and interacts with the adversary who corrupts the user with $\ms{upk}$ in the list $\ms{CUL}$.
%This oracle acts as the honest group manager and interacts with the adversary who corrupts the user $\ms{upk}$ by running $\ms{Issue}$ algorithm at current epoch. This user with public key $\ms{upk}$ must be in the list $\ms{CUL}$.  %It returns a identity $\ms{uid}$ for this user.
%\item $\ms{UKgen}(\ms{pp})$. This oracle runs the $\ms{UKgen}$ algorithm and adds the $upk$ to the list $\ms{HUL}$. Then it returns $\ms{upk}$.
\item $\ms{SndToU}(\ms{uid},M_{\rm{in}})$. This oracle engages in the $\langle\ms{Join},\ms{Issue}\rangle$ protocol between the adversary, who corrupts the group manager, and the honest $\ms{Join}-$executing user $\ms{uid}$ in the list $\ms{HUL}$.

%This oracle returns $\bot$ if $\ms{upk}\notin \ms{HUL}$. Otherwise it acts as an honest user at current epoch and interacts with the adversary who corrupts the group manager. %by running the $\ms{Join}$ algorithm at current epoch.
\item $\ms{RReg}(\ms{uid})$. This oracle returns registration information  $\mb{reg}[\ms{uid}]$ of user $\ms{uid}$.% to the adversary.
\item $\ms{MReg}(\ms{uid},\rho)$. This oracle modifies $\mb{reg}[\ms{uid}]$ to any $\rho$ chosen by the adversary.
\item $\ms{RevealU}(\ms{uid})$. This oracle returns secret signing key $\ms{gsk}[\ms{uid}]$ to the adversary, and adds this user to the list $\ms{BUL}$.
%This oracle reveals user secret key $\mb{usk}$ and private group signing key $\ms{gsk}$ corresponding to this public key to the adversary and add this user to the list $\ms{BUL}$. This user must be in the list $\ms{HUL}$ but not $\ms{BUL}$ or $\ms{CUL}$.
\item $\ms{Sign}(\ms{uid},M,\tau)$. This oracle returns a signature on the message $M$ by the user $\ms{uid}$ at epoch $\tau$, and then add the signature to the list $\ms{SL}$.
%The oracle runs the $\ms{Sign}(\ms{gpk},\ms{gsk}[\ms{uid}],\ms{info}_\tau,M)$ algorithm if everything is well defined and user with this identity $\ms{uid}$ is in the list $\ms{HUL}$ and active at current epoch $\tau$. It returns the signature $\Sigma$ to the adversary and adds $\Sigma$ to the list $\ms{SL}$.
\item $\ms{Chal}_b(\ms{info}_{\tau},\ms{uid}_0,\ms{uid}_1,M)$. This oracle receives %$\ms{info}_\tau$, two user identities and message $M$
 inputs from the adversary, returns a signature on $M$ by the user $\ms{uid}_b$ at epoch $\tau$, and then adds the signature to the list $\ms{CL}$. It is required that the two challenged users are active at epoch $\tau$ and this oracle can only be called once.
%This oracle runs the $\ms{Sign}(\ms{gpk},\ms{gsk}[\ms{uid}_b],\ms{info}_\tau,M)$ algorithm if everything is well defined and both users with the identities $\ms{uid}_0,\ms{uid}_1$ are in the list $\ms{HUL}$ and active at current epoch $\tau$. It returns the signature $\Sigma$ to the adversary and adds $(M,\Sigma,\tau)$ to the list $\ms{CL}$. This oracle is designed for defining anonymity experiment and can only be called once.

\item $\ms{Trace}(M,\Sigma,\ms{Info}_{\tau})$. This oracle returns the signer of the signature together with a proof, with respect to the epoch $\tau$. We require that the signature is not in the list $\ms{CL}$.
%This oracle runs the $\ms{Trace}(\ms{gpk},\ms{tsk},\ms{info}_\tau,\mb{reg},M,\Sigma)$ algorithm and return the result to the adversary. It is require that $(M,\Sigma,\tau)$ is not in the list $\ms{CL}$ and $\ms{Verify}(\ms{gpk},\ms{info}_\tau,M,\Sigma)=1$.
\item $\ms{GUpdate}(\mc{S})$. This oracle allows the adversary to update the group at current epoch $\tau_{\rm{current}}$. It is required that $\mc{S}$ is a set of active users at current epoch.
%This oracle runs $\ms{GUpdate}(\ms{gpk},\ms{msk},\ms{info}_{\tau_{\rm{current}}},\mc{S},\mb{reg})$ algorithm to allow the adversary to remove some active users from the group.% at epoch $\tau$.
\end{description}
We also need the following polynomial-time algorithm in security experiments to ease composition.

$\ms{IsActive}(\ms{info}_{\tau},\mb{reg},\ms{uid})\rightarrow 0/1$: It outputs $1$ if this user is active at epoch $\tau$ and $0$ otherwise.

We refer the readers to~\cite{BCCGG16} for detailed descriptions of the above oracles.

\begin{definition}
For any security parameter $\lambda$ and any PPT adversary $\mc{A}$, we define correctness and security experiments in Fig.~\ref{fig-exp}.

For correctness, non-frameability, traceability and tracing soundness, the advantage of the adversary is defined as the probability of outputting~$1$ in the corresponding experiment. For anonymity, the advantage is defined as the absolute difference of probability of outputting~$1$ between experiment $\mb{Exp}_{\ms{FDGS},\mcA}^{\ms{anon-1}}$ and experiment $\mb{Exp}_{\ms{FDGS},\mcA}^{\ms{anon-0}}$.

A fully dynamic group signature scheme is said to be correct and secure (i.e., anonymous, non-frameable, traceable and tracing sound) if the advantages of the adversary in all considered experiments are negligible in~$\lambda$.
\end{definition}

\begin{figure}%[!htb]

\begin{center}
\begin{minipage}{11.6cm}
\underline{Experiment: $\mb{Exp}_{\ms{FDGS},\mcA}^{\ms{corr}}(\lambda)$}\\
     $\ms{pp}\leftarrow \ms{GSetup}(\lambda),\ms{HUL}:=\emptyset$.\\
    $\langle(\ms{info},\ms{mpk},\ms{msk});(\ms{tpk},\ms{tsk})\rangle\leftarrow \langle\ms{GKgen_{GM}}(\ms{pp}),\ms{GKgen_{TM}}(\ms{pp})\rangle$.\\
     Set $\ms{gpk}=(\ms{pp},\ms{mpk},\ms{tpk})$.
    $(\ms{uid},M,\tau)\leftarrow\mcA^{\ms{AddU,RReg,GUpdate}}(\ms{gpk},\ms{info})$.\\
     If $\ms{uid}\notin\ms{HUL}$ or $\ms{gsk}[\ms{uid}]=\bot$ or $\ms{info}_{\tau}=\bot$ or $\ms{IsActive}(\ms{info}_{\tau},\mb{reg},\ms{uid})=0$, return $0$.\\
    $\Sigma\leftarrow\ms{Sign}(\ms{gpk},\ms{gsk}[\ms{uid}],\ms{info}_{\tau},M)$,
    $(\ms{uid}^*,\Pi_{\ms{trace}})\leftarrow\ms{Trace}(\ms{gpk},\ms{tsk},\ms{info}_{\tau},\mb{reg},M,\Sigma)$.\\
    Return $1$ if $(\ms{Verify}(\ms{gpk},\ms{info}_{\tau},M,\Sigma)=0\vee \ms{Judge}(\ms{gpk},\ms{uid},\ms{info}_{\tau},\Pi_{\ms{trace}},m,\Sigma)=0\vee\ms{uid}\neq \ms{uid}^*  ).$\\
\underline{Experiment: $\mb{Exp}_{\ms{FDGS},\mcA}^{\ms{anon-b}}(\lambda)$}\\
     $\ms{pp}\leftarrow \ms{GSetup}(\lambda),\ms{HUL,CUL,BUL,CL,SL}:=\emptyset$.\\
     $(\ms{st}_{\ms{init}},\ms{info},\ms{mpk},\ms{msk})\leftarrow\mcA^{\langle\cdot,\ms{GKgen_{TM}\rangle}(\ms{pp})}(\ms{init}:\ms{pp})$.\\
     Return $0$ if $\ms{GKgen_{TM}}$ did not accept or $\mcA$'s output is not well-formed.\\
     Denote the output of $\ms{GKgen_{TM}}$ as $(\ms{tpk},\ms{tsk})$, and set $\ms{gpk}=(\ms{pp},\ms{mpk},\ms{tpk})$.\\
     $b^*\leftarrow\mcA^{\ms{AddU,CrptU,RevealU,SndToU,Trace,MReg,Chal_b}}(\ms{play}:\ms{st_{init}},\ms{gpk})$.
     Return $b^*$.\\
\underline{Experiment: $\mb{Exp}_{\ms{FDGS},\mcA}^{\ms{non-frame}}(\lambda)$}\\
     $\ms{pp}\leftarrow \ms{GSetup}(\lambda),\ms{HUL,CUL,BUL,SL}=\emptyset$.\\
    $(\ms{st}_{\ms{init}},\ms{info},\ms{mpk},\ms{msk},\ms{tpk},\ms{tsk})\leftarrow\mcA(\ms{init}:\ms{pp})$.\\
     Return $0$ if $\mcA$'s output is not well-formed, otherwise set $\ms{gpk}=(\ms{pp},\ms{mpk},\ms{tpk})$.\\
     $(M,\Sigma,\ms{uid},\Pi_{\ms{trace}},\ms{info}_{\tau})\leftarrow\mcA^{\ms{CrptU,RevealU,SndToU,MReg,Sign}}(\ms{play}:\ms{st_{init}},\ms{gpk})$.\\
     Return $1$ if
     $(\ms{Verify}(\ms{gpk},\ms{info}_{\tau},M,\Sigma)=1\wedge
     \ms{Judge}(\ms{gpk},\ms{uid},\ms{info}_{\tau},\Pi_{\ms{trace}},M,\Sigma)=1\wedge $\\
     $\ms{uid}\in\ms{HUL}\setminus\ms{BUL}\wedge(M,\Sigma,\tau)\notin \ms{SL})$.\\
\underline{Experiment: $\mb{Exp}_{\ms{FDGS},\mcA}^{\ms{trace}}(\lambda)$}\\
     $\ms{pp}\leftarrow \ms{GSetup}(\lambda),\ms{HUL,CUL,BUL,SL}=\emptyset$.\\
    $(\ms{st}_{\ms{init}},\ms{tpk},\ms{tsk})\leftarrow\mcA^{\langle\ms{GKgen_{GM}}(\ms{pp}),\cdot\rangle}(\ms{init}:\ms{pp})$.\\
     Return $0$ if $\ms{GKgen_{GM}}$ did not accept or $\mcA$'s output is not well-formed.\\
     Denote the output of $\ms{GKgen_{GM}}$ as $(\ms{mpk},\ms{msk},\ms{info})$, and set $\ms{gpk}=(\ms{pp},\ms{mpk},\ms{tpk})$.\\
     $(M,\Sigma,\tau)\leftarrow\mcA^{\ms{AddU,CrptU,SndToGM,RevealU,MReg,Sign,GUpdate}}(\ms{play}:\ms{st_{init}},\ms{gpk},\ms{info})$.\\
     If $\ms{Verify}(\ms{gpk},\ms{info}_{\tau},M,\Sigma)=0$, return $0$.\\
     $(\ms{uid},\Pi_{\ms{trace}})\leftarrow\ms{Trace}(\ms{gpk},\ms{tsk},\ms{info}_{\tau},\mb{reg},M,\Sigma)$.\\
     Return $1$ if
      $(\ms{IsActive}(\ms{info}_{\tau},\mb{reg},\ms{uid})=\bot \vee
     \ms{Judge}(\ms{gpk},\ms{uid},\ms{info}_{\tau},\Pi_{\ms{trace}},M,\Sigma)=0\vee\ms{uid}=0 )$.\\
\underline{Experiment: $\mb{Exp}_{\ms{FDGS},\mcA}^{\ms{trace-sound}}(\lambda)$}\\
    $\ms{pp}\leftarrow \ms{GSetup}(\lambda),\ms{CUL}=\emptyset$.
     $(\ms{st}_{\ms{init}},\ms{info},\ms{mpk},\ms{msk},\ms{tpk},\ms{tsk})\leftarrow\mcA(\ms{init}:\ms{pp})$.\\
     Return $0$ if $\mcA$'s output is not well-formed, otherwise set $\ms{gpk}=(\ms{pp},\ms{mpk},\ms{tpk})$.\\
    $(M,\Sigma,\ms{uid}_0,\Pi_{\ms{trace},0},\ms{uid}_1,\Pi_{\ms{trace},1},\ms{info}_{\tau})\leftarrow\mcA^{\ms{CrptU,MReg}}(\ms{play}:\ms{st_{init}},\ms{gpk})$.\\
    Return $1$ if
    $(\ms{Verify}(\ms{gpk},\ms{info}_{\tau},M,\Sigma)=1\wedge \ms{uid}_0~(\neq \bot)\neq\ms{uid}_1~(\neq \bot)\wedge \ms{Judge}(\ms{gpk},\ms{uid}_b,\ms{info}_{\tau},\Pi_{\ms{trace},b},M,\Sigma)=1 ~\text{for} ~b~\in\{0,1\})$.
     %$\ms{uid}_0=\ms{uid}_1$ or $\ms{uid}_0=\bot$ or $\ms{uid}_0=\bot$, then return $0$.\\
     %If $\exists~b\in\{0,1\}$, s.t. $\ms{Judge}(\ms{gpk},\ms{uid}_b,\ms{info}_{\tau},\Pi_{\ms{trace},b},M,\Sigma)=0$, then return $0$.\\ %\mb{upk}[\ms{uid}]
     %Return $1$.
  \end{minipage}
\end{center}
\caption{Experiments to define security requirements of $\ms{FDGS}$.}\label{fig-exp}
\end{figure}

\subsection{Background on Lattices}
We recall the average-case lattice problems \textsf{SIS} and \textsf{LWE}, together with their hardness results.
\begin{definition}[\cite{Ajtai96,GPV08}]
\emph{
The $\mathsf{SIS}^{\infty}_{n,m,q,\beta}$ problem is as follows: Given uniformly random matrix $\mathbf{A} \in \mathbb{Z}_q^{n \times m}$, find a non-zero vector $\mathbf{x} \in \mathbb{Z}^m$ such that~$\|\mathbf{x}\|_\infty \leq \beta$ and $\mathbf{A\cdot x=0} \bmod q.$
}
\end{definition}
If $m, \beta = \mathsf{poly}(n)$, and $q > \beta\cdot\widetilde{\mathcal{O}}(\sqrt{n})$, then the $\mathsf{SIS}^{\infty}_{n,m,q,\beta}$ problem  is at least as hard as worst-case lattice problem $\mathsf{SIVP}_\gamma$ for some $\gamma = \beta \cdot \widetilde{\mathcal{O}}(\sqrt{nm})$ (see, e.g.,~\cite{GPV08,MP13}).
Specifically, when $\beta=1$, $q = \widetilde{\mathcal{O}}(n)$, $m = 2n \lceil\log q \rceil$, the $\mathsf{SIS}^{\infty}_{n,m,q,1}$ problem is at least as hard as $\mathsf{SIVP}_{\gamma}$ with~$\gamma =\widetilde{\mathcal{O}}(n)$.

In the last decade, numerous \textsf{SIS}-based cryptographic primitives have been proposed. In this work, we will extensively employ $2$ such constructions:
\begin{itemize}%[leftmargin=0.25cm]
\item Our group signature scheme is based on the Merkle tree accumulator from~\cite{LLNW16}, which is built upon a specific family of collision-resistant hash functions.
\item Our zero-knowledge argument systems use the statistically hiding and computationally binding string commitment scheme from~\cite{KTX08}.
\end{itemize}
For appropriate setting of parameters, the security of the above two constructions can be based on the worst-case hardness of $\mathsf{SIVP}_{\widetilde{\mathcal{O}}(n)}$.

\smallskip

In the group signature in Section~\ref{section:main-scheme}, we will employ the multi-bit version of Regev's encryption scheme~\cite{Regev05}, presented in~\cite{KTX07}. The scheme is based on the hardness of the \textsf{LWE} problem.
\begin{definition}[\cite{Regev05}]
\emph{Let $n,m \geq 1$, $q \geq 2$, and let $\chi$ be a probability distribution on $\mathbb{Z}$. For $\mathbf{s} \in \mathbb{Z}_q^n$, let $\mathcal{A}_{\mathbf{s}, \chi}$ be the distribution obtained by sampling $\mathbf{a} \xleftarrow{\$} \mathbb{Z}_q^n$ and $e \hookleftarrow \chi$, and outputting $(\mathbf{a}, \mathbf{s}^\top\cdot\mathbf{a} + e) \in \mathbb{Z}_q^n \times \mathbb{Z}_q$. The $\textsf{LWE}_{n,q,\chi}$ problem asks to distinguish~$m$ samples chosen according to $\mathcal{A}_{\mathbf{s},\chi}$ (for $\mathbf{s} \xleftarrow{\$} \mathbb{Z}_q^n$) and $m$ samples chosen according to the uniform distribution over $\mathbb{Z}_q^n \times \mathbb{Z}_q$.}
\end{definition}
If $q$ is a prime power, $\chi$ is the discrete Gaussian distribution $D_{\mathbb{Z}, \alpha q}$, where $\alpha q \geq 2\sqrt{n}$, then $\mathsf{LWE}_{n,q,\chi}$ is as least as hard as $\mathsf{SIVP}_{\widetilde{\mathcal{O}}(n/\alpha)}$ (see~\cite{Regev05,MM11,MP13}).

\subsection{Stern-like Protocols for Lattice-Based Relations}\label{subsection:Stern-background}
The zero-knowledge (\textsf{ZK}) argument systems appearing in this paper operate in the framework of
Stern's protocol~\cite{Ste96}. Let us now recall some background. This protocol was originally proposed in the context of code-based cryptography, and was later adapted into the lattice setting by Kawachi et al.~\cite{KTX08}. Subsequently, it was empowered by Ling et al.~\cite{LNSW13} to handle the matrix-vector relations associated with the \textsf{SIS} and \textsf{LWE} problems, and further developed to design several lattice-based schemes: group signatures~\cite{LLNW14-PKC,LNW15,LLNW16,LMN16,LLMNW16-dgs}, policy-based signatures~\cite{ChengNW16} and group encryption~\cite{LLMNW16-ge}.

Stern-like protocols are quite useful in the context of lattice-based privacy-preserving systems, when one typically works with modular linear equations of the form $\sum_{i}\mathbf{M}_i\cdot \mathbf{x}_i = \mathbf{v} \bmod q$ - where $\{\mathbf{M}_i\}_i$, $\mathbf{v}$ are public, and one wants to prove in \textsf{ZK} that secret vectors $\{\mathbf{x}_i\}_i$ satisfy certain constraints, e.g., they have small norms and/or have coordinates arranged in a special way. The high-level ideas can be summarized as follows. If the given constraints are invariant under certain type of permutations of coordinates, then one readily uses uniformly random permutations to prove those constraints. Otherwise, one performs some pre-processings with $\{\mathbf{x}_i\}_i$ to reduce to the former case. Meanwhile, to prove that the modular linear equation holds, one makes use of a standard masking technique.

The basic protocol consists of $3$ moves: commitment, challenge, response.
If the statistically hiding and computationally binding string commitment scheme from~\cite{KTX08} is employed in the first move, then one obtains a statistical zero-knowledge argument of knowledge (\textsf{ZKAoK}) with perfect completeness, constant soundness error $2/3$, and communication cost $\mathcal{O}(|w|\cdot \log q)$, where $|w|$ denotes the total bit-size of the secret vectors. In many applications, the protocol is repeated $\kappa = \omega(\log \lambda)$ times, for security parameter $\lambda$, to achieve negligible soundness error, and then made non-interactive via the Fiat-Shamir heuristic~\cite{FS86}. In the random oracle model, this results in a non-interactive zero-knowledge argument of knowledge (\textsf{NIZKAoK}) with bit-size $\mathcal{O}(|w|\cdot \log q)\cdot \omega(\log \lambda)$.

\section{Updatable Lattice-Based Merkle Hash Trees}\label{section:Merkle-trees}
We first recall the lattice-based Merkle-tree accumulator from~\cite{LLNW16}, and then, we equip it with a simple updating algorithm which allows to change an accumulated value in time logarithmic in the size of the accumulated set. This updatable hash tree will serve as the building block of our construction in Section~\ref{section:main-scheme}.

\subsection{Cryptographic Accumulators}
An \emph{accumulator scheme} is a tuple of polynomial-time algorithms defined below. %$(\mathsf{TSetup}, \mathsf{TAcc}, \mathsf{TWitness}, \mathsf{TVerify})$ defined as follows:
\begin{description}
  \item[$\mathsf{TSetup}(\lambda)$] On input security parameter $\lambda$, output the public parameter~$\ms{pp}$. %The other algorithms take $pp$ as an implicit input.
  \item[$\mathsf{TAcc}_{\ms{pp}}$]
  On input a set $R = \{\mathbf{d}_0, \ldots, \mathbf{d}_{N-1}\}$ of $N$ data values, output an accumulator value $\mathbf{u}$.

  \item[$\mathsf{TWitness}_{\ms{pp}}$]
  On input a data set $R$ and a value $\mathbf{d}$, output $\bot$ if $\mathbf{d} \notin R$; otherwise output a witness~$w$ for the fact that~$\mathbf{d}$ is accumulated in~$\mathsf{TAcc}(R)$. (Typically, the size of ${w}$ should be short (\emph{e.g.}, constant or logarithmic in $N$) to be useful.)

  \item[$\mathsf{TVerify}_{\ms{pp}}$]
  On input accumulator value $\mathbf{u}$ and a value-witness pair $(\mathbf{d}, w)$, output~$1$ (which indicates that $(\mathbf{d}, w)$ is valid for the accumulator $\mathbf{u}$) or $0$.
%  \item $\ms{TUpdate}_{\ms{pp}}$ On input a new value $\mb{d}^*$ and $\ell$ bits $(j_1,j_2,\cdots,j_\ell)$, set $\mb{d}_{j}=\mb{d}^*$, where $j=\sum_{i=\ell-1}^{0}2^i\cdot j_{\ell-i}$. Update the accumulator value and some tree nodes values which will be specified later.
 \end{description}
An accumulator scheme is called correct if for all $\ms{pp}\leftarrow \mathsf{TSetup}(\lambda)$, we have  $\mathsf{TVerify}_{\ms{pp}}\big(\mathsf{TAcc}_{\ms{pp}}(R), \mathbf{d}, \mathsf{TWitness}_{\ms{pp}}(R, \mathbf{d})\big) = 1$ for all $\mathbf{d} \in R$.

A natural security requirement for accumulators, as considered in~\cite{BariP97,CL02a,LLNW16}, says that it is infeasible to prove that a value $\mathbf{d}^*$ was accumulated in a value $\mathbf{u}$ if it was not. This property is formalized as follows.
\begin{definition}[\cite{LLNW16}]\label{definition:security-of-acc}
{\emph
An accumulator scheme $(\mathsf{TSetup}, \mathsf{TAcc}, \mathsf{TWitness}, \mathsf{TVerify})$ is called secure if for all PPT adversaries $\mathcal{A}$:
\begin{eqnarray*}
\mathrm{Pr}\big[\ms{pp} \leftarrow \mathsf{TSetup}(\lambda); (R, \mathbf{d}^*, w^*) \leftarrow \mathcal{A}(\ms{pp}):~~~~~~~~~~~~~~~~~~~~~ \\
\mathbf{d}^* \not \in R \wedge \mathsf{TVerify}_{\ms{pp}}(\mathsf{TAcc}_{\ms{pp}}(R), \mathbf{d}^*, w^*)=1\big] = \ms{negl}(\lambda).
\end{eqnarray*}
}
\end{definition}

\subsection{The LLNW Merkle-tree Accumulator}\label{subsection:LLNW-tree}
{\sc Notations. } Hereunder we will use the notation $x \xleftarrow{\$} S$ to indicate that $x$ is chosen uniformly at random from finite set $S$. For bit $b \in \{0,1\}$, we let $\bar{b}= 1-b$.

\smallskip

\noindent
The Merkle-tree accumulator scheme from~\cite{LLNW16} works with parameters $n=\mathcal{O}(\lambda)$, $q=\widetilde{\mathcal{O}}(n^{1.5})$, $k=\lceil \log_2 q\rceil$, and $m=2nk$. The set $\mathbb{Z}_q$ is identified by $\{0,\ldots, q-1\}$. %Define ``powers-of-2'' matrix $\mb{G}=\mb{I}_n\otimes \mb{g}^t$, where $\mb{g}^t=(1,2,\cdots, 2^{k-1})$.
Define the ``powers-of-2'' matrix
\begin{eqnarray*}
\mathbf{G} = \left[\begin{array}{cccc}
               1~2~4~\ldots~2^{k-1} &  &  &  \\
                %& 1~2~4~\ldots~2^{k-1} &  &  \\
                &  & \ldots&  \\
                &  &  & 1~2~4~\ldots~2^{k-1}
             \end{array}\right] \in \mathbb{Z}_q^{n \times nk}.
\end{eqnarray*}
Note that for every $\mathbf{v} \in \mathbb{Z}_q^n$, we have $\mathbf{v} = \mathbf{G}\cdot \mathsf{bin}(\mathbf{v})$, where $\mathsf{bin}(\mathbf{v}) \in \{0,1\}^{nk}$ denotes the binary representation of $\mathbf{v}$. The scheme is built upon the following family of \textsf{SIS}-based collision-resistant hash functions.
\begin{definition}\label{definition:hash-family}
The function family $\mathcal{H}$ mapping $\{0,1\}^{nk} \times \{0,1\}^{nk}$ to $\{0,1\}^{nk}$ is defined as $\mathcal{H} = \{h_{\mathbf{A}} \hspace*{2.5pt} | \hspace*{2.5pt} \mathbf{A} \in \mathbb{Z}_q^{n \times m}\}$, where for $\mathbf{A} = [\mathbf{A}_0 | \mathbf{A}_1]$ with $\mathbf{A}_0, \mathbf{A}_1 \in \mathbb{Z}_q^{n \times nk}$, and for any $(\mathbf{u}_0, \mathbf{u}_1) \in \{0,1\}^{nk} \times \{0,1\}^{nk}$, we have:
\[
h_{\mathbf{A}}(\mathbf{u}_0, \mathbf{u}_1)= \mathsf{bin}\big(\mathbf{A}_0 \cdot \mathbf{u}_0 + \mathbf{A}_1\cdot \mathbf{u}_1 \bmod q\big) \in \{0,1\}^{nk}.
\]
\end{definition}
Note that $h_{\mathbf{A}}(\mathbf{u}_0, \mathbf{u}_1) = \mathbf{u} \Leftrightarrow \mathbf{A}_0\cdot \mathbf{u}_0  +  \mathbf{A}_1 \cdot \mathbf{u}_1 = \mathbf{G}\cdot \mathbf{u} \bmod q$.

%\begin{lemma}[\cite{LLNW16}]\label{lemma:collision-resistant-hash}
%The function family $\mathcal{H}$ is collision-resistant, assuming the hardness of the $\mathsf{SIS}_{n, m, q, 1}^\infty$ problem.
%\end{lemma}

\begin{comment}
\begin{proof}
Given $\mathbf{A} = [\mathbf{A}_0 | \mathbf{A}_1] \xleftarrow{\$} \mathbb{Z}_q^{n \times m}$, if one can find two \emph{distinct} pairs $(\mathbf{u}_0, \mathbf{u}_1) \in \big(\{0,1\}^{nk}\big)^2$ and $(\mathbf{v}_0, \mathbf{v}_1)\in \big(\{0,1\}^{nk}\big)^2$ such that $h_{\mathbf{A}}(\mathbf{u}_0, \mathbf{u}_1) = h_{\mathbf{A}}(\mathbf{v}_0, \mathbf{v}_1) \bmod q$, then one can obtain a \emph{non-zero} vector $\mathbf{z} =
\left(
  \begin{array}{c}
    \mathbf{u}_0 - \mathbf{v}_0 \\
    \mathbf{u}_1 - \mathbf{v}_1  \\
  \end{array}
\right)
\in \{-1,0,1\}^m$ such that $$\mathbf{A}\cdot\mathbf{z} = \mathbf{A}_0\cdot(\mathbf{u}_0 - \mathbf{v}_0) + \mathbf{A}_1\cdot (\mathbf{u}_1 - \mathbf{v}_1) = \mathbf{G}\cdot h_{\mathbf{A}}(\mathbf{u}_0, \mathbf{u}_1) - \mathbf{G}\cdot h_{\mathbf{A}}(\mathbf{v}_0, \mathbf{v}_1) = \mathbf{0} \bmod q.$$ In other words, $\mathbf{z}$ is a valid solution to the $\mathsf{SIS}_{n, m, q, 1}^\infty$ problem associated with matrix $\mathbf{A}$. The lemma then follows from the worst-case to average-case reduction from~$\mathsf{SIVP}_{\widetilde{\mathcal{O}}(n)}$.\qed
\end{proof}
\end{comment}
%\subsection{Our Merkle-Tree Accumulator}

\medskip
A Merkle tree with $N=2^\ell$ leaves, where $\ell$ is a positive integer, then can be constructed based on the function family $\mathcal{H}$ as follows.
\begin{description}
  \item[$\mathsf{TSetup}(\lambda).$] Sample $\mathbf{A} \xleftarrow{\$} \mathbb{Z}_q^{n \times m}$, and output $\ms{pp}= \mathbf{A}$.
      \smallskip

  \item[$\mathsf{TAcc}_{\mathbf{A}}(R = \{\mathbf{d}_0 \in \{0,1\}^{nk}, \ldots, \mathbf{d}_{N-1} \in \{0,1\}^{nk}\}).$]
  For every $j \in [0,N-1]$, let $\mathsf{bin}(j)= (j_1, \ldots, j_\ell) \in \{0,1\}^\ell$ be the binary representation of~$j$, and let  $\mathbf{d}_j=\mathbf{u}_{j_1, \ldots, j_\ell}$. Form the tree of depth~$\ell = \log N$ based on the $N$ leaves $\mathbf{u}_{0,0,\ldots, 0}, \ldots, \mathbf{u}_{1, 1,\ldots, 1}$ as follows:
    \begin{enumerate}
        \item At depth $i\in [\ell]$, the node $\mathbf{u}_{b_1,\ldots, b_i} \in \{0,1\}^{nk}$, for all $(b_1,\ldots, b_i) \in \{0,1\}^i$, is defined as $h_{\mathbf{A}}(\mathbf{u}_{b_1,\ldots, b_i,0}, \mathbf{u}_{b_1,\ldots, b_i, 1})$.\smallskip
        \item At depth $0$: The root $\mathbf{u} \in \{0,1\}^{nk}$ is defined as $h_{\mathbf{A}}(\mathbf{u}_0, \mathbf{u}_1)$.
    \end{enumerate}
  The algorithm outputs the accumulator value $\mathbf{u}$.\smallskip

  \item[$\mathsf{TWitness}_{\mathbf{A}}(R, \mathbf{d}).$] If $\mathbf{d} \not \in R$, return $\bot$. Otherwise, $\mathbf{d} = \mathbf{d}_j$ for some $j \in [0, N-1]$ with binary representation $(j_1,\ldots, j_\ell)$. Output the witness $w$ defined as:
       \begin{eqnarray*}%\label{equation:witness-definition}
       w =  \big((j_1,\ldots, j_\ell), (\mathbf{u}_{j_1,\ldots, j_{\ell-1},\bar{j_\ell}}, \ldots, \mathbf{u}_{j_1,\bar{j_2}}, \mathbf{u}_{\bar{j_1}}) \big) \in \{0,1\}^\ell \times \big(\{0,1\}^{nk}\big)^\ell,
       \end{eqnarray*}
  for $\mathbf{u}_{j_1,\ldots, j_{\ell-1},\bar{j_\ell}}, \ldots, \mathbf{u}_{j_1,\bar{j_2}}, \mathbf{u}_{\bar{j_1}}$ computed by algorithm $\mathsf{TAcc}_{\mathbf{A}}(R)$.
  \smallskip
  \item[$\mathsf{TVerify}_{\mathbf{A}}\big(\mathbf{u}, \mathbf{d}, w\big).$] Let the given witness $w$ be of the form:
  \[
  w = \big((j_1,\ldots, j_\ell),(\mathbf{w}_\ell, \ldots, \mathbf{w}_1)\big) \in \{0,1\}^\ell \times \big(\{0,1\}^{nk}\big)^\ell.
  \]
  The algorithm recursively computes the path $\mathbf{v}_\ell, \mathbf{v}_{\ell-1}, \ldots, \mathbf{v}_1, \mathbf{v}_0 \in \{0,1\}^{nk}$ as follows: $\mathbf{v}_\ell = \mathbf{d}$ and
  \begin{eqnarray}\label{equation:path-update}
  \forall \hspace*{1pt}i \in  \{\ell-1, \ldots, 1, 0\}: \hspace*{2pt}
    \mathbf{v}_i = \begin{cases}
                      h_{\mathbf{A}}(\mathbf{v}_{i+1}, \mathbf{w}_{i+1}), \text{ if } j_{i+1}=0; \\
                      h_{\mathbf{A}}(\mathbf{w}_{i+1}, \mathbf{v}_{i+1}), \text{ if } j_{i+1}=1.
                 \end{cases}
  \end{eqnarray}
  Then it returns $1$ if $\mathbf{v}_0 = \mathbf{u}$. Otherwise, it returns $0$.
\end{description}
The following lemma states the correctness and security of the above Merkle tree accumulator.
\begin{lemma}[\cite{LLNW16}]\label{acc-sec}
The given accumulator scheme is correct and is secure in the sense of Definition~\ref{definition:security-of-acc}, assuming the hardness of the $\mathsf{SIS}_{n,m,q,1}^\infty$ problem.
\end{lemma}
\subsection{An Efficient Updating Algorithm}\label{subsection:updating-trees}
Unlike the static group signature scheme from~\cite{LLNW16}, our fully dynamic construction of Section~\ref{section:main-scheme} requires to regularly edit the accumulated values without having to reconstruct the whole tree. To this end, we equip the Merkle tree accumulator from~\cite{LLNW16} with a simple, yet efficient, updating algorithm: to change the value at a given leaf, we simply modify all values in the path from that leaf to the root. The algorithm, which takes as input a bit string $\mathsf{bin}(j) = (j_1,j_2,\ldots,j_\ell)$ and a value $\mathbf{d}^* \in \{0,1\}^{nk}$, is formally described below.

Given the tree in Section~\ref{subsection:LLNW-tree}, algorithm $\ms{TUpdate}_{\mb{A}}((j_1,j_2,\ldots,j_\ell),\mb{d}^*)$ performs the following steps:
    \begin{enumerate}
        \item Let $\mathbf{d}_j$ be the current value at the leaf of position determined by $\mathsf{bin}(j)$, and let
        $((j_1,\ldots, j_\ell),(\mathbf{w}_{j,\ell}, \ldots, \mathbf{w}_{j,1}))$ be its associated witness. \smallskip
        \item Set $\mathbf{v}_{\ell}: = \mathbf{d}^*$ and recursively compute the path $\mb{v}_{\ell}, \mb{v}_{\ell-1},\ldots,\mb{v}_{1},\mb{v}_{0}\in\{0,1\}^{nk}$ as in~(\ref{equation:path-update}).
        \item Set $\mb{u}:=\mb{v}_{0}$; $\mb{u}_{j_1}:=\mb{v}_{1}; \ldots;  \mb{u}_{j_1,j_2,\ldots,j_{\ell-1}}:=\mb{v}_{\ell-1}$; $\mb{u}_{j_1,j_2,\ldots,j_{\ell}}:=\mb{v}_{\ell}=\mb{d}^*$.
    \end{enumerate}
It can be seen that the provided algorithm runs in time $\mathcal{O}(\ell)= \mathcal{O}(\log N)$. In Fig.~\ref{tree-illustration}, we give an illustrative example of a tree with $2^3=8$ leaves.
%\begin{comment}
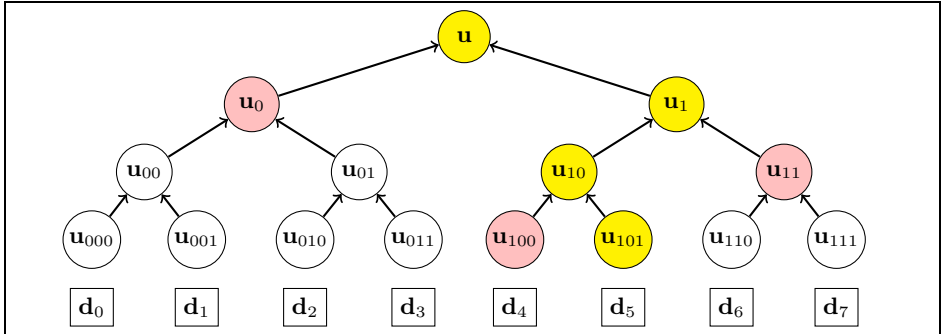
\begin{figure}
  \centering
  \begin{tikzpicture}[scale=0.9]
  %root
  \node[circle,draw, inner sep=4.5pt, fill=yellow] (root) at  (0,0) {$\mathbf{u}$};
  %\draw[line width =0.5] (0,0) circle [radius=10pt] node at (0,0) {$\mathbf{u}$};

  %leaves
  \node[circle,draw, inner sep=1.5pt] (000) at (-5.5,-3) {$\mathbf{u}_{000}$};
  \node[circle,draw, inner sep=1.5pt] (111) at (5.5,-3) {$\mathbf{u}_{111}$};

  \node[circle,draw, inner sep=1.5pt] (011) at (-0.75,-3) {$\mathbf{u}_{011}$};
  \node[circle,draw, inner sep=1.5pt, fill=pink] (100) at (0.75,-3) {$\mathbf{u}_{100}$};

  \node[circle,draw, inner sep=1.5pt] (010) at (-2.35,-3) {$\mathbf{u}_{010}$};
  \node[circle,draw, inner sep=1.5pt, fill=yellow] (101) at (2.35,-3) {$\mathbf{u}_{101}$};

  \node[circle,draw, inner sep=1.5pt] (001) at (-3.95,-3) {$\mathbf{u}_{001}$};
  \node[circle,draw, inner sep=1.5pt] (110) at (3.95,-3) {$\mathbf{u}_{110}$};

  %data
  \node[rectangle,draw, minimum size=0.5cm] (x000) at (-5.5,-4) {$\mathbf{d}_{0}$};
  \node[rectangle,draw, minimum size=0.5cm] (x111) at (5.5,-4) {$\mathbf{d}_{7}$};
  \node[rectangle,draw, minimum size=0.5cm] (x011) at (-0.75,-4) {$\mathbf{d}_{3}$};
  \node[rectangle,draw, minimum size=0.5cm] (x100) at (0.75,-4) {$\mathbf{d}_{4}$};

  \node[rectangle,draw, minimum size=0.5cm] (x010) at (-2.35,-4) {$\mathbf{d}_{2}$};
  \node[rectangle,draw, minimum size=0.5cm] (x101) at (2.35,-4) {$\mathbf{d}_{5}$};

  \node[rectangle,draw, minimum size=0.5cm] (x001) at (-3.95,-4) {$\mathbf{d}_{1}$};
  \node[rectangle,draw, minimum size=0.5cm] (x110) at (3.95,-4) {$\mathbf{d}_{6}$};

  %depth=2
  \node[circle,draw, inner sep=2.5pt] (00) at (-4.725,-2) {$\mathbf{u}_{00}$};
  \node[circle,draw, inner sep=2.5pt, fill=pink] (11) at (4.725,-2) {$\mathbf{u}_{11}$};

  \node[circle,draw, inner sep=2.5pt] (01) at (-1.55,-2) {$\mathbf{u}_{01}$};
  \node[circle,draw, inner sep=2.5pt, fill=yellow] (10) at (1.55,-2) {$\mathbf{u}_{10}$};

  %depth=1
  \node[circle,draw, inner sep=3.5pt, fill=pink] (0) at (-3.1375,-1) {$\mathbf{u}_{0}$};
  \node[circle,draw, inner sep=3.5pt, fill=yellow] (1) at (3.1375,-1) {$\mathbf{u}_{1}$};

  %\draw[line width =0.5] (-3.1375,-1) circle [radius=10pt] node at (-3.1375,-1) {$\mathbf{u}_{0}$};
  %\draw[line width =0.5] (3.1375,-1) circle [radius=10pt] node at (3.1375,-1) {$\mathbf{u}_{1}$};

  %edges
  \draw [->, thick] (0) edge (root);
  \draw [->, thick] (1) edge (root);

  \draw [->, thick] (00) edge (0);
  \draw [->, thick] (01) edge (0);
  \draw [->, thick] (10) edge (1);
  \draw [->, thick] (11) edge (1);

  \draw [->, thick] (000) edge (00);
  \draw [->, thick] (001) edge (00);
  \draw [->, thick] (010) edge (01);
  \draw [->, thick] (011) edge (01);
  \draw [->, thick] (100) edge (10);
  \draw [->, thick] (101) edge (10);
  \draw [->, thick] (110) edge (11);
  \draw [->, thick] (111) edge (11);

 % \draw [->] (x000) edge (000);
 % \draw [->] (x001) edge (001);
 % \draw [->] (x010) edge (010);
 % \draw [->] (x011) edge (011);
 % \draw [->] (x100) edge (100);
 % \draw [->] (x101) edge (101);
 % \draw [->] (x110) edge (110);
 % \draw [->] (x111) edge (111);

  %\draw [line width=0.5] (-7, 0.7) rectangle (-1.15, -3.5) node[scale=0.8] at (-5, -1.75) {$\mathbf{A\cdot x = u} \pmod q$} node[scale=0.8] at (-5, -2.25) {$\mathbf{x} \in \mathsf{Secret}_\beta(d)$} node[scale=0.8] at (-5, -2.75) {$d = (1\ldots 0) \in \{0,1\}^\ell$};
  \end{tikzpicture}
  \caption{A Merkle tree with $2^3= 8$ leaves, which accumulates the data blocks $\mathbf{d}_0, \ldots, \mathbf{d}_7$ into the value $\mathbf{u}$ at the root. The bit string $(101)$ and the pink nodes form a witness to the fact that $\mathbf{d}_5$ is accumulated in $\mathbf{u}$. If we replace $\mb{d}_5$ by a new value $\mb{d}^*$, we only need to update the yellow nodes.  }\label{tree-illustration}
\end{figure}

\section{Our Fully Dynamic Group Signatures from Lattices}\label{section:main-scheme}

In this section, we construct our lattice-based fully dynamic group signature and prove its security in Bootle et al.'s  model~\cite{BCCGG16}. %As such a scheme implies a secure scheme in the BMW model~\cite{BMW03}, a natural approach is to begin with a good lattice-based static group signature and investigate how to make it work in the fully dynamic setting.
We start with the LLNW scheme~\cite{LLNW16}, which works in the static setting.

While other constructions of lattice-based group signatures employ trapdoor-based ordinary signature schemes (e.g., \cite{Boy10,CHKP10}) to certify users, the LLNW scheme relies on a \textsf{SIS}-based Merkle tree accumulator which we recalled in Section~\ref{subsection:LLNW-tree}. The \textsf{GM}, who manages a group of $N= 2^\ell$ users, chooses uniformly random vectors $\mathbf{x}_0, \ldots, \mathbf{x}_{N-1} \in \{0,1\}^m$;  hashes them to $\mathbf{p}_0, \ldots, \mathbf{p}_{N-1} \in \{0,1\}^{nk}$, respectively; builds a tree on top of these hash values; and lets the tree root $\mathbf{u} \in \{0,1\}^{nk}$ be part of the group public key. The signing key of user $i$ consists of $\mathbf{x}_i$ and the witness for the fact that $\mathbf{p}_i$ was accumulated in $\mathbf{u}$. When generating group signatures, the user proves knowledge of a valid pair $(\mathbf{x}_i, \mathbf{p}_i)$ and of the tree path from $\mathbf{p}_i$ to $\mathbf{u}$. The user also has to encrypt the binary representation $\mathsf{bin}(i)$ of his identity~$i$, and prove that the ciphertext is well-formed. The encryption layer utilizes the Naor-Yung double-encryption paradigm~\cite{NY90} with Regev's \textsf{LWE}-based cryptosystem, and thus, it is also lattice-trapdoor-free.

To upgrade the LLNW scheme directly into a fully dynamic group signature, some tweaks and new ideas are needed. First, to enable the non-frameability feature, we let the user compute the pair $(\mathbf{x}_i, \mathbf{p}_i)$ on his own.
The second problem we have to consider is that Merkle hash trees seem to be a static primitive. To this end, we equip the accumulator with an efficient updating algorithm (see Section~\ref{subsection:updating-trees}). Now, the challenging question is how to handle user enrollment and revocation in a simple manner (i.e., without having to reset the whole tree).
To tackle these issues, we associate each of the $N$ potential users with a leaf in the tree, and then use the updating algorithm to set up the system so that:
\begin{enumerate}
%\item  Each of the $N$ potential users is associated with a leaf in the tree;
\item  If a user has not joined the group or has been revoked, the value at the leaf associated with him is set as~$\mathbf{0}$;
\item  When a user joins the group, that value is set as his public key $\mathbf{p}_i$.
\end{enumerate}
Our setup guarantees that only active users (i.e., who has joined and has not been revoked at the given epoch) have their \emph{non-zero} public keys accumulated into the updated root. This effectively gives us a method to separate active users who can sign from those who cannot: when signing messages, the user proceeds as in the LLNW scheme, and is asked to additionally prove in \textsf{ZK} that $\mathbf{p}_i \neq \mathbf{0}$. %In other words, the seemingly big gap between being fully static and being fully dynamic has been reduced to a small difference!

At this point, the arising question is how to prove the inequality $\mathbf{p}_i \neq \mathbf{0}$ in the framework of the Stern-like~\cite{Ste96} protocol from~\cite{LLNW16}. One would naturally hope that this extra job could be done without losing too much in terms of efficiency. Here, the surprising and somewhat unexpected fact is that we can actually do it while \emph{gaining} efficiency, thanks to a technique originally proposed in~\cite{LNSW13}.

To begin with, let $\mathsf{B}_t^L$ denote the set of all vectors in $\{0,1\}^L$ having Hamming weight exactly $t$. In Stern-like protocols (see Section~\ref{subsection:Stern-background}), a common technique for proving in \textsf{ZK} the possession of $\mathbf{p} \in \{0,1\}^{nk}$ consists of appending $nk$ ``dummy'' entries to it to obtain $\mathbf{p}^* \in \mathsf{B}_{nk}^{2nk}$, and demonstrating to the verifier that a random permutation of $\mathbf{p}^*$ belongs to the ``target set'' $\mathsf{B}_{nk}^{2nk}$. This suffices to convince the verifier that the original vector $\mathbf{p}$ belongs to $\{0,1\}^{nk}$, while the latter cannot learn any additional information about $\mathbf{p}$, thanks to the randomness of the permutation. This extending-then-permuting technique was first proposed in~\cite{LNSW13}, and was extensively used in the underlying protocol of the LLNW scheme. Now, to address our question, we will employ a modified version of this technique, which was also initially suggested in~\cite{LNSW13}. Let us think of another ``target set'', so that it is possible to extend  $\mathbf{p} \in \{0,1\}^{nk}$ to an element of that set if and only if $\mathbf{p}$ is non-zero. That set is $\mathsf{B}_{nk}^{2nk-1}$. Indeed, the extended vector $\mathbf{p}^*$ belongs to $\mathsf{B}_{nk}^{2nk-1}$ if and only if the original vector has Hamming weight at least $nk- (nk-1)=1$, which means that it cannot be a zero-vector. When combining with the permuting step, this modification allows us to additionally prove the given inequality while working with smaller dimension. As a result, our fully dynamic scheme produces slightly shorter signatures than the original static scheme.

Finally, we remark that the fully dynamic setting requires a proof of correct opening, which boils down to proving correct decryption for Regev's encryption scheme. It involves modular linear equations with bounded-norm secrets, and can be easily handled using Stern-like techniques from~\cite{LNSW13,LLMNW16-dgs}.

\subsection{Description of the Scheme}\label{subsection:scheme-description}
Our scheme is described as follows.
\begin{description}
  \item[$\ms{GSetup}(\lambda)$.] On input security parameter $\lambda$, this algorithm specifies the following:
    \begin{itemize}
        \item An expected number of potential users $N=2^{\ell}=\mathsf{poly}(\lambda)$. \smallskip
        \item Dimension $n = \mathcal{O}(\lambda)$, prime modulus $q = \widetilde{\mathcal{O}}(n^{1.5})$, and $k=\lceil \log_2 q\rceil$.  These parameters implicitly determine the ``powers-of-$2$'' matrix $\mathbf{G} \in \mathbb{Z}_q^{n \times nk}$, as defined in Section~\ref{section:Merkle-trees}. \smallskip
            %$$\mb{G}=\mb{I}_n\otimes (1, 2, \ldots, 2^{k-1}) \in \mathbb{Z}_q^{n \times nk}.$$
        \item Matrix dimensions $m = 2nk$ for the hashing layer, and $m_E = 2(n+\ell)k$ for the encryption layer. \smallskip

        \item An integer $\beta = \sqrt{n}\cdot \omega(\log n)$, and a $\beta$-bounded noise distribution $\chi$. \smallskip
        \item A hash function $\mathcal{H}_{\mathsf{FS}}: \{0,1\}^* \rightarrow \{1,2,3\}^\kappa$, where $\kappa = \omega(\log \lambda)$, to be modelled as a random oracle in the Fiat-Shamir transformations~\cite{FS86}.
           \item Let $\mathsf{COM}: \{0,1\}^* \times \{0,1\}^m \rightarrow \mathbb{Z}_q^n$ be the string commitment scheme from~\cite{KTX08}, to be used in our zero-knowledge argument systems. %(These parameters will be used in constructing our zero-knowledge argument systems and making them non-interactive via the Fiat-Shamir transformation~\cite{FS86}.)
            \smallskip
        \item Uniformly random matrix $\mathbf{A} \in \mathbb{Z}_q^{n \times m}$. \smallskip
    \end{itemize}
    The algorithm outputs public parameters $$\ms{pp}=\{\lambda, N, n,q,k, m, m_E, \ell,\beta, \chi, \kappa, \mathcal{H}_{\mathsf{FS}}, \mathsf{COM},  \mathbf{A}\}.$$ %\medskip
  \item[$\langle\ms{GKgen_{GM}(pp)},\ms{GKgen_{TM}(pp)}\rangle$.] The group manager \textsf{GM} and the tracing manager \textsf{TM} initialize their keys and the public group information as follows.
    \begin{itemize}
        \item $\mathsf{GKgen_{GM}(pp)}.$ This algorithm samples $\mathsf{msk} \xleftarrow{\$} \{0,1\}^m$ and computes $\mathsf{mpk} = \mathbf{A}\cdot \mathsf{msk} \bmod q$, and outputs $(\mathsf{mpk}, \mathsf{msk})$.
            Here, we consider $\mathsf{mpk}$ as an identifier of the group managed by \textsf{GM} who has $\mathsf{mpk}$ as his public key. Furthermore, as in~\cite[Sec.~3.3,~full~version]{BCCGG16}, we assume that the group information board is visible to everyone, but can only be edited by a party knowing $\mathsf{msk}$. \smallskip
        \item $\mathsf{GKgen_{TM}(pp)}.$ This algorithm initializes the Naor-Yung double-encryption mechanism with the  $\ell$-bit version Regev encryption scheme. It first chooses $\mathbf{B} \xleftarrow{\$} \mathbb{Z}_q^{n \times m_E}$. For each $i \in \{ 1,2\}$, it samples $\mathbf{S}_i\xleftarrow{\$} \chi^{n \times \ell}$, $\mathbf{E}_i \hookleftarrow \chi^{\ell \times m_E}$, and computes $\mathbf{P}_i = \mathbf{S}_i^\top \hspace*{-2.5pt}\cdot\hspace*{-1.5pt} \mathbf{B} + \mathbf{E}_i \in \mathbb{Z}_q^{\ell \times m_E}$. Then, \textsf{TM} sets $\mathsf{tsk} = (\mathbf{S}_1, \mathbf{E}_1)$, and $\mathsf{tpk} = (\mathbf{B}, \mathbf{P}_1, \mathbf{P}_2)$. \smallskip
        \item \textsf{TM} sends $\mathsf{tpk}$ to \textsf{GM} who initializes the following: \smallskip
            \begin{itemize}
                %\item Group information $\mathsf{info}_{0}: = \emptyset$;
                \item Table $\mathbf{reg}: = (\mathbf{reg}[0][1], \mathbf{reg}[0][2], \ldots, \mathbf{reg}[N\hspace*{-2pt}-\hspace*{-2pt}1][1], \mathbf{reg}[N\hspace*{-2pt}-\hspace*{-2pt}1][2])$, where for each $i \in [0, N\hspace*{-1.5pt}-\hspace*{-1.5pt}1]$: $\mathbf{reg}[i][1] = \mathbf{0}^{nk}$ and $\mathbf{reg}[i][2] = 0$. Looking ahead, $\mathbf{reg}[i][1]$ will be used to record the public key of a registered user, while $\mathbf{reg}[i][2]$ stores the epoch at which the user joins.
                    %all entries of which are zero-vectors $\mathbf{0}^{nk}$;
                \item The Merkle tree $\mathcal{T}$ built on top of $\mathbf{reg}[0][1], \ldots, \mathbf{reg}[N\hspace*{-1.5pt}-\hspace*{-1.5pt}1][1]$. (Note that $\mathcal{T}$ is an all-zero tree at this stage, but it will be modified when a new user joins the group, or when \textsf{GM} computes the updated group information.)
                \item Counter of registered users $c: = 0$.
            \end{itemize}
        Then, \textsf{GM} outputs $\mathsf{gpk} = (\mathsf{pp}, \mathsf{mpk}, \mathsf{tpk})$ and announces the initial group information $\mathsf{info}  = \emptyset$. He keeps $\mathcal{T}$ and $c$ for himself.

    \end{itemize}
    \item [$\ms{UKgen}(\ms{pp})$.] Each potential group user samples $\mathsf{usk} = \mb{x} \xleftarrow{\$} \{0,1\}^{m}$, and computes $\mathsf{upk} = \mb{p}=\text{bin}(\mb{A}\cdot\mb{x})\bmod q \in \{0,1\}^{nk}$.
        %\smallskip

        Without loss of generality, we assume that every honestly generated $\mathsf{upk}$ is a non-zero vector. (Otherwise, the user would either pick $\mathbf{x} = \mathbf{0}$ or accidentally find a solution to the $\mathsf{SIS}_{n,m,q,1}^\infty$ problem associated with matrix $\mathbf{A}$ - both happen with negligible probability.)
    \item [$\langle \mathsf{Join}, \mathsf{Issue}\rangle$.] If the user with key pair $(\mathsf{upk}, \mathsf{usk})= (\mathbf{p}, \mathbf{x})$ requests to join the group at epoch $\tau$, he sends $\mathbf{p}$ to \textsf{GM}. If the latter accepts the request, then the two parties proceed as follows.
        \begin{enumerate}
        \item \textsf{GM} issues a member identifier for the user as $\mathsf{uid} = \mathsf{bin}(c) \in \{0,1\}^\ell$. The user then sets his long-term signing key as $\mathsf{gsk}[c] = (\mathsf{bin}(c), \mathbf{p}, \mathbf{x})$.
            \item \textsf{GM} performs the following updates:
            \begin{itemize}
                \item Update $\mathcal{T}$ by running algorithm $\ms{TUpdate}_{\mathbf{A}}(\mathsf{bin}(c), \mathbf{p})$. \smallskip
                \item Register the user to table $\mathbf{reg}$ as  $\mathbf{reg}[c][1]:= \mathbf{p}$; \hspace*{2.5pt} $\mathbf{reg}[c][2]:= \tau$.
                \item Increase the counter $c: = c+1$. \smallskip
            \end{itemize}
        \end{enumerate}
    \item[$\mathsf{GUpdate}(\mathsf{gpk}, \mathsf{msk}, \mathsf{info}_{\tau_{\rm current}}, \mathcal{S}, \mathbf{reg})$.]
    This algorithm is run by \textsf{GM} to update the group information while also advancing the epoch. It works as follows.
    \begin{enumerate}
       % \item If $\mathsf{info}_{\tau_{\rm current}} = \mathsf{info}_0$, i.e., the algorithm is executed for the first time, then \textsf{GM} builds the initial Merkle tree on top of entries of table $\mathbf{reg}$, by running algorithm $\mathsf{TAcc}_{\mathbf{A}}(\mathbf{reg}[0], \ldots, \mathbf{reg}[N-1])$. The tree being built is denoted as $\mathcal{T}_{\tau_{\rm current}}$.
        %    \smallskip
        \item\label{step:update-2} Let the set $\mathcal{S}$ contain the public keys of registered users to be revoked. If $\mathcal{S} = \emptyset$, then go to Step~\ref{step:update-3}.

            Otherwise, $\mathcal{S} = \{\mathbf{reg}[i_1][1], \ldots, \mathbf{reg}[i_r][1]\}$, for some $r \in [1, N]$ and some $i_1, \ldots, i_r \in [0,N\hspace*{-1.5pt}-\hspace*{-1.5pt}1]$. Then, for all $t \in [r]$, \textsf{GM} runs $\ms{TUpdate}_{\mathbf{A}}(\mathsf{bin}(i_t), \mathbf{0}^{nk})$ to update the tree $\mathcal{T}$.

            \smallskip

       \item \label{step:update-3} At this point, by construction, each of the zero leaves in the tree $\mathcal{T}$ corresponds to either a revoked user or a potential user who has not yet registered. In other words,        only active users who are allowed to sign in the new epoch $\tau_{\rm new}$ have their non-zero public keys, denoted by $\{\mathbf{p}_j\}_j$, accumulated in the root $\mathbf{u}_{\tau_{\rm new}}$ of the updated tree.

           For each $j$, let $w_j \in \{0,1\}^\ell \times (\{0,1\}^{nk})^\ell$ be the witness for the fact that $\mathbf{p}_j$ is accumulated in $\mathbf{u}_{\tau_{\rm new}}$. Then \textsf{GM} publishes the group information of the new epoch as:
           \[
            \mathsf{info}_{\tau_{\rm new}} = \big(\mathbf{u}_{\tau_{\rm new}}, \{w_j\}_j\big).
           \]
    \end{enumerate}
    We remark that the $\mathsf{info}_{\tau}$ outputted at each epoch by \textsf{GM} is technically not part of the verification key. Indeed, as we will describe below, in order to verify signatures bound to epoch~$\tau$, the verifiers only need to download the first component $\mathbf{u}_\tau$ of size $\widetilde{\mathcal{O}}(\lambda)$ bits. Meanwhile, each active signer only has to download the respective witness of size $\widetilde{\mathcal{O}}(\lambda)\cdot \ell$. \medskip

  \item \hspace*{-2.5pt}$\mathsf{Sign}(\mathsf{gpk},\mathsf{gsk}[j], \mathsf{info}_\tau, M)$\textbf{.}~ To sign message $M$ using the group information $\mathsf{info}_\tau$ at epoch $\tau$, the user possessing $\mathsf{gsk}[j] = (\mathsf{bin}(j), \mathbf{p}, \mathbf{x})$ first checks if $\mathsf{info}_\tau$ includes a witness containing $\mathsf{bin}(j)$. If this is not the case, return $\bot$. Otherwise, the user downloads $\mathbf{u}_\tau$ and the witness of the form $\big(\mathsf{bin}(j), (\mathbf{w}_\ell, \ldots, \mathbf{w}_1)\big)$ from $\mathsf{info}_\tau$, and proceeds as follows.
      \begin{enumerate}
            \item Encrypt vector $\mathsf{bin}(j) \in \{0,1\}^\ell$ twice using Regev's encryption scheme. Namely, for each $i \in \{1,2\}$, sample $\mathbf{r}_i \xleftarrow{\$} \{0,1\}^{m_E}$ and compute \begin{eqnarray*}
			\mathbf{c}_i &=& (\mathbf{c}_{i,1}, \mathbf{c}_{i,2}) \\ &=& \Bigl( \mathbf{B}\cdot \mathbf{r}_i \bmod q, ~
				\mathbf{P}_i\cdot \mathbf{r}_i + \big\lfloor \frac{q}{2} \big\rfloor\cdot \mathsf{bin}(j) \bmod q \Bigr) \in \mathbb{Z}_q^{n} \times \mathbb{Z}_q^\ell.
				\end{eqnarray*}

            \item Generate a \textsf{NIZKAoK} $\Pi_{\mathsf{gs}}$ to demonstrate the possession of a valid tuple
            \begin{eqnarray}\label{equation:tuple-zeta}
            \zeta= (\mathbf{x}, \mathbf{p}, \mathsf{bin}(j), \mathbf{w}_\ell, \ldots, \mathbf{w}_1, \mathbf{r}_1, \mathbf{r}_2)
            \end{eqnarray}
            such that:

          \begin{enumerate}[(i)]
                \item \label{sign:condition-1} $\mathsf{TVerify}_{\mathbf{A}}\big( \mathbf{u}_\tau, \mathbf{p}, \big(\mathsf{bin}(j), (\mathbf{w}_\ell, \ldots, \mathbf{w}_1)\big) \big) =1$
                and  $\mathbf{A}\cdot \mathbf{x} = \mathbf{G}\cdot \mathbf{p} \bmod q$;
                \item \label{sign:condition-2}$\mathbf{c}_1$ and $\mathbf{c}_2$ are both correct encryptions of $\mathsf{bin}(j)$ with randomness $\mathbf{r}_1$ and $\mathbf{r}_2$, respectively; \smallskip
                 \item \label{sign:condition-3}$\mathbf{p} \neq \mathbf{0}^{nk}$.
          \end{enumerate}
           Note that statements \ref{sign:condition-1} and \ref{sign:condition-2} were covered by the LLNW protocol~\cite{LLNW16}. Meanwhile, statement \ref{sign:condition-3} is handled using the technique described at the beginning of this Section. We thus obtain a Stern-like interactive zero-knowledge argument system which is a slight modification of the one from~\cite{LLNW16}.
           %The detailed description of the protocol is presented in Appendix~\ref{appendix:ZK}.
           %This is done by running the protocol in Section~\ref{subsection:ZK-main}.
            The detailed description of the protocol is presented in Section~\ref{subsection:ZK-main}.
           The protocol is repeated $\kappa = \omega(\log \lambda)$ times to achieve negligible soundness error  and   made non-interactive via the Fiat-Shamir heuristic as a triple $\Pi_{\mathsf{gs}}= (\{\mathrm{CMT}_i\}_{i=1}^\kappa, \mathrm{CH}, \{\mathrm{RSP}\}_{i=1}^\kappa)$, where
          \[
            \mathrm{CH} = \mathcal{H}_{\mathsf{FS}}\big(M, (\{\mathrm{CMT}_i\}_{i=1}^\kappa, \mathbf{A}, \mathbf{u}_\tau, \mathbf{B}, \mathbf{P}_1, \mathbf{P}_2, \mathbf{c}_{1}, \mathbf{c}_{2}\big) \in \{1,2,3\}^\kappa.
          \]
        \item Output the group signature
        \begin{eqnarray}\label{equation:GS-output}
        \Sigma = (\Pi_{\mathsf{gs}}, \mathbf{c}_1, \mathbf{c}_2).
        \end{eqnarray}
      \end{enumerate}
   \item [$\ms{Verify}(\ms{gpk},\ms{info}_{\tau},M,\Sigma)$.] This algorithm proceeds as follows: \smallskip
  \begin{enumerate}
  \item Download $\mathbf{u}_\tau \in \{0,1\}^{nk}$ from $\mathsf{info}_\tau$.
  \item Parse $\Sigma$ as $\Sigma = \big(\{\mathrm{CMT}_i\}_{i=1}^\kappa, (Ch_1, \ldots, Ch_\kappa), \{\mathrm{RSP}\}_{i=1}^\kappa, \mathbf{c}_1, \mathbf{c}_2\big)$.

      If $(Ch_1, \ldots, Ch_\kappa) \neq \mathcal{H}_{\mathsf{FS}}\big(M, (\{\mathrm{CMT}_i\}_{i=1}^\kappa, \mathbf{A}, \mathbf{u}_\tau, \mathbf{B}, \mathbf{P}_1, \mathbf{P}_2, \mathbf{c}_{1}, \mathbf{c}_{2}\big)$, then return~$0$.
  \item For each $i \in [\kappa]$, run the verification phase of the protocol in Section~\ref{subsection:ZK-main} to check the validity of $\mathrm{RSP}_i$ with respect to $\mathrm{CMT}_i$ and $Ch_i$. If any of the conditions does not hold, then return~$0$.
      \item Return $1$. \smallskip
  \end{enumerate}
\item [$\mathsf{Trace}(\mathsf{gpk}, \mathsf{tsk}, \mathsf{info}_\tau, \mathbf{reg}, M, \Sigma$).]
This algorithm parses $\mathsf{tsk}$ as  $(\mathbf{S}_1, \mathbf{E}_1)$, parses $\Sigma$ as in~(\ref{equation:GS-output}), and performs the following steps.
    \begin{enumerate}
        \item Use $\mathbf{S}_1$ to decrypt $\mathbf{c}_1 = (\mathbf{c}_{1,1}, \mathbf{c}_{1,2})$ to
            obtain a string $\mathbf{b}' \in \{0,1\}^\ell$ (i.e., by computing $\big\lfloor\frac{(\mathbf{c}_{1,2} - \mathbf{S}_1^\top\cdot \mathbf{c}_{1,1})}{\lfloor q/2\rfloor}\big\rceil$). \smallskip
        \item If $\mathsf{info}_\tau$ does not include a witness containing $\mathbf{b}'$, then return $\bot$.
        \item Let $j' \in [0,N-1]$ be the integer having binary representation $\mathsf{b}'$. If the record $\mathbf{reg}[j'][1]$ in table $\mathbf{reg}$ is $\mathbf{0}^{nk}$, then return $\bot$.
        \item Generate a \textsf{NIZKAoK} $\Pi_{\mathsf{trace}}$ to demonstrate the possession of $\mathbf{S}_1 \in \mathbb{Z}^{n \times \ell}$, $\mathbf{E}_1 \in \mathbb{Z}^{\ell \times m_E}$, and $\mathbf{y} \in \mathbb{Z}^\ell$, such that:
            \begin{eqnarray}\label{equation:pi-trace}
            \begin{cases}
                \| \mathbf{S}_1 \|_\infty \leq \beta; \hspace*{2.8pt}  \| \mathbf{E}_1 \|_\infty \leq \beta; \hspace*{2.8pt} \| \mathbf{y}\|_\infty \leq  \lceil q/5 \rceil; \\
                \mathbf{S}_1^\top \cdot \mathbf{B} + \mathbf{E}_1 = \mathbf{P}_1 \bmod q; \\
                \mathbf{c}_{1,2} - \mathbf{S}_1^\top\cdot \mathbf{c}_{1,1} = \mathbf{y} + \lfloor q/2 \rfloor\cdot \mathbf{b}' \bmod q.
            \end{cases}
            \end{eqnarray}
      As the statement involves  modular linear equations with bounded-norm secrets, we can obtain a statistical zero-knowledge argument by employing the Stern-like interactive protocol from~\cite{LLMNW16-dgs}.
        %with public input $(\mathbf{B}, \mathbf{P}_1, \mathbf{c}_{1,1}, \mathbf{c}_{1,2}, \mathbf{b}')$ and prover's witness $(\mathbf{S}_1, \mathbf{E}_1, \mathbf{y})$.
        For completeness, the detailed description of the protocol is presented in Section~\ref{subsection:correct-decryption}.
       %This is done by running the protocol in Section~\ref{subsection:correct-decryption}.
        The protocol is repeated $\kappa = \omega(\log \lambda)$ times to achieve negligible soundness error  and   made non-interactive via the Fiat-Shamir heuristic as a triple $\Pi_{\mathsf{trace}}= (\{\mathrm{CMT}_i\}_{i=1}^\kappa, \mathrm{CH}, \{\mathrm{RSP}\}_{i=1}^\kappa)$, where
          \begin{eqnarray}\label{equation:pi-trace-FS}
            \mathrm{CH} = \mathcal{H}_{\mathsf{FS}}\big((\{\mathrm{CMT}_i\}_{i=1}^\kappa, \mathsf{gpk}, \mathsf{info}_\tau, M, \Sigma, \mathbf{b}'\big) \in \{1,2,3\}^\kappa.
          \end{eqnarray}
        \item Set $\mathsf{uid} = \mathbf{b}'$ and output $(\mathsf{uid}, \Pi_{\mathsf{trace}})$. \smallskip
    \end{enumerate}
\item [$\ms{Judge}(\ms{gpk}, \mathsf{uid}, \ms{info}_{\tau}, \Pi_{\mathsf{trace}}, M, \Sigma)$.] This algorithm consists of verifying the argument $\Pi_{\mathsf{trace}}$ w.r.t. common input $(\mathsf{gpk}, \mathsf{info}_\tau, M, \Sigma, \mathsf{uid})$, in a similar manner as in algorithm $\mathsf{Verify}$.

    If $\Pi_{\mathsf{trace}}$ does not verify, return $0$. Otherwise, return $1$.
\end{description}

\subsection{Analysis of the Scheme}\label{subsection:analysis-of-scheme}

\noindent
{\sc Efficiency. }
We first analyze the efficiency of the scheme described in Section~\ref{subsection:scheme-description}, with respect to security parameter $\lambda$ and parameter $\ell = \log N$.
\begin{itemize}
\item The public key $\mathsf{gpk}$ contains several matrices, and has bit-size $\widetilde{\mathcal{O}}(\lambda^2 + \lambda\cdot \ell)$.\smallskip
\item For each $j \in [0,N-1]$, the signing key $\mathsf{gsk}[j]$ has bit-size $\ell + nk + m = \widetilde{\mathcal{O}}(\lambda)+ \ell$.
\item At each epoch, the signature verifiers downloads $nk = \widetilde{\mathcal{O}}(\lambda)$ bits, while each active signer downloads $\widetilde{\mathcal{O}}(\lambda\hspace*{-2pt}\cdot\hspace*{-2pt}\ell)$ bits.\smallskip
\item The size of signature $\Sigma$ is dominated by that of the Stern-like \textsf{NIZKAoK} $\Pi_{\sf gs}$, which is $\mathcal{O}(|\zeta|\cdot \log q)\cdot \omega(\log \lambda)$, where $|\zeta|$ denotes the bit-size of the witness-tuple $\zeta$ in~(\ref{equation:tuple-zeta}). Overall, $\Sigma$ has bit-size $\widetilde{\mathcal{O}}(\lambda\hspace*{-2pt}\cdot\hspace*{-2pt}\ell)$.\smallskip
\item The Stern-like \textsf{NIZKAoK} $\Pi_{\sf trace}$ has bit-size $\widetilde{\mathcal{O}}(\ell^2 + \lambda\cdot \ell)$.
\end{itemize}

\smallskip

\noindent
{\sc Correctness. }
We now demonstrate that the scheme is correct with overwhelming probability, based on the perfect completeness of Stern-like protocols, and the correctness of Regev's encryption scheme.

First, note that a signature $\Sigma=(\Pi_{\ms{gs}},\mb{c}_1,\mb{c}_2)$ generated by an active and honest user $j$ is always accepted by algorithm $\ms{Verify}$. Indeed, such a user can always compute a tuple $\zeta=(\mb{x},\mb{p},\ms{bin}(j),\mb{w}_\ell,\ldots,\mb{w}_1,\mb{r}_1,\mb{r}_2)$ satisfying conditions \ref{sign:condition-1},\ref{sign:condition-2} and \ref{sign:condition-3} in the $\ms{Sign}$ algorithm. %Let the signature be $\Sigma=(\Pi_{\ms{gs}},\mb{c}_1,\mb{c}_2)$ and message be $M$ and user identity be $\ms{bin}(j)$. Since this user is honest and active at epoch $\tau$, then he possesses $\zeta=(\mb{x},\mb{p},\ms{bin}(j),\mb{w}_\ell,\ldots,\mb{w}_1,\mb{r}_1,\mb{r}_2)$ satisfying conditions \ref{sign:condition-1}, \ref{sign:condition-2},\ref{sign:condition-3} in the $\ms{Sign}$ algorithm.
%            \begin{enumerate}
%                \item $\mathsf{TVerify}_{\mathbf{A}}\big( \mathbf{u}_\tau, \mathbf{p}, \big(\mathsf{bin}(j), (\mathbf{w}_\ell, \ldots, \mathbf{w}_1)\big) \big) =1$
%                and  $\mathbf{A}\cdot \mathbf{x} = \mathbf{G}\cdot \mathbf{p} \bmod q$;  \smallskip
%                \item $\mathbf{p} \neq \mathbf{0}^{nk}$;        \smallskip
%                \item $\mathbf{c}_1$ and $\mathbf{c}_2$ are both correct encryptions of $\mathsf{bin}(j)$ with randomness $\mathbf{r}_1$ and $\mathbf{r}_2$, respectively.
%          \end{enumerate}
The completeness of the underlying argument system then guarantees that $\Sigma$ is always accepted by algorithm $\mathsf{Verify}$.

Next, we show that algorithm  $\ms{Trace}$ outputs $\ms{bin}(j)$ with overwhelming probability, and produces a proof $\Pi_{\ms{trace}}$ accepted by algorithm $\ms{Judge}$. Observe that, the decryption algorithm essentially computes
\begin{eqnarray*}
\mathbf{e}=\mathbf{c}_{1,2}-\mathbf{S}_1^T\mathbf{c}_{1,1}
%&=\mb{P}_1\cdot \mb{r}_1+\lceil q/2\rceil \mb{d}_i-\mb{S}_1^t\cdot \mb{B}\cdot \mb{r}_1\mod q\\
%&=(\mb{S}_1^t\cdot\mb{B}+\mb{E}_1)\cdot \mb{r}_1+\lceil q/2\rceil \mb{d}_i-\mb{S}_1^t\cdot \mb{B}\cdot \mb{r}_1\mod q\\
 =\mb{E}_1\cdot \mb{r}_1+\lfloor q/2\rfloor \cdot \mathsf{bin}(j)\mod q,
\end{eqnarray*} and sets the $j$-th bit of $\mb{b}'$ to be $0$ if $j$-th entry of $\mb{e}$ is closer to~$0$ than to $\lfloor q/2\rfloor$ and~$1$ otherwise. Note that our parameters are set so that $\|\mb{E}_1\cdot \mb{r}_1\|_{\infty}< q/5$, for $\mathbf{E}_1 \hookleftarrow \chi^{\ell \times m_E}$ and $\mathbf{r}_1 \xleftarrow{\$} \{0,1\}^{m_E}$. This ensures that $\mb{b'}=\ms{bin}(j)$ with overwhelming probability.

Further, as the user is active, $\ms{info}_{\tau}$ must contain $w=(\ms{bin}(j),\mb{w}_\ell,\ldots,\mb{w}_1)$ and $\mb{reg}[j][1]$ in table $\mb{reg}$ is not $ \mb{0}^{nk}$. Therefore, algorithm $\ms{Trace}$ will move to the $4$-th step, where it can always obtain the tuple $(\mb{S}_1,\mb{E}_1,\mb{y})$ satisfying the conditions in~(\ref{equation:pi-trace}).
By the completeness of the argument system, $\Pi_{\ms{trace}}$ will be accepted by the algorithm $\ms{Judge}$.

\smallskip

\noindent
{\sc Security. } In Theorem~\ref{theorem:main-scheme}, we prove that our scheme satisfies the security requirements of the Bootle et al.'s model~\cite{BCCGG16}.
\begin{theorem}\label{theorem:main-scheme}
Assume that the Stern-like argument systems used in Section~\ref{subsection:scheme-description} are simulation-sound. Then,
in the random oracle model, the given fully dynamic group signature satisfies the anonymity, traceability, non-frameability and tracing soundness requirements under the $\mathsf{LWE}_{n, q, \chi}$ and $\mathsf{SIS}_{n,m,q,1}^\infty$ assumptions.
\end{theorem}
In the random oracle model, the proof of Theorem~\ref{theorem:main-scheme} relies on the following facts:
\begin{enumerate}
\item The Stern-like zero-knowledge argument systems being used are simulation-sound;
\item The underlying encryption scheme, which is obtained from Regev cryptosystem~\cite{Regev05} via the Naor-Yung transformation~\cite{NY90}, is {IND}-{CCA2} secure;
\item The Merkle tree accumulator we employ is secure in the sense of Definition~\ref{definition:security-of-acc};
\item For a properly generated key-pair $(\mathbf{x}, \mathbf{p})$, it is infeasible to find $\mathbf{x}' \in \{0,1\}^m$ such that $\mathbf{x}' \neq \mathbf{x}$ and $\mathsf{bin}(\mathbf{A}\cdot \mathbf{x}' \bmod q) = \mathbf{p}$.
\end{enumerate}
%Details of the proof of Theorem~\ref{theorem:main-scheme} are deferred to Appendix~\ref{appendix:fdgs-security-proof}.
\begin{comment}
\section{Correctness and Security Requirements of Fully Dynamic Group Signatures}\label{appendix-security-fdgs-definition}
A fully dynamic group signature must satisfy \emph{correctness} and security requirements: \emph{anonymity}, \emph{non-frameability}, \emph{traceability} and \emph{tracing soundness}~\cite{BCCGG16}.

\input{detailed-definition-fdgs-model.tex}

\vspace{-0.3cm}
\end{comment}

%\section{Proof of Theorem~\ref{theorem:main-scheme} }\label{appendix:fdgs-security-proof}

The proof of Theorem~\ref{theorem:main-scheme} is established by Lemmas~\ref{lemma:anonymity}-\ref{lemma:tracing-soundness} given below.

\begin{lemma}\label{lemma:anonymity}
Assume that the $\mathsf{LWE}_{n, q, \chi}$ problem is hard. Then the given $\ms{FDGS}$ scheme provides anonymity in the random oracle model.
\end{lemma}
\begin{proof}
%Recall $\mb{Exp}_{\ms{FDGS},\mc{A}}^{\ms{anon}-0}(\lambda)$ generates the signature using $\ms{gsk}[i_0]$. Specifically, it means $\mb{c}_1\mb{c}_2$ encrypt $\ms{bin}(i_0)$.
 Let $\mc{A},\mc{B}$ be PPT algorithms, act as the adversary and challenger in the games respectively. %and $\mathcal{B}$ simulates all the oracles accessible to $\mc{A}$. %Let $(\ms{info}_{\tau},\ms{bin}(i_0),\ms{bin}(i_1),m)$ be the inputs of the $\ms{Chal}$ oracle.
 We construct a sequence of indistinguishable games, in which the first game is the experiment $\mb{Exp}_{\ms{FDGS},\mc{A}}^{\ms{anon}-0}(\lambda)$ and the last one is the experiment $\mb{Exp}_{\ms{FDGS},\mc{A}}^{\ms{anon}-1}(\lambda)$.  We will prove the lemma by demonstrating that any two consecutive games are indistinguishable.
 For each $i$, we denote by $W_i$ the output of the adversary in game $i$.

Specifically, we consider the following games.
\begin{description}
\item {\bf{Game}} $0$: This is exactly the experiment $\mb{Exp}_{\ms{FDGS},\mc{A}}^{\ms{anon}-0}(\lambda)$. \smallskip
\item{\bf{Game}} $1$: This game is the same as Game $0$ with only one modification: add $(\mb{S}_2,\mb{E}_2)$ to $\ms{tsk}$. This will make no difference to the view of the adversary. So $\text{Pr}[W_1=1]=\text{Pr}[\mb{Exp}_{\ms{FDGS},\mc{A}}^{\ms{anon}-0}(\lambda)=1]$.
%Instead of discarding secret keys $(\mb{S}_2,\mb{E}_2)$ after generating the tracing public key $\ms{tpk}$, the challenger will keep it. This will make no difference to the view of the adversary. So $\text{Pr}[W_1=1]=\text{Pr}[\mb{Exp}_{\ms{FDGS},\mc{A}}^{\ms{anon}-0}(\lambda)=1]$.
\smallskip

\item {\bf{Game}} $2$: This game is identical to Game~$1$ except that it generates a simulated proof for the oracle $\ms{Trace}$ by programming the random oracle $\mc{H}_{\ms{FS}}$ even though the challenger $\mc{B}$ has correct witnesses to generate a real proof. The view of the adversary, however, is statistically indistinguishable between Game $1$ and Game $2$ by statistical zero-knowledge property of our argument system generating $\Pi_{\ms{trace}}$. Therefore we have $\text{Pr}[W_1=1]\approx\text{Pr}[W_2=1]$.

\smallskip

\item {\bf{Game}} $3$:  This game uses $\mb{S}_2$ instead of $\mb{S}_1$ to simulate the oracle $\ms{Trace}$. The view of $\mc{A}$ is the same as in Game~$2$ until event $F_1$, where $\mc{A}$ queries the $\ms{trace}$ oracle a valid signature $(M,\Pi_{\ms{gs}},\mb{c}_1,\mb{c}_2,\ms{info}_\tau)$ with $\mb{c}_1,\mb{c}_2$ encrypting distinct bit strings, happens. Since $F_1$ breaks the soundness of our argument generating $\Pi_{\ms{gs}}$, we have $|\text{Pr}[W_2=1]-\text{Pr}[W_3=1]|\leq \text{Pr}[F_1]\leq \ms{Adv}_{\Pi_{\ms{gs}}}^{\ms{sound}}(\lambda)=\ms{negl}(\lambda)$.

\smallskip

\item{\bf{Game}} $4$: This game generates a simulated proof for the oracle $\ms{Chal}$. The view is indistinguishable to $\mc{A}$ by statistical zero-knowledge property of our argument system generating $\Pi_{\ms{gs}}$. Therefore we have $\text{Pr}[W_3=1]\approx\text{Pr}[W_4=1]$.
%This game is identical to Game~$3$ with one difference. In the $\ms{Chal}$ oracle it generates a simulated proof by programming the random oracle $\mc{H}_{\ms{FS}}$ while in Game~$3$ it generates a real proof. We say the view is indistinguishable to adversary by the zero-knowledge property of our argument system for generating $\Pi_{\ms{gs}}$. Therefore we have $\text{Pr}[W_3=1]\approx\text{Pr}[W_4=1]$.

\smallskip

\item{\bf{Game}} $5$: This is the same as Game~$4$ except that $\mb{c}_1$ is now encryption of $\ms{bin}(i_1)$ while $\mb{c}_2$ is still encryption of $\ms{bin}(i_0)$ for the $\ms{Chal}$ query. By the semantic security of our encryption scheme for public key $(\mb{B},\mb{P}_1)$ (which relies on $\ms{LWE}$ assumption), this change is negligible to the adversary. Recall that we use $\mb{S}_2$ for the $\ms{Trace}$ queries. So the change of $\mb{c}_1$ makes no difference to the view of the adversary. Therefore we have $|\text{Pr}[W_4=1]-\text{Pr}[W_5=1]|=\ms{negl}(\lambda)$.

\smallskip

\item{\bf{Game}} $6$: This game follows Game~$5$ with one change: it switches back to use $\mb{S}_1$ for the $\ms{Trace}$ queries and discards $(\mb{S}_2,\mb{E}_2)$. %That is to say, in Game~$5$ and Game~$6$, the adversary receives simulated proofs $\Pi_{\ms{gs}}$ and $\Pi_{\ms{trace}}$ with $\mb{c}_1,\mb{c}_2$ encrypting different user identities in the $\ms{Chal}$ oracle. The only difference is that whether we use $\mb{S}_1$ or $\mb{S}_2$ to open in $\ms{Trace}$ oracle.
    This modification is indistinguishable to the adversary until event $F_2$, where $\mc{A}$ queries the $\ms{trace}$ oracle a valid signature $(M,\Pi_{\ms{gs}},\mb{c}_1,\mb{c}_2,\ms{info}_\tau)$ with $\mb{c}_1,\mb{c}_2$ encrypting different bit strings, happens. However event $F_2$ breaks the simulation soundness of the argument system generating $\Pi_{\ms{gs}}$. Therefore we have $|\text{Pr}[W_5=~1]-\text{Pr}[W_6=~1]|\leq \ms{Adv}_{\Pi_{\ms{gs}}}^{\ms{ss}}(\lambda)=\ms{negl}(\lambda)$.

\smallskip

\item{\bf{Game}} $7$: In this game, it changes $\mb{c}_2$ to be encryption of $\ms{bin}(i_1)$ for the $\ms{Chal}$ query while the remaining parts are the same as in Game $6$. By the semantic security of our encryption scheme for public key $(\mb{B},\mb{P}_2)$, this change is negligible to the adversary. Note that we now use $\mb{S}_1$ for the $\ms{Trace}$ queries. So changing $\mb{c}_2$ makes no difference to the view of the adversary. Therefore we have $|\text{Pr}[W_6=1]-\text{Pr}[W_7=1]|=\ms{negl}(\lambda)$.

\smallskip

\item{\bf{Game}} $8$: We now generate real proof for the $\ms{Chal}$ query instead of simulated proof. By the statistical zero-knowledge property of our argument system generating $\Pi_{\ms{gs}}$, the view of the adversary in Game~$7$ and Game~$8$ are statistically indistinguishable, i.e., we have $\text{Pr}[W_7=~1]\approx\text{Pr}[W_8=~1]$.

\smallskip

\item{\bf{Game}} $9$: This game produces real proof, for the $\ms{Trace}$ queries. By the statistical zero-knowledge property of our argument system generating $\Pi_{\ms{trace}}$, we have Game~$8$ and Game~$9$  are statistically indistinguishable to the adversary, i.e., we have $\text{Pr}[W_8=1]\approx\text{Pr}[W_9=1]$. This is actually the experiment $\mb{Exp}_{\ms{FDGS},\mc{A}}^{\ms{anon}-1}(\lambda)$. Hence, we have  $\text{Pr}[W_9=1]=\text{Pr}[\mb{Exp}_{\ms{FDGS},\mc{A}}^{\ms{anon}-1}(\lambda)=1]$.
\end{description}
\smallskip
As a result, we have that  $$|\text{Pr}[\mb{Exp}_{\ms{FDGS},\mc{A}}^{\ms{anon}-1}(\lambda)=1]-\text{Pr}[\mb{Exp}_{\ms{FDGS},\mc{A}}^{\ms{anon}-0}(\lambda)=~1]| = \mathsf{negl}(\lambda),$$
and hence our scheme is anonymous. \qed
\end{proof}

\begin{lemma}\label{lemma:non-frame}
Assume that the $\mathsf{SIS}_{n,m,q,1}^\infty$ problem is hard. Then the given $\ms{FDGS}$ scheme provides non-frameability in the random oracle model.
\end{lemma}
\begin{proof}
We prove non-frameability by contradiction. Suppose that $\mc{A}$ succeeds with non-negligible advantage $\epsilon$. Then we can build a PPT algorithm $\mc{B}$ that either breaks the security of our accumulator or solves $\ms{SIS}_{n,q,m,1}^{\infty}$ problem, which are featured by $\mb{A}$, with also non-negligible probability.

\smallskip

Given a matrix $\mb{A}$ from the environment that $\mc{B}$ is in, %from an instance of $\ms{SIS}_{n,m,q,1}$ problem,
it first generates all parameters $\ms{pp}$ as we do in $\ms{GSetup}$, then invokes $\mc{A}$ with $\ms{pp}$, and then proceeds as described in the experiment. Here $\mc{B}$ can consistently answer all the oracle queries made by $\mc{A}$.
%$\mc{B}$ interacts with $\mc{A}$ and plays the role of the challenger in the experiment $\mb{Exp}_{\ms{FDGS},\mc{A}}^{\ms{non-frame}}(\lambda)$. In this experiment, $\mc{B}$ learns the keys of group manager ${\ms{GM}}$ and tracing manager $\ms{TM}$ and adds honest users to the group. Therefore $\mc{B}$ can consistently answer all the oracle queries made by $\mc{A}$.
When $\mc{A}$ outputs $(M^*,\Sigma^*,\ms{bin}(j^*),\Pi_{\ms{trace}}^*,\ms{info}_\tau)$ and wins, $\mc{B}$ does following. %
\begin{comment}
Recall that $\mc{A}$ succeeds if the following holds:

\begin{eqnarray}\label{equation:non-frame-condition}
\begin{cases}
     \ms{bin}(j^*)\in\ms{HUL}\setminus\ms{BUL},(\ms{bin}(j^*),M^*,\Sigma^*,\tau)\notin \ms{SL}.\\
     \ms{Verify}(\ms{gpk},\ms{info}_{\tau},M^*,\Sigma^*)=1.\\
    \ms{Judge}(\ms{gpk},\ms{bin}(j^*),\ms{info}_{\tau},\Pi_{\ms{trace}}^*,M^*,\Sigma^*)=1.\\
    \end{cases}
\end{eqnarray}
\end{comment}

Parse $\Sigma^*$ as $(\Pi_{\ms{gs}}^*,\mb{c}_1^*,\mb{c}_2^*)$, where $\Pi_{\ms{gs}}^*~=~(\{\ms{CMT}_i^*\}_{i=1}^{\kappa},\ms{CH}^*,\{\ms{RSP}_i^*\}_{i=1}^{\kappa})$ and $\ms{RSP}_i^*$ is a valid response w.r.t $\ms{CMT}_i^*$ and $\ms{CH}_i^*$ for $i\in[\kappa]$ since $\mathcal{A}$ wins.
%We may assume $(\mb{A},\mb{u}_{\tau},\mb{B},\mb{P}_1,\mb{P}_2,\mb{c}_1^*,\mb{c}_2^*)$ has correct witness $\zeta$ and $(\ms{gpk},\ms{info}_{\tau},M^*,\Sigma^*,\ms{bin}(j^*))\in L_{R_{\ms{trace}}}$ by soundness of our argument systems for relations $R_{\ms{gs}},R_{\ms{trace}}$.
We claim that $\mc{A}$ queried the tuple $(M^*,\{\ms{CMT}_i^*\}_{i=1}^{k}, \mb{A},\mb{u}_{\tau},\mb{B},\mb{P}_1,\mb{P}_2,\mb{c}_1^*,\mb{c}_2^*)$, denoted as $\xi^*$, to the hash oracle $\mc{H}_{\ms{FS}}$ with overwhelming probability. Otherwise guessing correctly this value occurs only with probability $3^{-\kappa}$, which is negligible. Therefore, with probability $\epsilon'=\epsilon-3^{-\kappa}$, the tuple $\xi^*$ has been an input of one hash query, denoted as $t^*\in\{1,2\cdots,Q_H\}$, where $Q_H$ is the total number of hash queries made by $\mc{A}$.
To employ the Forking lemma of Brickell et al.~\cite{Brickell00}, $\mc{B}$ runs polynomial-number executions of $\mc{A}$ exactly the same as in the original run until the $t^*$ hash query, that is, from $t^*$ hash query on, $\mc{B}$ will answer $\mc{A}$ with fresh and independent values for hash queries on each new execution. Note that the input of $t^*$ hash query must be $\xi^*$ as in the original run.
%Then $\mc{B}$ runs $32\cdot Q_H/\epsilon'$ additional executions of the adversary $\mc{A}$ with the same random tapes and random inputs as in the original run. To employ Forking lemma, we assume the adversary receives exactly the same answers until the $t^*$ hash query. Namely the first $t^*-1$ hash queries receive the same value as in the initial run. Therefore, the input of $t^*$ hash query must be $(M^*,\{\ms{CMT}_i^*\}_{i=1}^{k}, \mb{A},\mb{u}_{\tau},\mb{B},\mb{P}_1,\mb{P}_2,\mb{c}_1^*,\mb{c}_2^*)$ as in the original run. From this moment on, the adversary will receive fresh and independent values for hash queries for each new execution.
By the Forking Lemma~\cite{Brickell00}, with probability $\geq \frac{1}{2}$, $\mc{B}$ obtains $3$-fork involving the same tuple  $\xi^*$ with pairwise distinct hash value $\ms{CH}_{t^*}^{(1)}$,$\ms{CH}_{t^*}^{(2)}$,$\ms{CH}_{t^*}^{(3)}\in\{1,2,3\}^{\kappa}$. We have $(\ms{CH}_{t^*,j}^{(1)},\ms{CH}_{t^*,j}^{(2)},\ms{CH}_{t^*,j}^{(3)})=(1,2,3)$ with probability $1-(\frac{7}{9})^{\kappa}$ for some $j\in\{1,2,\ldots,\kappa\}$.
Then by the argument of knowledge of the system generating $\Pi_{\ms{gs}}$, from the valid responses $(\ms{RSP}_{t^*,j}^{(1)},\ms{RSP}_{t^*,j}^{(2)},\ms{RSP}_{t^*,j}^{(3)})$, we can extract the
witnesses $\zeta'=
(\mb{x}',\mb{p}',w_{\tau}',\mb{r}_1',\mb{r}_2'),
$
where~$$w_{\tau}'~=~(\ms{bin}(j'),\mb{w'}_{\ell,\tau},\ldots,\mb{w'}_{1,\tau})\in\{0,1\}^{\ell}\times(\{0,1\}^{nk})^{\ell},$$ such that we have:
\begin{enumerate}[(i)]
\item $\ms{TVerify}_{\mb{A}}(\mb{u}_{\tau},\mb{p'},w_{\tau}')=1,\mb{p'}\neq0$; \smallskip
\item $\mb{A}\cdot\mb{x}'=\mb{G}\cdot\mb{p}'\mod q$; \smallskip
\item  $\mb{c}_1^*,\mb{c}_2^*$  are encryptions of $\ms{bin}(j')$ with randomness $\mb{r}_1'$ and $\mb{r}_2'$.
\end{enumerate}
By the correctness of our encryption scheme, $\mb{c}_{1}^*$ is decrypted to $\ms{bin}(j')$. % and $\ms{bin}(j')=\ms{bin}(j^*)$ with overwhelming probability.
The fact $\mc{A}$ wins implies algorithm $\ms{Judge}$ outputs $1$. By the soundness of our system generating $\Pi_{\rm{trace}}$, $\ms{bin}(j^*)$ is decrypted from $\mb{c}_1^*$. Hence we have $\ms{bin}(j')=\ms{bin}(j^*)$ with overwhelming probability.
%By $(\ms{gpk},\ms{info}_{\tau},M^*,\Sigma^*,\ms{bin}(j^*))\in L_{R_{\ms{trace}}}$, we have $\mb{c}_{1}^*$ decrypted to $\ms{bin}(j^*)$. By the correctness of our encryption scheme, we have $\ms{bin}(j')=\ms{bin}(j^*)$ with overwhelming probability.

%\smallskip
Now we consider the following two cases:

\smallskip
\noindent{\bf{Case}} $1$: $\ms{bin}(j')$ is a inactive user at epoch $\tau$, i.e., the value at leaf $\ms{bin}(j')$ is $\mb{0}^{nk}$. This violates the accumulator security since we have $\ms{TVerify}_{\mb{A}}(\mb{u}_{\tau},\mb{p'},w_{\tau}')=~1$ with $\mb{p'}\neq \mb{0}$ at leaf $\ms{bin}(j')$. % this means at tau, we have p=0^nk at node bin(j'); but the adversary find a witness such that p not = 0 at node bin(j'), this violate the security of the accumulator.  $\ms{info}_\tau$ contains a witness $w_{\tau}'$, i.e.,

\smallskip
\noindent{\bf{Case}} $2$: $\ms{bin}(j^*)$ is an active user at epoch $\tau$. The fact $ \ms{bin}(j^*)\in\ms{HUL}\setminus\ms{BUL}$ indicates $\mc{A}$ does not know $\ms{gsk}[j^*]=(\ms{bin}(j^*),\mb{p'},\mb{x}_{j^*})$, where $\mb{x}_{j^*}$ was initially chosen by $\mc{B}$ and satisfies $\mb{A}\cdot\mb{x}_{j^*}=\mb{G}\cdot\mb{p}'\mod q$. Following the same argument of \cite{LLNW16}, we claim that $\mb{x}_{j^*}\neq \mb{x'}$ with probability at least $1/2$. Then we find a nonzero vector $\mb{z}=\mb{x}_{j^*}-\mb{x'}$ satisfying $\mb{A}\mb{z}=0\mod q$. Recall $\mb{x}_{j^*},\mb{x'}\in\{0,1\}^{m}$.  Thus $\mc{B}$ solves a $\ms{SIS}_{n,q,m,1}^{\infty}$ problem with non-negligible probability if $\mc{A}$ breaks the non-frameability of our construction with non-negligible probability.
\qed
\end{proof}

\begin{lemma}\label{lemma:traceability}
Assume that the $\mathsf{SIS}_{n,m,q,1}^\infty$ problem is hard. Then the given $\ms{FDGS}$ scheme satisfies traceability in the random oracle model.
\end{lemma}
\begin{proof}
Recall that adversary $\mc{A}$ wins traceability experiment if it generates a valid signature that either: (i) traces to an inactive user;  or (ii) traces to an active user but we cannot generate a proof accepted by the algorithm $\ms{Judge}$. We will prove that both cases occur with negligible probability. % and hence prove that $\mc{A}$ wins with negligible probability.

%Parse $\Sigma^*$ as $(\Pi_{\ms{gs}}^*,\mb{c}_1^*,\mb{c}_2^*)$, where $\Pi_{\ms{gs}}^*~=~(\{\ms{CMT}_i^*\}_{i=1}^{\kappa},\ms{CH}^*,\{\ms{RSP}_i^*\}_{i=1}^{\kappa})$ and $\ms{RSP}_i^*$ is a valid response w.r.t $\ms{CMT}_i^*$ and $\ms{CH}^*[i]$ for $i\in[\kappa]$ since $\mathcal{A}$ wins.

Let $(\ms{info}_{\tau},M,\Sigma)$ be the output of $\mc{A}$ in the experiment $\mb{Exp}_{\ms{FDGS},\mc{A}}^{\ms{trace}}(\lambda)$ and we compute $(\ms{bin}(j),\Pi_{\ms{trace}})$ by running the algorithm $\ms{Trace}$ . %(\ms{gpk},\ms{tsk},\ms{info}_{\tau},\mb{reg},M,\Sigma)$.
Parse $\Sigma$ as $(\Pi_{\ms{gs}},\mb{c}_1,\mb{c}_2)$, where $\Pi_{\ms{gs}}=(\{\ms{CMT}_i\}_{i=1}^{\kappa},\ms{CH},\{\ms{RSP}_i\}_{i=1}^{\kappa})$, and $\ms{RSP}_i$ is a valid response w.r.t. $\ms{CMT}_i$ and $\ms{CH}_i$ for $i\in[\kappa]$ since $(\ms{info}_{\tau},M,\Sigma)$ is a valid signature outputted by $\mathcal{A}$. Now we are in the same situation as in Lemma \ref{lemma:non-frame}. By the same technique, we can extract witnesses $\zeta=(\mb{x},\mb{p},w_{\tau},\mb{r}_1,\mb{r}_2)$, where $w_{\tau}=(\ms{bin}(j'),\mb{w}_{\ell,\tau},\ldots,\mb{w}_{1,\tau})\in\{0,1\}^{\ell}\times(\{0,1\}^{nk})^{\ell}$
satisfying $\ms{TVerify}_{\mb{A}}(\mb{u}_{\tau},\mb{p},w_{\tau})=1$, $\mb{p}\neq \mb{0}^{nk}$, and $\mb{c}_1,\mb{c}_2$ are encryptions of $\ms{bin}(j')$ using randomness $\mb{r}_1$ and $\mb{r}_2$. %Therefore case (i) happens with negligible probability.
Correctness of our encryption scheme implies $\mb{c}_{1}$ is decrypted to $\ms{bin}(j')$. Correct decryption of the $\ms{Trace}$ algorithm, which is run by the challenger, indicates that $\ms{bin}(j)$ is decrypted from $\mb{c}_1$. Hence we have $\ms{bin}(j')=\ms{bin}(j)$ with overwhelming probability.

Now, we note that case (i) considered above only happens with negligible probability. In fact, if $\ms{bin}(j)$ is not active at epoch $\tau$,  and $\ms{TVerify}_{\mb{A}}(\mb{u}_{\tau},\mb{p},w_{\tau})=1$ with $\mb{p}\neq\mb{0}$ at leaf $\ms{bin}(j)$, then the  security of our accumulator is violated.
Furthermore,  case (ii) also occurs with negligible probability since the challenger possesses valid witnesses to generate the proof $\Pi_{\rm{trace}}$, which will be accepted by the $\ms{Judge}$ algorithm, thanks to the completeness of the underlying argument system. It then follows that our scheme satisfies the traceability requirement.
\qed
\end{proof}

\begin{lemma}\label{lemma:tracing-soundness}
The given $\ms{FDGS}$ scheme satisfies tracing soundness in the random oracle model.
\end{lemma}

\begin{proof}
Let $(M,\Sigma,\ms{bin}(j_0),\Pi_{\ms{trace},0},\ms{bin}(j_1),\Pi_{\ms{trace},1},\ms{info}_{\tau})$ be the output of the adversary in the experiment $\mb{Exp}_{\ms{FDGS},\mcA}^{\ms{trace-sound}}(\lambda)$. Recall that adversary wins if the following conditions hold:
\begin{enumerate}[(i)]
    \item $\ms{bin}(j_0)\neq\ms{bin}(j_1)$;  \smallskip
    \item $\ms{Verify}(\ms{gpk},\ms{info}_{\tau},M,\Sigma)=1$;\smallskip
    \item $\ms{Judge}(\ms{gpk},\ms{bin}(j_b),\ms{info}_{\tau},\Pi_{\ms{trace},b},M,\Sigma)=1$ for $b\in\{0,1\}$;\smallskip
    \item $\ms{bin}(j_b)\neq\bot$ for $b\in\{0,1\}$.
\end{enumerate}
% is implied by the fact that the $\ms{Judge}$ algorithm outputs~$1$.

%The fact that the $\ms{Judge}$ algorithm outputs $1$ implies the verification algorithm of our argument system generating $\Pi_{\rm{trace}}$ outputs $1$  with public inputs $(\ms{gpk},\ms{info}_{\tau},M,\Sigma,\ms{bin}(j_b))$ w.r.t. proof $\Pi_{\ms{trace},b}$ for $b\in\{0,1\}$.
Let us parse $\Pi_{\ms{trace},b}$ as $(\{\ms{CMT}_{i,b}\}_{i=1}^{\kappa},\ms{CH}_b,\{\ms{RSP}_{i,b}\}_{i=1}^{\kappa})$. The fact that the $\ms{Judge}$ algorithm outputs $1$ implies that $\ms{RSP}_{i,b}$ is a valid response w.r.t. $\ms{CMT}_{i,b}$ and $\ms{CH}_{i,b}$ for $i\in[\kappa]$ and $b\in\{0,1\}$. By using the extraction technique for argument system generating $\Pi_{\ms{trace}}$ as in Lemma \ref{lemma:non-frame},
%Let $(\ms{info}_{\tau},M,\Sigma)$ be the output of $\mc{A}$ in the experiment $\mb{Exp}_{\ms{FDGS},\mc{A}}^{\ms{trace}}(\lambda)$ and compute  $(\ms{bin}(j),\Pi_{\ms{trace}})$ by running the $\ms{Trace}$ algortihm. %(\ms{gpk},\ms{tsk},\ms{info}_{\tau},\mb{reg},M,\Sigma)$.
%Parse $\Sigma$ as $(\Pi_{\ms{gs}},\mb{c}_1,\mb{c}_2)$, where $\Pi_{\ms{gs}}=(\{\ms{CMT}_i\}_{i=1}^{\kappa},\ms{CH},\{\ms{RSP}_i\}_{i=1}^{\kappa})$ and $\ms{RSP}_i$ is a valid response w.r.t $\ms{CMT}_i$ and $\ms{CH}[i]$ for $i\in[\kappa]$ since $(\ms{info}_{\tau},M,\Sigma)$ is a valid signature outputed by $\mathcal{A}$. Now we are in the same situation as in Lemma \ref{lemma:non-frame}.
%Then by soundness of the argument system, we can assume $(\ms{gpk},\ms{info}_{\tau},M,\Sigma,\ms{bin}(j_b))\in L_{R_{\ms{trace}}}$ for $b\in\{0,1\}$.
%Following the same argument as in the lemma \ref{lemma:non-frame},
we can extract the witnesses $\mb{S}_{1,b},\mb{E}_{1,b},\mb{y}_{b}$ for $b\in\{0,1\}$ satisfying:
% equations in (\ref{equation:pi-trace}). %i.e.,
               \begin{eqnarray}\label{trace-sound-extract-witness}
               \begin{cases}
                \| \mathbf{S}_{1,b} \|_\infty \leq \beta; \hspace*{2.8pt}  \| \mathbf{E}_{1,b} \|_\infty \leq \beta; \hspace*{2.8pt} \| \mathbf{y}_b\|_\infty \leq  \lceil q/5 \rceil; \\
                \mathbf{S}_{1,b}^\top \cdot \mathbf{B} + \mathbf{E}_{1,b} = \mathbf{P}_1 \bmod q; \\
                \mathbf{c}_{1,2} - \mathbf{S}_{1,b}^\top\cdot \mathbf{c}_{1,1} = \mathbf{y}_b + \lfloor q/2 \rfloor\cdot \ms{bin}(j_b) \bmod q.
           \end{cases}
          \end{eqnarray}
%We show that $\ms{bin}(j_0)=\ms{bin}(j_1)$ with overwhelming probability. %Subtract last equation of (\ref{trace-sound-extract-witness}) for case $b=0$ from that for case $b=1$,
Note that, we have %\begin{equation}\label{trace-sound-compare}
$(\mathbf{S}_{1,0}^\top-\mathbf{S}_{1,1}^\top)\cdot \mathbf{c}_{1,1} = \mathbf{y}_1-\mb{y}_0 + \lfloor q/2 \rfloor\cdot (\ms{bin}(j_1)-\ms{bin}(j_0)) \bmod q.$
%\end{equation}
Suppose that  $\ms{bin}(j_0)\neq\ms{bin}(j_1)$, then $\|\lfloor q/2 \rfloor\cdot(\ms{bin}(j_1)-\ms{bin}(j_0))\|_{\infty}= \lfloor q/2 \rfloor$. On the other hand, we have $ \| \mathbf{y}_1-\mathbf{y}_0\|_\infty \leq  2\cdot\lceil  q/5 \rceil$.
%by the three inequalities of (\ref{trace-sound-extract-witness}).
It then follows that $$\|\mathbf{y}_1-\mb{y}_0 + \lfloor q/2 \rfloor\cdot (\ms{bin}(j_1)-\ms{bin}(j_0))\|_{\infty}> 0,$$
which implies that
$\mathbf{S}_{1,0}^\top\neq\mathbf{S}_{1,1}^\top$. In other words,
we can obtain two  solutions for the equation $\mathbf{S}^\top \cdot \mathbf{B} + \mathbf{E} = \mathbf{P}_1 \bmod q$, which contradicts the fact that there exists at most one solution for $\ms{LWE}$ samples $(\mb{B},\mb{P}_1)$.

By contradiction, we have $\ms{bin}(j_0)=\ms{bin}(j_1)$ with overwhelming probability. This implies that case (i) does not hold with overwhelming probability once cases (ii),(iii) and (iv) hold. As a result, the advantage of the adversary attacking the tracing soundness of our scheme is negligible (in $\lambda$). In other words, our construction satisfies tracing soundness.
\qed
\end{proof}

\subsection{Details of the Main Zero-Knowledge Argument System}\label{subsection:ZK-main}
Our protocol is a modification of the one from~\cite{LLNW16}, in which we additionally prove that $\mathbf{p} \neq \mathbf{0}^{nk}$ using the technique discussed at the beginning of Section~\ref{section:main-scheme}.
We will employ a new strategy for Stern-like protocols, suggested in~\cite{LLMNW16-dgs}, which consists of unifying all the equations to be proved into just one equation of the form $\mathbf{M}\cdot \mathbf{z} = \mathbf{u} \bmod q$, for some public matrix $\mathbf{M}$ and vector $\mathbf{u}$ over $\mathbb{Z}_q$. This strategy will help us to deliver a simpler presentation than in~\cite{LLNW16}.

Let us first examine the equations associated with the execution of algorithm $\mathsf{TVerify}_{\mathbf{A}}\big( \mathbf{u}_\tau, \mathbf{p}, ((j_1, \ldots, j_\ell), (\mathbf{w}_\ell, \ldots, \mathbf{w}_1)) \big)$. This algorithm computes the tree path $\mathbf{v}_\ell, \mathbf{v}_{\ell-1}, \ldots, \mathbf{v}_1, \mathbf{v}_0 \in \{0,1\}^{nk}$ as follows: $\mathbf{v}_\ell= \mathbf{p}$ and
\begin{eqnarray}\label{equation:main-protocol-tree-path}
  \forall \hspace*{1pt}i \in  \{\ell-1, \ldots, 1, 0\}: \hspace*{2pt}
    \mathbf{v}_i = \begin{cases}
                      h_{\mathbf{A}}(\mathbf{v}_{i+1}, \mathbf{w}_{i+1}), \text{ if } j_{i+1}=0; \\
                      h_{\mathbf{A}}(\mathbf{w}_{i+1}, \mathbf{v}_{i+1}), \text{ if } j_{i+1}=1.
                 \end{cases}
  \end{eqnarray}
The algorithm outputs $1$ if $\mathbf{v}_0 = \mathbf{u}_\tau$.
Observe that relation~(\ref{equation:main-protocol-tree-path}) can be equivalently rewritten as: $\forall \hspace*{1pt}i \in  \{\ell-1, \ldots, 1, 0\},$
\begin{eqnarray*}
 &&  \bar{j}_{i+1}\cdot h_{\mathbf{A}}(\mathbf{v}_{i+1}, \mathbf{w}_{i+1}) + j_{i+1} \cdot h_{\mathbf{A}}(\mathbf{w}_{i+1}, \mathbf{v}_{i+1}) = \mathbf{v}_i \\[2.5pt]
&\Leftrightarrow&\bar{j}_{i+1} \hspace*{-1.5pt}\cdot\hspace*{-1.5pt} \big(\mathbf{A}_0 \cdot \mathbf{v}_{i+1} + \mathbf{A}_1 \cdot \mathbf{w}_{i+1}\big) + j_{i+1}\hspace*{-1.5pt}\cdot\hspace*{-1.5pt} \big(\mathbf{A}_0 \cdot \mathbf{w}_{i+1} + \mathbf{A}_1\cdot \mathbf{v}_{i+1}\big) = \mathbf{G}\cdot \mathbf{v}_{i} \bmod q\\
&\Leftrightarrow& \mathbf{A}\cdot  \left(\begin{array}{c}
                               \bar{j}_{i+1} \cdot\mathbf{v}_{i+1} \\
                                j_{i+1} \cdot \mathbf{v}_{i+1}
                             \end{array}\right) + \mathbf{A} \cdot \left(\begin{array}{c}
                               {j}_{i+1}\cdot \mathbf{w}_{i+1} \\
                                \bar{j}_{i+1} \cdot\mathbf{w}_{i+1}
                             \end{array}\right) = \mathbf{G}\cdot\mathbf{v}_{i} \bmod q \\
                             &\Leftrightarrow& \mathbf{A}\cdot \mathsf{ext}(j_{i+1}, \mathbf{v}_{i+1}) + \mathbf{A}\cdot \mathsf{ext}(\bar{j}_{i+1}, \mathbf{w}_{i+1}) = \mathbf{G}\cdot\mathbf{v}_{i} \bmod q.
\end{eqnarray*}

In the above, we use the notation $\mathsf{ext}(b, \mathbf{v})$, for bit $b$ and vector $\mathbf{v}$, to denote the vector
$\left(
\begin{array}{c}
\bar{b}\cdot \mathbf{v} \\
b\cdot \mathbf{v} \\
\end{array}
\right)$. Later on, we will also use the notation $\mathsf{ext}_2(b)$, for bit $b$, to denote the $2$-dimensional vector  $(\bar{b}, b)^\top$. \smallskip

Given the above discussion, our protocol can be summarized as follows.
\begin{description}
  \item[Public input:] $\mathbf{A}, \hspace*{2.8pt}\mathbf{G}, \hspace*{2.8pt}\mathbf{u}_\tau, \hspace*{2.8pt}\mathbf{B}, \hspace*{2.8pt} \mathbf{P}_1, \hspace*{2.8pt} \mathbf{P}_2, \hspace*{2.8pt}\mathbf{c}_{1} = (\mathbf{c}_{1,1}, \mathbf{c}_{1,2}), \hspace*{2.8pt}\mathbf{c}_{2} = (\mathbf{c}_{2,1}, \mathbf{c}_{2,2})$. \medskip

  \item[Prover's goal:] Proving knowledge of secret $\mathbf{x} \in \{0,1\}^m$, \hspace*{2.8pt}$\mathbf{0} \neq \mathbf{p} \in \{0,1\}^{nk}$, \hspace*{2.8pt}$j_1, \ldots, j_\ell \in \{0,1\}$, \hspace*{2.8pt}$\mathbf{v}_{1}, \ldots, \mathbf{v}_{\ell-1}, \mathbf{w}_1, \ldots, \mathbf{w}_\ell \in \{0,1\}^{nk}$, \hspace*{2.8pt}$\mathbf{r}_1, \mathbf{r}_2 \in \{0,1\}^{m_E}$, such that the following equations hold:
\end{description}
\vspace*{-0.35cm}
\begin{eqnarray}\label{equation:ZK-main-protocol-original}
\begin{cases}
    \mathbf{A}\cdot \mathsf{ext}(j_{1}, \mathbf{v}_{1}) + \mathbf{A}\cdot \mathsf{ext}(\bar{j}_{1}, \mathbf{w}_{1}) = \mathbf{G}\cdot\mathbf{u}_{\tau} \bmod q; \\
    \mathbf{A}\cdot \mathsf{ext}(j_{2}, \mathbf{v}_{2}) + \mathbf{A}\cdot \mathsf{ext}(\bar{j}_{2}, \mathbf{w}_{2}) - \mathbf{G}\cdot\mathbf{v}_{1} = \mathbf{0} \bmod q; \\[1.5pt]
    \hspace*{40pt} \ldots \hspace*{5pt} \ldots \hspace*{5pt} \ldots \\[1.5pt]
    \mathbf{A}\cdot \mathsf{ext}(j_{\ell}, \mathbf{p}) + \mathbf{A}\cdot \mathsf{ext}(\bar{j}_{\ell}, \mathbf{w}_{\ell}) -  \mathbf{G}\cdot\mathbf{v}_{\ell-1} = \mathbf{0} \bmod q; \\
    \mathbf{A}\cdot \mathbf{x} - \mathbf{G}\cdot \mathbf{p} = \mathbf{0}\bmod q; \\
    \mathbf{B}\cdot \mathbf{r}_b = \mathbf{c}_{b,1} \bmod q,  \text{ for }b\in\{1,2\}; \\
\mathbf{P}_b \cdot \mathbf{r}_b + \big\lfloor \frac{q}{2}\big \rceil \cdot (j_1, \ldots, j_\ell)^\top = \mathbf{c}_{b,2} \bmod q, \mbox{ for } b\in\{1,2\}. \\
%\mathbf{B}\cdot \mathbf{r}_2 = \mathbf{c}_{2,1} \bmod q;  \\
 %\mathbf{P}_2 \cdot \mathbf{r}_2 + \big\lfloor \frac{q}{2}\big \rceil \cdot (j_1, \ldots, j_\ell)^\top = \mathbf{c}_{2,2}\bmod q.
\end{cases}
\end{eqnarray}

Now, we employ extending-then-permuting techniques for Stern-like protocols~\cite{Ste96,LNSW13,LLNW16,LLMNW16-dgs} to handle the secret objects, as follows. For any positive integer~$i$, we let~$\mathcal{S}_i$ denote the set of all permutations of~$i$ elements.

\begin{itemize}
    \item To prove that $\mathbf{p} \in \{0,1\}^{nk}$ and $\mathbf{p} \neq \mathsf{0}^{nk}$, we append $nk-1$ ``dummy'' entries to $\mathbf{p}$ to get $\mathbf{p}^* \in \mathsf{B}_{nk}^{2nk-1}$. Note that, for any $\pi_p \in \mathcal{S}_{2nk-1}$, we have: $$\mathbf{p}^* \in \mathsf{B}_{nk}^{2nk-1} \Leftrightarrow \pi_p(\mathbf{p}^*) \in \mathsf{B}_{nk}^{2nk-1}.\vspace*{-0.1cm}$$
    \item To prove that $\mathbf{x} \in \{0,1\}^m$, we append $m$ ``dummy'' entries to $\mathbf{x}$ to get $\mathbf{x}^* \in \mathsf{B}_{m}^{2m}$. Note that, for any $\pi_x \in \mathcal{S}_{2m}$, we have: $\mathbf{x}^* \in \mathsf{B}_m^{2m} \Leftrightarrow \pi_x(\mathbf{x}^*) \in \mathsf{B}_m^{2m}$.
    \item Similarly, we extend $\mathbf{v}_1, \ldots, \mathbf{v}_{\ell-1}, \mathbf{w}_1, \ldots, \mathbf{w}_\ell$ to $\mathbf{v}^*_1, \ldots, \mathbf{v}^*_{\ell-1}, \mathbf{w}^*_1, \ldots, \mathbf{w}^*_\ell \in \mathsf{B}_{nk}^{2nk}$, respectively. We also extend
      $\mathbf{r}_1, \mathbf{r}_2$ to $\mathbf{r}_1^*, \mathbf{r}_2^* \in \mathsf{B}_{m_E}^{2m_E}$, respectively.
    \item For each $i =1, \ldots, \ell$, to prove that the bit $j_i\in \{0,1\}$ is involved in both encryption layer and Merkle tree layer, we extend it to $\mathbf{j}_i = \mathsf{ext}_2(j_i)$. Then we use the following permuting technique.

        For bit $b \in \{0,1\}$, we let $T_b$ be the permutation that transforms vector $\mathbf{t} = (t_0, t_1)^\top \in \mathbb{Z}^2$ into vector $T_b(\mathbf{t}) = (t_b, t_{\overline{b}})$. Note that, for any $b_i \in \{0,1\}$, we have: $\mathbf{j}_i = \mathsf{ext}_2(j_i) \Leftrightarrow T_{b_i}(\mathbf{j}_i) = \mathsf{ext}_2(j_i \oplus b_i)$, where $\oplus$ denotes the addition modulo $2$. Then, to prove knowledge of $j_i \in \{0,1\}$, we instead prove knowledge of $\mathbf{j}_i = \mathsf{ext}_2(j_i)$, by sampling $b_i \xleftarrow{\$} \{0,1\}$ and showing the verifier that the right-hand side of the equivalence holds.
        Here, the crucial point is that $b_i$ acts as a ``one-time pad'' that perfectly hides $j_i$. Furthermore, to prove that $j_i$ appears in a different layer of the system, we will set up a similar equivalence in that layer, and use the same $b_i$ at both places.
    \item Now, we have to handle the extended vectors appearing in the context after the above extensions. They are
    $\widehat{\mathbf{v}}_1 = \mathsf{ext}(j_{1}, \mathbf{v}^*_{1}), \ldots, \widehat{\mathbf{v}}_{\ell-1}= \mathsf{ext}(j_{\ell-1}, \mathbf{v}^*_{\ell-1}) \in \{0,1\}^{4nk}$ and $\widehat{\mathbf{w}}_1= \mathsf{ext}(\bar{j}_{1}, \mathbf{w}^*_{1}), \ldots, \widehat{\mathbf{w}}_\ell= \mathsf{ext}(\bar{j}_{\ell}, \mathbf{w}^*_{\ell}) \in \{0,1\}^{4nk}$, as well as $\widehat{\mathbf{p}}= \mathsf{ext}(j_\ell, \mathbf{p}^*) \in \{0,1\}^{4nk-2}$. To prove that these vector are well-formed, we will use the following permuting techniques.

    Let $r \in \{2nk, 2nk-1\}$. For $b \in \{0,1\}$ and $\pi \in \mathcal{S}_{r}$, we define the permutation $F_{b, \pi}$ that transforms  $\mathbf{t} = \left(
                    \begin{array}{c}
                      \mathbf{t}_0 \\
                      \mathbf{t}_1 \\
                    \end{array}
                  \right)
     \in \mathbb{Z}^{2r}$ consisting of $2$ blocks of size $r$ into $F_{b, \pi}(\mathbf{t}) = \left(
                                                                                                         \begin{array}{c}
                                                                                                           \pi(\mathbf{t}_b) \\
                                                                                                           \pi(\mathbf{t}_{\bar{b}}) \\
                                                                                                         \end{array}
                                                                                                       \right)
     $. Namely, $F_{b,\pi}$ first rearranges the blocks of $\mathbf{t}$ according to $b$ (it keeps the arrangement of blocks if $b=0$, or swaps them if $b=1$), then it permutes each block according to $\pi$. Note that, we have the following, for all $b_1, \ldots, b_\ell \in \{0,1\}$, $\phi_{v,1}, \ldots, \phi_{v, \ell-1}, \phi_{w,1}, \ldots, \phi_{w, \ell} \in \mathcal{S}_{2nk}$, and $\pi_p \in \mathcal{S}_{2nk-1}$:
     \[
     \begin{cases}
     \forall i\in [\ell-1]:  \widehat{\mathbf{v}}_i = \mathsf{ext}(j_{i}, \mathbf{v}^*_{i}) \hspace*{5pt} \Longleftrightarrow \hspace*{5pt} F_{b_i,\phi_{v,i}}( \widehat{\mathbf{v}}_{i}) =  \mathsf{ext}(j_i \oplus b_i, \phi_{v,i}(\mathbf{v}^*_i)); \\
     \forall i \in [\ell]:  \widehat{\mathbf{w}}_i= \mathsf{ext}(\bar{j}_{i}, \mathbf{w}^*_{i})\hspace*{5pt}\Longleftrightarrow \hspace*{5pt} F_{b_i,\phi_{w,i}}(\widehat{\mathbf{w}}_i) =  \mathsf{ext}(\overline{j_i \oplus b_i}, \phi_{w,i}(\mathbf{w}^*_i)); \\
      \widehat{\mathbf{p}}= \mathsf{ext}(j_\ell, \mathbf{p}^*) \hspace*{5pt}\Longleftrightarrow \hspace*{5pt} F_{b_\ell,\pi_p}(\widehat{\mathbf{p}}) =   \mathsf{ext}(j_\ell \oplus b_\ell, \pi_p(\mathbf{p}^*)).
     \end{cases}
     \]
   Combining these permuting techniques with the ones described earlier, we can prove that the considered extended vectors are well-formed, with respect to the secret $j_1, \ldots, j_\ell, \mathbf{v}^*_1, \ldots, \mathbf{v}_{\ell-1}^*, \mathbf{w}_1^*, \ldots, \mathbf{w}_\ell^*, \mathbf{p}^*$.

\end{itemize}
Given the above transformations, we now unify all the secret objects into vector $\mathbf{z} \in \{0,1\}^{D}$, where $D=10nk\ell+2m+ 4m_E+ 2\ell-3$, of the following form (we abuse the notation of transposition):
\begin{eqnarray}\label{equation:main-protocol-vector-z}
%\nonumber&&
\hspace*{-12pt}\big(\hspace*{1pt}
\mathbf{v}_1^* \hspace*{1.5pt}\|\hspace*{1pt} \widehat{\mathbf{v}}_1\hspace*{1.5pt}\|\hspace*{1.5pt} \widehat{\mathbf{w}}_1 \hspace*{1.5pt}\| \ldots \|\hspace*{1pt} \mathbf{v}_{\ell-1}^* \hspace*{1pt}\|\hspace*{1.5pt} \widehat{\mathbf{v}}_{\ell-1} \hspace*{1pt}\|\hspace*{1pt} \widehat{\mathbf{w}}_{\ell-1} \hspace*{1.5pt}\|\hspace*{1.5pt}
 \mathbf{p}^* \hspace*{1.5pt}\|\hspace*{1.5pt} \widehat{\mathbf{p}} \hspace*{1.5pt}\| \hspace*{1pt}\widehat{\mathbf{w}}_{\ell} \hspace*{1pt}\| \hspace*{1.5pt}\mathbf{x}^* \hspace*{1.5pt}\| \hspace*{1pt}\mathbf{r}_1^* \hspace*{1.5pt}\|\hspace*{1pt} \mathbf{r}_2^* \hspace*{1.5pt}\|\hspace*{1pt} \mathbf{j}_1 \hspace*{1pt} \| \ldots \| \hspace*{1pt}\mathbf{j}_\ell
\big).
\end{eqnarray}
Then, we observe that, the equations in~(\ref{equation:ZK-main-protocol-original}) can be equivalently combined into one equation $\mathbf{M}\cdot \mathbf{z} = \mathbf{u} \bmod q$, where matrix $\mathbf{M}$ and vector $\mathbf{u}$ are built from the public input.

Next, to apply the permuting techniques we have discussed above, let us define $\mathsf{VALID}$ as the set of all vectors in $\{0,1\}^D$, that have the form~(\ref{equation:main-protocol-vector-z}), where $\mathbf{x}^* \in \mathsf{B}_m^{2m}$, $\mathbf{r}_1^*, \mathbf{r}_2^* \in \mathsf{B}_{m_E}^{2m_E}$, $\mathbf{p}^* \in \mathsf{B}_{nk}^{2nk-1}$, $\mathbf{v}^*_1, \ldots, \mathbf{v}^*_{\ell-1} \in \mathsf{B}_{nk}^{2nk}$, and there exist $j_1, \ldots, j_\ell \in \{0,1\}$, $\mathbf{w}^*_1, \ldots, \mathbf{w}^*_\ell \in \mathsf{B}_{nk}^{2nk}$, such that:
\[
\begin{cases}
\widehat{\mathbf{p}}= \mathsf{ext}(j_\ell, \mathbf{p}^*), \hspace*{6.8pt} \forall i \in [\ell]: \mathbf{j}_i = \mathsf{ext}_2(j_i), \\
\forall i \in [\ell-1]: \widehat{\mathbf{v}}_i = \mathsf{ext}(j_{i}, \mathbf{v}^*_{i}), \hspace*{6.8pt} \forall i \in [\ell]: \widehat{\mathbf{w}}_i= \mathsf{ext}(\bar{j}_{i}, \mathbf{w}^*_{i}).
\end{cases}
\]
It can be seen that our vector $\mathbf{z}$ belongs to this tailored set $\mathsf{VALID}$. To prove that $\mathbf{z} \in \mathsf{VALID}$ using random permutations, let us determine how to permute the coordinates of $\mathbf{z}$. To this end, we first define the set:
\[
\overline{\mathcal{S}} = \{0,1\}^\ell \times \mathcal{S}_{2m} \times \mathcal{S}_{2nk-1} \times (\mathcal{S}_{2m_E})^2 \times (\mathcal{S}_{2nk})^{2\ell-1}.
\]
Then, $\forall \hspace*{1pt}\eta \hspace*{-1.5pt}= \hspace*{-1.5pt}\big((b_1, \ldots, b_\ell), \pi_x, \pi_p, (\pi_{r,1}, \pi_{r,2}), (\phi_{v,1}, \ldots, \phi_{v,\ell-1}, \phi_{w,1}, \ldots, \phi_{w, \ell})\big) \hspace*{-1.5pt}\in\hspace*{-1.5pt} \overline{\mathcal{S}}$, we let $\Gamma_\eta$ be the permutation that, when applying to  vector $\mathbf{t} \in \mathbb{Z}^D$ whose blocks are as in~(\ref{equation:main-protocol-vector-z}), it transforms those blocks as follows:
\begin{itemize}
    \item For all $i \in [\ell-1]$: \hspace*{5pt} $\mathbf{v}_i^* \mapsto  \phi_{v,i}(\mathbf{v}_i^*)$; \hspace*{5pt} $\widehat{\mathbf{v}}_i \mapsto F_{b_i,\phi_{v,i}}( \widehat{\mathbf{v}}_{i})$.  \smallskip
    \item For all $i \in [\ell]$: \hspace*{5pt} $\widehat{\mathbf{w}}_i \mapsto F_{b_i,\phi_{w,i}}( \widehat{\mathbf{w}}_{i})$; \hspace*{5pt} $\mathbf{j}_i \mapsto T_{b_i}(\mathbf{j}_i)$. \smallskip
    \item $\mathbf{p}^* \mapsto \pi_p(\mathbf{p}^*)$; \hspace*{5pt} $\widehat{\mathbf{b}} \mapsto F_{b_\ell,\pi_p}(\widehat{\mathbf{p}})$. \smallskip
    \item $\mathbf{x}^* \mapsto \pi_x(\mathbf{x}^*)$; \hspace*{5pt} $\mathbf{r}_1^* \mapsto \pi_{r,1}(\mathbf{r}_1^*)$; \hspace*{5pt} $\mathbf{r}_2^* \mapsto \pi_{r,2}(\mathbf{r}_2^*)$.
\end{itemize}
It now can be checked that we have the desired equivalence: For all $\eta \in \overline{\mathcal{S}}$,
\[
\mathbf{z} \in \mathsf{VALID} \hspace*{5pt} \Longleftrightarrow \hspace*{5pt} \Gamma_\eta(\mathbf{z}) \in \mathsf{VALID}.
\]

At this point, we can now run a simple Stern-like protocol to prove knowledge of $\mathbf{z} \in \mathsf{VALID}$ such that $\mathbf{M}\cdot \mathbf{z} = \mathbf{u} \bmod q$. The common input is the pair $(\mathbf{M}, \mathbf{u})$, while the prover's secret input is $\mathbf{z}$. The interaction between prover $\mathcal{P}$ and verifier $\mathcal{V}$ is described in Fig.~\ref{Figure:Interactive-Protocol}. The protocol employs the string commitment scheme \textsf{COM} from~\cite{KTX08} that is statistically hiding and computationally binding if the $\mathsf{SIVP}_{\widetilde{\mathcal{O}}(n)}$ problem is hard.

\begin{figure}[!htbp]
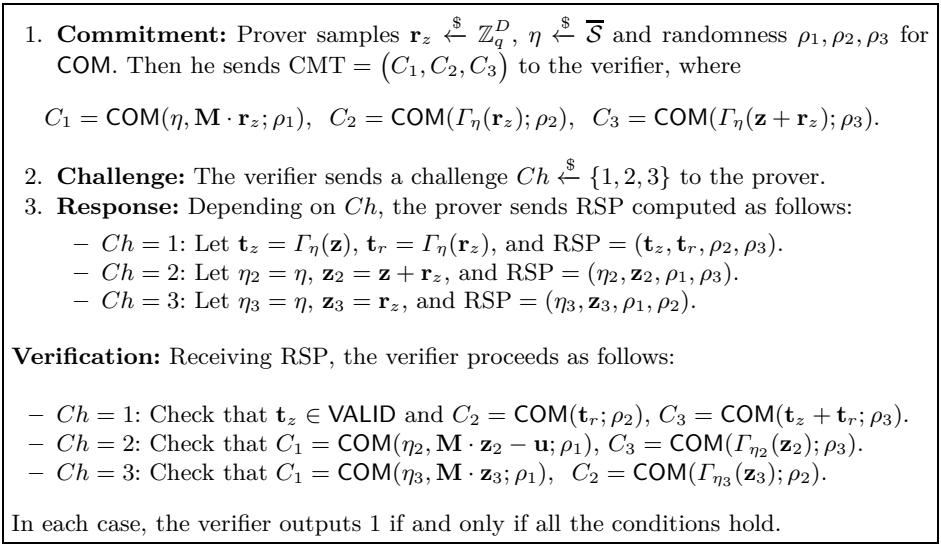


\begin{enumerate}
  \item \textbf{Commitment:} Prover samples $\mathbf{r}_z \xleftarrow{\$} \mathbb{Z}_q^D$, $\eta \xleftarrow{\$} \overline{\mathcal{S}}$ and randomness $\rho_1, \rho_2, \rho_3$ for $\mathsf{COM}$.
Then he sends $\mathrm{CMT}= \big(C_1, C_2, C_3\big)$ to the verifier, where
    \begin{eqnarray*}
    \hspace*{-17.5pt}
        C_1 =  \mathsf{COM}(\eta, \mathbf{M}\cdot \mathbf{r}_z; \rho_1), \hspace*{5pt}
        C_2 =  \mathsf{COM}(\Gamma_{\eta}(\mathbf{r}_z); \rho_2), \hspace*{5pt}
        C_3 =  \mathsf{COM}(\Gamma_{\eta}(\mathbf{z} + \mathbf{r}_z); \rho_3).
    \end{eqnarray*}

  \item \textbf{Challenge:} The verifier sends a challenge $Ch \xleftarrow{\$} \{1,2,3\}$ to the prover.
  \item \textbf{Response:} Depending on $Ch$, the prover sends $\mathrm{RSP}$ computed as follows:
  \smallskip
\begin{itemize}
\item $Ch = 1$: Let $\mathbf{t}_{z} = \Gamma_{\eta}(\mathbf{z})$, $\mathbf{t}_{r} = \Gamma_{\eta}(\mathbf{r}_z)$, and $\mathrm{RSP} = (\mathbf{t}_z, \mathbf{t}_r, \rho_2, \rho_3)$.

\item $Ch = 2$: Let $\eta_2 = \eta$, $\mathbf{z}_2 = \mathbf{z} + \mathbf{r}_z$, and
    $\mathrm{RSP} = (\eta_2, \mathbf{z}_2, \rho_1, \rho_3)$.
\item $Ch = 3$: Let $\eta_3 = \eta$, $\mathbf{z}_3 = \mathbf{r}_z$, and
 $\mathrm{RSP} = (\eta_3, \mathbf{z}_3, \rho_1, \rho_2)$.
\end{itemize}
\end{enumerate}
\textbf{Verification:}  Receiving $\mathrm{RSP}$, the verifier proceeds as follows:
\smallskip
%\vspace{-0.25cm}
          \begin{itemize}%[leftmargin=0.2cm,itemindent=.2cm,labelwidth=\itemindent,labelsep=0.2cm,align=left]
            \item $Ch = 1$: Check that $\mathbf{t}_z \in \mathsf{VALID}$ and $C_2 = \mathsf{COM}(\mathbf{t}_r; \rho_2)$, ${C}_3 = \mathsf{COM}(\mathbf{t}_z + \mathbf{t}_r; \rho_3)$.

             \item $Ch = 2$: Check that $C_1 = \mathsf{COM}(\eta_2, \mathbf{M}\cdot \mathbf{z}_2 - \mathbf{u}; \rho_1)$, ${C}_3 = \mathsf{COM}(\Gamma_{\eta_2}(\mathbf{z}_2); \rho_3)$.

            \item $Ch = 3$: Check that $C_1 =  \mathsf{COM}(\eta_3, \mathbf{M}\cdot \mathbf{z}_3; \rho_1), \hspace*{5pt}
        C_2 =  \mathsf{COM}(\Gamma_{\eta_3}(\mathbf{z}_3); \rho_2).$

          \end{itemize}
          In each case, the verifier outputs $1$ if and only if all the conditions hold.
%\rule{0pt}{3ex}
\caption{Our Stern-like zero-knowledge argument of knowledge.}
\label{Figure:Interactive-Protocol}
\end{figure}
%The properties of the given protocol are summarized in Theorem~\ref{Theorem:zk-protocol}.
\vspace*{-0.25cm}
\begin{theorem}\label{Theorem:zk-protocol}
Assume that the $\mathsf{SIVP}_{\widetilde{\mathcal{O}}(n)}$ problem is hard. Then the protocol in Fig.~\ref{Figure:Interactive-Protocol} is a statistical \emph{\textsf{ZKAoK}} with perfect completeness, soundness error~$2/3$, and communication cost~$\widetilde{\mathcal{O}}(D\log q)$. Namely:
\begin{itemize}
\item There exists a polynomial-time simulator that, on input $(\mathbf{M}, \mathbf{u})$, outputs an accepted transcript statistically close to that produced by the real prover.
\item There exists a polynomial-time knowledge extractor that, on input a commitment $\mathrm{CMT}$ and $3$ valid responses $(\mathrm{RSP}_1,\mathrm{RSP}_2,\mathrm{RSP}_3)$ to all $3$ possible values of the challenge $Ch$, outputs $\mathbf{z}' \in \mathsf{VALID}$ such that $\mathbf{M}\cdot \mathbf{z}' = \mathbf{u} \bmod q.$
\end{itemize}
\end{theorem}
Given vector $\mathbf{z}'$ outputted by the knowledge extractor, we can compute a tuple $\zeta'$ satisfying the conditions described at the beginning of this subsection, simply by ``backtracking'' the transformation steps we have done on the way. The protocol has communication cost  $\widetilde{\mathcal{O}}((10nk\ell+2m+ 4m_E+ 2\ell-3)\log q) = \widetilde{\mathcal{O}}(\lambda\cdot \ell)$ bits.

The proof of Theorem~\ref{Theorem:zk-protocol} employs standard simulation and extraction techniques for Stern-like protocols~\cite{KTX08,LNSW13,LNW15,LLMNW16-dgs} and is deferred to Appendix~\ref{appdendix:proof-for-stern}.

\subsection{Zero-Knowledge Argument of Correct Decryption}\label{subsection:correct-decryption}
We now present the underlying \textsf{ZKAoK} used by the tracing manager to generate $\Pi_{\sf trace}$. The protocol allows the prover to prove knowledge of decryption keys for the multi-bit version of Regev's encryption scheme, as well as prove the correctness of decryption. While the argument system in Section~\ref{subsection:ZK-main} deals only with $\{0,1\}$ witness vectors, here, we have to handle secret vectors/matrices of infinity norm larger than~$1$. To this end, we will employ the decomposition technique from~\cite{LNSW13,LLMNW16-dgs} to transform them to vectors with infinity norm~$1$. We thus first recall this technique.
\smallskip

\noindent
{\bf Decompositions.}
For any $B \in \mathbb{Z}_+$, define the number $\delta_B:=\lfloor \log_2 B\rfloor +1 = \lceil \log_2(B+1)\rceil$ and the sequence $B_1, \ldots, B_{\delta_B}$, where $B_j = \lfloor\frac{B + 2^{j-1}}{2^j} \rfloor$, $\forall j \in [1,\delta_B]$. As observed in~\cite{LNSW13}, the sequence satisfies $\sum_{j=1}^{\delta_B} B_j = B$ and
 any integer $v \in [-B, B]$ can be decomposed into vector $(v^{(1)}, \ldots, v^{(\delta_B)})^\top \hspace*{-2pt}\in \hspace*{-1pt}\{-1,0,1\}^{\delta_B}$ such that $\sum_{j=1}^{\delta_B}B_j \cdot v^{(j)} \hspace*{-1pt}=\hspace*{-1pt} v$.

Next, for any positive integers $\mathfrak{m}, B$, we define the decomposition matrix:
\begin{eqnarray*}
\mathbf{H}_{\mathfrak{m},B}: = \begin{bmatrix} B_1 \ldots  B_{\delta_B} &  & & & \\
  &  B_1 \ldots  B_{\delta_B} &   &  & \\
	  &   &  &  \ddots  &  \\
			  &   &  &    & B_1 \ldots  B_{\delta_B}  \\
\end{bmatrix} \in \mathbb{Z}^{\mathfrak{m} \times \mathfrak{m}\delta_B}.
\end{eqnarray*}
It then follows that, any vector $\mathbf{v} \in [-B,B]^{\mathfrak{m}}$ can be decomposed into vector $\overline{\mathbf{v}} \in \{-1,0,1\}^{\mathfrak{m}\delta_B}$ such that $\mathbf{H}_{\mathfrak{m},B} \cdot \overline{\mathbf{v}} = \mathbf{v}$.

 Once we have obtained a transformed witness vector $\overline{\mathbf{v}} \in \{-1,0,1\}^{\mathfrak{m}\delta_B}$, we then employ the usual extending-permuting technique for Stern-like protocols. For any positive integer $i$, we let $\mathsf{R}_{i}^3$ be the set of all vectors in $\{-1,0,1\}^{3i}$ that have exactly~$i$ coordinates equal to~$j$, for each $j \in \{-1,0,1\}$. Note that this set is ``closed'' under permutations, i.e., for any $\pi \in \mathcal{S}_{3i}$ and any vector $\mathbf{x}\in\mathsf{R}_{i}^3$, the following equivalence holds:
\begin{eqnarray}\label{eqn:permutation-for-3}
    \mathbf{x}\in\mathsf{R}_{i}^3\Leftrightarrow \pi(\mathbf{x})\in\mathsf{R}_{i}^3.
    \end{eqnarray}

\medskip

\noindent
The protocol is summarized as follows.
\begin{description}
  \item[Public input:] $\mathbf{B} \in \mathbb{Z}_q^{n \times m_E}$, $\mathbf{P}_1 \in \mathbb{Z}_q^{\ell \times m_E}$, $\mathbf{c}_{1,1} \in \mathbb{Z}_q^n, \mathbf{c}_{1,2} \in \mathbb{Z}_q^\ell$,  $\mathbf{b}' \in \{0,1\}^\ell$.
      \smallskip
  \item[Prover's goal:] Proving knowledge of  $\mathbf{S}_1 \in [-\beta, \beta]^{n \times \ell}$, $\mathbf{E}_1 \in [-\beta, \beta]^{\ell \times m_E}$, and $\mathbf{y} \in [-\lceil q/5 \rceil, \lceil q/5 \rceil]^\ell$, such that:
\end{description}
\vspace*{-0.25cm}
\begin{eqnarray}\label{equation:ZK-trace-goal-1}
\begin{cases}
                \mathbf{S}_1^\top \cdot \mathbf{B} + \mathbf{E}_1 = \mathbf{P}_1 \bmod q; \\
                \mathbf{c}_{1,2} - \mathbf{S}_1^\top\cdot \mathbf{c}_{1,1} = \mathbf{y} + \lfloor q/2 \rfloor\cdot \mathbf{b}' \bmod q.
\end{cases}
\end{eqnarray}

For each $j \in [\ell]$, let $\mathbf{s}_j, \mathbf{e}_j, \mathbf{p}_j$ be the $j$-th column of matrix $\mathbf{S}_1, \mathbf{E}_1^\top, \mathbf{P}_1^\top$, respectively; and let $y_j, c_j, b_j$ be the $j$-th entry of vector $\mathbf{y}, \mathbf{c}_{1,2}, \mathbf{b}'$, respectively.
Then observe that (\ref{equation:ZK-trace-goal-1}) can be re-written as:
\begin{eqnarray}\label{equation:ZK-trace-goal-2}
\forall j \in [\ell]:
\begin{cases}
\mathbf{B}^T \cdot \mathbf{s}_j + \mathbf{I}_{m_E} \cdot \mathbf{e}_j = \mathbf{p}_j \bmod q \\
\mathbf{c}_{1,1}^T \cdot \mathbf{s}_j + 1\cdot y_j = {c}_j- b_j\cdot  \lfloor q/2\rfloor\bmod q.
\end{cases}
\end{eqnarray}

Using basic algebra, we can manipulate the equations in~(\ref{equation:ZK-trace-goal-2})by rearranging the secret vectors and forming public matrices accordingly. As a result, we obtain a unifying equation of the form
\[\mathbf{M}_0\cdot \mathbf{z}_0=\mathbf{u}\bmod q,\] where  $\mathbf{M}_0,~\mathbf{u}$ are built from the public input while $\mathbf{z}_0=(\mathbf{z}_1\|\mathbf{z}_2)$, with
\begin{eqnarray*}
\begin{cases}
\mathbf{z}_1 = (\mathbf{s}_1  \| \cdots \| \mathbf{s}_\ell  \| \mathbf{e}_1  \| \cdots \| \mathbf{e}_\ell) \in [-\beta,\beta]^{(n+m_E)\ell},\\ %\mathbf{v}_i^{\mathrm{id}[i]}=\mathbf{v}_i,~ \mathbf{v}_i^{1-\mathrm{id}[i]}=\mathbf{0}^m ~\forall i\in[\ell],\\
\mathbf{z}_2=(y_1, \ldots, y_\ell)^\top \in \big[\lceil-q/5\rceil, \lceil q/5\rceil\big]^{\ell}.

\end{cases}
\end{eqnarray*}
Next, we can use the decomposition  and extension  techniques described above to handle our  secret vectors. First, decompose $\mathbf{z}_1$ into $\overline{\mathbf{z}}_1 \in \{-1,0,1\}^{D_1}$ with $D_1=(n+m_E)\ell\delta_\beta$ and  $\mathbf{z}_2$ into  $\overline{\mathbf{z}}_2 \in \{-1,0,1\}^{D_2}$ with $D_2=\ell\delta_{\lceil q/5\rceil}$, and form the secret vector $\overline{\mathbf{z}}_0=(\overline{\mathbf{z}}_1\|\overline{\mathbf{z}}_2)\in\{-1,0,1\}^{L}$ with $D'=D_1+D_2$. Then extend $\overline{\mathbf{z}}_0$ to $\mathbf{z} \in \mathsf{R}_{D'}^3$. Using some basic algebra, we can form public matrix $\mathbf{M}$ such that
 \[\mathbf{M}\cdot \mathbf{z} = \mathbf{M}_0\cdot \mathbf{z}_0 =\mathbf{u} \bmod q.\]

%Next, we form vector $\mathbf{w} = (\mathbf{s}_1  \| \cdots \| \mathbf{s}_\ell  \| \mathbf{e}_1  \| \cdots \| \mathbf{e}_\ell) \in [-\beta,\beta]^{(n+m_E)\ell}$, then decompose it into $\overline{\mathbf{w}} \in \{-1,0,1\}^{L_1}$ with $L_1=(n+m_E)\ell\delta_\beta$. %, and extend $\overline{\mathbf{w}}$ to $\mathbf{w}^* \in \mathsf{R}_{3(n+m_E)\ell\delta_\beta}$.
%At the same time, we decompose vector $\mathbf{y} = (y_1, \ldots, y_\ell)^T \in \big[\lceil-q/5\rceil, \lceil q/5\rceil\big]^{\ell}$ into $\overline{\mathbf{y}} \in \{-1,0,1\}^{L_2}$ with $L_2=\ell\delta_{\lceil q/5\rceil}$.%, and then extend $\overline{\mathbf{y}}$ to $\mathbf{y}^* \in \mathsf{R}_{3\ell\delta_{\lceil q/5\rceil}}$.

%We then combine the above $2$ vectors into $\overline{\mathbf{z}}=(\overline{\mathbf{w}}\|\overline{\mathbf{y}})$, and extend $\overline{\mathbf{z}}$ to $\mathbf{z} \in \mathsf{R}_{L}^3$ with $L=L_1+L_2$. Note that equations in~(\ref{eq:protocol-2-step-1}) can be equivalently unified to an equation:
%\[\mathbf{M}\cdot \mathbf{z} = \mathbf{u} \bmod q,\]
%where matrix $\mathbf{M}$ and vector $\mathbf{u}$ are built from the public input.
%Now, we let $D = 3(n+m_E)\ell\delta_\beta + 3\ell\delta_{\lceil q/5\rceil}$ and
Now, define the tailored set $\mathsf{VALID} = \mathsf{R}_{D'}^3$ and let $\overline{\mathcal{S}} = \mathcal{S}_{D}$ with $D=3D'$. For every $\eta \in \overline{\mathcal{S}}$,  define $\Gamma_\eta = \eta$. It is checked that our secret vector $\mathbf{z}\in \mathsf{VALID}$. Observe the equivalence in~(\ref{eqn:permutation-for-3}), for any  $\eta \in \overline{\mathcal{S}}$, the following equivalence also holds: \begin{eqnarray*}
    \mathbf{z}\in\mathsf{VALID}\Longleftrightarrow \Gamma_{\eta}(\mathbf{z})\in\mathsf{VALID}.
    \end{eqnarray*}

 At this point, we can see that the desired statistical $\mathsf{ZKAoK}$ protocol can be obtained by running the interactive protocol described in Fig.~\ref{Figure:Interactive-Protocol}.

%let $\mathbf{z} = ((\mathbf{w}^*)^T \| (\mathbf{y}^*)^T)^T \in \{-1,0,1\}^D$, where $D = 3(n+m_E)\ell\delta_\beta + 3\ell\delta_{\lceil q/5\rceil}$, then we obtain equation $\mathbf{M}\cdot \mathbf{z} = \mathbf{u} \bmod q$, for public matrix $\mathbf{M}$ and public vector $\mathbf{u}$.

%Using similar arguments as in Section~\ref{subsection:ZK-protocol-1}, we then can obtain the desired zero-knowledge argument system. The protocol has communication cost $D\log q \cdot \omega(\log \lambda)= \widetilde{\mathcal{O}}(\lambda)\cdot \mathcal{O}(t)$ bits.

\section{Achieving Deniability with Ease}\label{section:deniability}

This section presents how to make our fully dynamic group signature scheme deniable, in the sense of~\cite{IEHST16} and in a relatively simple manner. %in the sense of~\cite{IEHST16}. %In Section~\ref{subsection:background-denial}, we recall some background on deniability.
In Section~\ref{subsection:denial-model}, we first incorporate Ishida et al.'s notion of deniability into Bootle et al.'s model~\cite{BCCGG16}.  Then, we show in Section~\ref{subsection:denial-scheme} that a scheme satisfying this extended model can be easily constructed based on the one from Section~\ref{section:main-scheme} and a zero-knowledge argument of ``denial opening'', demonstrating that a given Regev ciphertext does not decrypt to a particular message. The zero-knowledge argument is described in detail in Section~\ref{subsection:ZK-denial-decryption}.

\begin{comment}
\subsection{Background of Deniability}\label{subsection:background-denial}
%In this section, we demonstrate that it is relatively simple to make our group signature scheme deniable, in the sense of~\cite{IEHST16}.
In a nutshell, deniable group signatures are group signatures with an additional functionality: it allows the group authority to provide a digital evidence that a given group user did \emph{not} generate a signature in question. In~\cite{IEHST16}, Ishida et al. discussed various situations in which such functionality helps to protect the privacy of users. For instance, the police wants to check whether a suspect was in a building at a specific time period, and the entrance and exit control of the building is implemented using a group signature. If the police asks the tracing manager to reveal the signers of all signatures generated during that period, then this will seriously violate the privacy of innocent users. In such situation, deniable group signatures make it possible to prove that
the signer of a given signature is not the suspect, while still keeping the signer's identity secret.
\end{comment}
\subsection{Fully Dynamic Group Signatures with Deniability}\label{subsection:denial-model}
We now extend Bootle et al.'s model of FDGS (in Section~\ref{subsection:FDGS-definitions}) to capture the notion of deniability. Our extension follows Ishida et al.'s approach~\cite{IEHST16} when they incorporate the deniability functionality into Sakai et al.'s model~\cite{SSEH+12}.

A fully dynamic group signature scheme with deniability ($\mathsf{FDGSwD}$) consists of all algorithms of an  $\mathsf{FDGS}$, as specified in Section~\ref{subsection:FDGS-definitions}, and two additional algorithms: $\mathsf{DTrace}$ and $\mathsf{DJudge}$.
\begin{description}
\item[$\ms{DTrace}(\ms{gpk},\ms{tsk},\ms{info}_{\tau},\mathbf{reg}, \mathsf{uid}',M,\Sigma)\rightarrow \Pi_{D(\ms{uid}')}$.] The denial opening algorithm, run by the $\ms{TM}$,  returns a proof $\Pi_{D(\ms{uid}')}$ showing that $\mathsf{uid}'$ did not compute the signature $\Sigma$ on the message $M$ at time $\tau$.
\item[$\ms{DJudge}(\ms{gpk},\ms{uid},\ms{info}_{\tau},\Pi_{D(\mathsf{uid}')},M,\Sigma)\rightarrow 0/1$.] The denial judgement algorithm checks the validity of $\Pi_{D(\mathsf{uid}')}$ outputted by the $\ms{DTrace}$ algorithm.
\end{description}

An $\mathsf{FDGSwD}$ scheme must satisfy \emph{correctness}, \emph{anonymity}, \emph{non-frameability}, \emph{traceability} and \emph{tracing soundness}.    %Again, this extension follows from the extension of Sakai et al.'s model~\cite{SSEH+12} to Ishida et al.'s model~\cite{IEHST16}.
To formalize these requirements, we first specify an additional oracle, called $\mathsf{DTrace}$.
\begin{description}
\item[$\mathsf{DTrace}(\mathsf{uid}',M,\Sigma,\mathsf{info}_{\tau})$.] This oracle returns a denial opening proof $\Pi_{D(\mathsf{uid}')}$. It is required that $\Sigma$ is not generated by the challenged oracle or the user $\mathsf{uid}'$ is not in the challenged user set to eliminate the cases where the adversary trivially breaks anonymity.
\end{description}

Then, the  correctness and security requirements of an $\mathsf{FDGSwD}$ scheme are defined by experiments $\mathbf{Exp}_{\mathsf{FDGSwD},\mathcal{A}}^{\ms{corr}}(\lambda)$ and
\[
\mathbf{Exp}_{\mathsf{FDGSwD},\mathcal{A}}^{\ms{anon-b}}(\lambda),
\mathbf{Exp}_{\mathsf{FDGSwD},\mathcal{A}}^{\ms{non-frame}}(\lambda),
\mathbf{Exp}_{\mathsf{FDGSwD},\mathcal{A}}^{\ms{trace}}(\lambda),
\mathbf{Exp}_{\mathsf{FDGSwD},\mathcal{A}}^{\ms{trace-sound}}(\lambda),
\]
described in \textbf{Fig.~\ref{fig:exp-denial-anon}}, \textbf{\ref{fig:exp-denial-corr-anon}}, and~\textbf{\ref{fig:exp-denial-trace-trace-sound}}. The differences between these requirements and the original ones of Bootle et al.'s model (Section~\ref{subsection:FDGS-definitions}) are highlighted in the below.

\smallskip

\noindent\emph{Correctness} %not only requires correctness as in the Bootle et al.'s model, as recalled in Section~\ref{subsection:FDGS-definitions}, it
additionally requires that a (denial) proof generated by algorithm $\mathsf{DTrace}$, with respect to a message-signature pair $(M,\Sigma)$ and an arbitrary user $\mathsf{uid}'$ who is not the real signer, is always accepted by algorithm $\mathsf{DJudge}$.
%Details are modelled in the experiment $\mathbf{Exp}_{\mathsf{FDGSwD},\mathcal{A}}^{\ms{corr}}(\lambda)$.
\vspace{-10pt}
\begin{figure}
\begin{center}
\begin{minipage}{12cm}
\underline{Experiment: $\mb{Exp}_{\ms{FDGSwD},\mcA}^{\ms{corr}}(\lambda)$}\\
     $\ms{pp}\leftarrow \ms{GSetup}(\lambda),\ms{HUL}:=\emptyset$.\\
    $\langle(\ms{info},\ms{mpk},\ms{msk});(\ms{tpk},\ms{tsk})\rangle\leftarrow \langle\ms{GKgen_{GM}}(\ms{pp}),\ms{GKgen_{TM}}(\ms{pp})\rangle$.\\
     Set $\ms{gpk}=(\ms{pp},\ms{mpk},\ms{tpk})$.
    $(\ms{uid},M,\tau)\leftarrow\mcA^{\ms{AddU,RReg,GUpdate}}(\ms{gpk},\ms{info})$.\\
     If $\ms{uid}\notin\ms{HUL}$ or $\ms{gsk}[\ms{uid}]=\bot$ or $\ms{info}_{\tau}=\bot$ or $\ms{IsActive}(\ms{info}_{\tau},\mb{reg},\ms{uid})=0$, return $0$.\\
    $\Sigma\leftarrow\ms{Sign}(\ms{gpk},\ms{gsk}[\ms{uid}],\ms{info}_{\tau},M)$,
    $(\ms{uid}^*,\Pi_{\mathsf{Trace}})\leftarrow\mathsf{Trace}(\ms{gpk},\ms{tsk},\ms{info}_{\tau},\mb{reg},M,\Sigma)$.\\
    $\mathsf{uid}'\leftarrow\mathcal{A}^{\mathsf{AddU,RReg,Gupdate}}(\Sigma,\mathsf{uid}^*,\Pi_{\ms{Trace}}).$\\
    If $\mathsf{uid}'=\mathsf{uid}\vee \mathsf{uid}'\notin \mathsf{HUL}$, then return $0$.\\
    $\Pi_{D(\mathsf{uid}')}\leftarrow\mathsf{DTrace}(\ms{gpk},\ms{tsk},\ms{info}_{\tau},\mb{reg},\mathsf{uid}',M,\Sigma).$\\
    Return $1$ if $\ms{Verify}(\ms{gpk},\ms{info}_{\tau},M,\Sigma)=0$ or
     $\ms{uid}\neq \ms{uid}^*$ or $\ms{Judge}(\ms{gpk},\ms{uid},\ms{info}_{\tau},\Pi_{\mathsf{Trace}},M,\Sigma)=0$ or $\mathsf{DJudge}(\mathsf{gpk},\mathsf{uid}',\mathsf{info}_\tau, \Pi_{D(\mathsf{uid}')} M,\Sigma)=0$. \\
     \vspace{-10pt}
\end{minipage}
\end{center}
\caption{Experiment to define correctness for an  $\mathsf{FDGSwD}$ scheme.} \label{fig:exp-denial-anon}
\end{figure}

\vspace{-10pt}

%Define  $\mathbf{Adv}_{\mathsf{FDGSwD},\mathcal{A}}^{\ms{corr}}(\lambda)$ of  adversary $\mathcal{A}$ against correctness of an $\mathsf{FDGSwD}$ scheme as $\text{Pr}[\mathbf{Exp}_{\mathsf{FDGSwD},\mathcal{A}}^{\mathsf{corr}}(\lambda)=1]$. An $\mathsf{FDGSwD}$ scheme is correct if the advantage of any $\mathrm{PPT}$ adversary $\mathcal{A}$ is negligible.

\smallskip
\noindent\emph{Anonymity} is defined with respect to an adversary who also has access to the additional oracle $\mathsf{DTrace}$ specified above. %Details are modelled in the experiment $\mathbf{Exp}_{\mathsf{FDGSwD},\mathcal{A}}^{\ms{anon-b}}(\lambda)$.

%Define  $\mathbf{Adv}_{\mathsf{FDGSwD},\mathcal{A}}^{\ms{anon}}(\lambda)$ of  adversary $\mathcal{A}$ against anonymity of an $\mathsf{FDGSwD}$ scheme as $|\text{Pr}[\mathbf{Exp}_{\mathsf{FDGSwD},\mathcal{A}}^{\mathsf{anon-1}}(\lambda)=1]-\text{Pr}[\mathbf{Exp}_{\mathsf{FDGSwD},\mathcal{A}}^{\mathsf{anon-0}}(\lambda)=1]|$. An $\mathsf{FDGSwD}$ scheme is anonymous if the advantage of any $\mathrm{PPT}$ adversary $\mathcal{A}$ is negligible.

\begin{figure}
\begin{center}
\begin{minipage}{12cm}
\underline{Experiment: $\mb{Exp}_{\ms{FDGSwD},\mcA}^{\ms{anon-b}}(\lambda)$}\\
     $\ms{pp}\leftarrow \ms{GSetup}(\lambda),\ms{HUL,CUL,BUL,CL,SL}:=\emptyset$.\\
     $(\ms{st}_{\ms{init}},\ms{info},\ms{mpk},\ms{msk})\leftarrow\mcA^{\langle\cdot,\ms{GKgen_{TM}\rangle}(\ms{pp})}(\ms{init}:\ms{pp})$.\\
     Return $0$ if $\ms{GKgen_{TM}}$ did not accept or $\mcA$'s output is not well-formed.\\
     Denote the output of $\ms{GKgen_{TM}}$ as $(\ms{tpk},\ms{tsk})$, and set $\ms{gpk}=(\ms{pp},\ms{mpk},\ms{tpk})$.\\
     $b^*\leftarrow\mcA^{\ms{AddU,CrptU,RevealU,SndToU,Trace,MReg,DTrace,Chal_b}}(\ms{play}:\ms{st_{init}},\ms{gpk})$.\\
     Return $b^*$.

\underline{Experiment: $\mb{Exp}_{\ms{FDGSwD},\mcA}^{\ms{non-frame}}(\lambda)$}\\
     $\ms{pp}\leftarrow \ms{GSetup}(\lambda),\ms{HUL,CUL,BUL,SL}=\emptyset$.\\
    $(\ms{st}_{\ms{init}},\ms{info},\ms{mpk},\ms{msk},\ms{tpk},\ms{tsk})\leftarrow\mcA(\ms{init}:\ms{pp})$.\\
     Return $0$ if $\mcA$'s output is not well-formed, otherwise set $\ms{gpk}=(\ms{pp},\ms{mpk},\ms{tpk})$.\\
     $(M,\Sigma,\ms{uid},\Pi_{\mathsf{Trace}},\ms{info}_{\tau})\leftarrow\mcA^{\ms{CrptU,RevealU,SndToU,MReg,Sign}}(\ms{play}:\ms{st_{init}},\ms{gpk})$.\\
     $\Pi_{D(\mathsf{uid})}\leftarrow\mathsf{DTrace}(\mathsf{gpk},\mathsf{tsk},\mathsf{info}_\tau,\mathbf{reg},\mathsf{uid},M,\Sigma).$\\
     If
     $\ms{Verify}(\ms{gpk},\ms{info}_{\tau},M,\Sigma)=0$ or $\ms{uid}\notin\ms{HUL}\setminus\ms{BUL}$ or $(\mathsf{uid},M,\Sigma,\tau)\in \ms{SL}$, return $0$. \\
     Return $1$ if
     $\ms{Judge}(\ms{gpk},\ms{uid},\ms{info}_{\tau},\Pi_{\mathsf{Trace}},M,\Sigma)=1$ or $\ms{DJudge}(\ms{gpk},\ms{uid},\ms{info}_{\tau},\Pi_{D(\mathsf{uid})},M,\Sigma)=0$.
     \vspace{-10pt}
\end{minipage}
\end{center}
\caption{Experiment to  anonymity and non-frameability of an  $\mathsf{FDGSwD}$ scheme.} \label{fig:exp-denial-corr-anon}
\end{figure}

\begin{comment}
\begin{figure}
\begin{center}
\begin{minipage}{12cm}

     \vspace{-10pt}
\end{minipage}
\end{center}
\caption{Experiment to define anonymity for an $\mathsf{FDGSwD}$ scheme.} \label{fig:exp-denial-anon}
\end{figure}

\vspace{-10pt}
\end{comment}

\smallskip
\noindent\emph{Non-frameability} additionally requires that it is infeasible for the adversary to generate a valid signature whose denial opening proof with respect to an honest user is not accepted by the $\mathsf{DJudge}$ algorithm. %Details are modelled in the experiment $\mathbf{Exp}_{\mathsf{FDGSwD},\mathcal{A}}^{\ms{non-frame}}(\lambda)$.
%Define  $\mathbf{Adv}_{\mathsf{FDGSwD},\mathcal{A}}^{\ms{non-frame}}(\lambda)$ of  adversary $\mathcal{A}$ against non-frameability of an $\mathsf{FDGSwD}$ scheme as $\text{Pr}[\mathbf{Exp}_{\mathsf{FDGSwD},\mathcal{A}}^{\mathsf{non-frame}}(\lambda)=1]$. An $\mathsf{FDGSwD}$ scheme is non-frameable if the advantage of any $\mathrm{PPT}$ adversary $\mathcal{A}$ is negligible.

\begin{comment}
\begin{figure}
\begin{center}
\begin{minipage}{12cm}

\end{minipage}
\end{center}
\caption{Experiment to define non-frameability for an $\mathsf{FDGSwD}$ scheme.} \label{fig:exp-denial-non-frame}
\end{figure}
\end{comment}

\smallskip
\noindent\emph{Traceability} additionally demands that the adversary is not able to produce a valid signature such that opening proof is accepted by the $\mathsf{Judge}$ algorithm, but the denial opening proof with respect to the same user is also accepted by the $\mathsf{DJudge}$ algorithm. %Details are modelled in the experiment $\mathbf{Exp}_{\mathsf{FDGSwD},\mathcal{A}}^{\ms{trace}}(\lambda)$.

%Define  $\mathbf{Adv}_{\mathsf{FDGSwD},\mathcal{A}}^{\ms{trace}}(\lambda)$ of  adversary $\mathcal{A}$ against traceability of an $\mathsf{FDGSwD}$ scheme as $\text{Pr}[\mathbf{Exp}_{\mathsf{FDGSwD},\mathcal{A}}^{\mathsf{trace}}(\lambda)=1]$. An $\mathsf{FDGSwD}$ scheme is traceable if the advantage of any $\mathrm{PPT}$ adversary $\mathcal{A}$ is negligible.

\smallskip
\noindent\emph{Tracing Soundness} also requires that the adversary should not be able to generate a valid signature such that the denial opening proof for the actual signer is accepted by the $\mathsf{DJudge}$ algorithm. %Details are modelled in the experiment $\mathbf{Exp}_{\mathsf{FDGSwD},\mathcal{A}}^{\ms{trace-sound}}(\lambda)$.

%Define  $\mathbf{Adv}_{\mathsf{FDGSwD},\mathcal{A}}^{\ms{trace-sound}}(\lambda)$ of  adversary $\mathcal{A}$ against tracing soundness of an $\mathsf{FDGSwD}$ scheme as $\text{Pr}[\mathbf{Exp}_{\mathsf{FDGSwD},\mathcal{A}}^{\mathsf{trace}}(\lambda)=1]$. An $\mathsf{FDGSwD}$ scheme is tracing sound if the advantage of any $\mathrm{PPT}$ adversary $\mathcal{A}$ is negligible.

\begin{figure}
\begin{center}
\begin{minipage}{12cm}

\underline{Experiment: $\mb{Exp}_{\ms{FDGSwD},\mcA}^{\mathsf{trace}}(\lambda)$}\\
     $\ms{pp}\leftarrow \ms{GSetup}(\lambda),\ms{HUL,CUL,BUL,SL}=\emptyset$.\\
    $(\ms{st}_{\ms{init}},\ms{tpk},\ms{tsk})\leftarrow\mcA^{\langle\ms{GKgen_{GM}}(\ms{pp}),\cdot\rangle}(\ms{init}:\ms{pp})$.\\
     Return $0$ if $\ms{GKgen_{GM}}$ did not accept or $\mcA$'s output is not well-formed.\\
     Denote the output of $\ms{GKgen_{GM}}$ as $(\ms{mpk},\ms{msk},\ms{info})$, and set $\ms{gpk}=(\ms{pp},\ms{mpk},\ms{tpk})$.\\
     $(M,\Sigma,\tau)\leftarrow\mcA^{\ms{AddU,CrptU,SndToGM,RevealU,MReg,Sign,GUpdate}}(\ms{play}:\ms{st_{init}},\ms{gpk},\ms{info})$.\\
     If $\ms{Verify}(\ms{gpk},\ms{info}_{\tau},M,\Sigma)=0$, return $0$.\\
     $(\ms{uid},\Pi_{\mathsf{Trace}})\leftarrow\mathsf{Trace}(\ms{gpk},\ms{tsk},\ms{info}_{\tau},\mb{reg},M,\Sigma)$.\\
     $\Pi_{D(\mathsf{uid})}\leftarrow\mathcal{A}^{\mathsf{AddU,CrptU,SndToGM,RevealU,MReg,Sign,GUpdate}}(\mathsf{play}:\mathsf{st_{init}},\mathsf{uid},\Pi_{\mathsf{Trace}})$.\\
     Return $1$ if
      $\ms{IsActive}(\ms{info}_{\tau},\mb{reg},\ms{uid})=0$ or $ \ms{uid}=\bot $ or
     $\ms{Judge}(\ms{gpk},\ms{uid},\ms{info}_{\tau},\Pi_{\mathsf{Trace}},M,\Sigma)=0$ or
     $\mathsf{DJudge}(\ms{gpk},\ms{uid},\ms{info}_{\tau},\Pi_{D(\mathsf{uid})},M,\Sigma)=1$.\\
    \underline{Experiment: $\mb{Exp}_{\ms{FDGSwD},\mcA}^{\ms{trace-sound}}(\lambda)$}\\
    $\ms{pp}\leftarrow \ms{GSetup}(\lambda),\ms{CUL}=\emptyset$.
     $(\ms{st}_{\ms{init}},\ms{info},\ms{mpk},\ms{msk},\ms{tpk},\ms{tsk})\leftarrow\mcA(\ms{init}:\ms{pp})$.\\
     Return $0$ if $\mcA$'s output is not well-formed, otherwise set $\ms{gpk}=(\ms{pp},\ms{mpk},\ms{tpk})$.\\
    $(M,\Sigma,\ms{uid}_0,\Pi_{\mathsf{Trace},0},\Pi_{D(\mathsf{uid}_0)},\ms{uid}_1,\Pi_{\mathsf{Trace},1},\ms{info}_{\tau})\leftarrow\mcA^{\ms{CrptU,MReg}}(\ms{play}:\ms{st_{init}},\ms{gpk})$.\\
    If $\ms{Verify}(\ms{gpk},\ms{info}_{\tau},M,\Sigma)=0\vee\ms{Judge}(\ms{gpk},\ms{uid}_0,\ms{info}_{\tau},\Pi_{\mathsf{Trace},0},M,\Sigma)=0$, return $0$.\\
    Return $1$ if
    $\{\ms{uid}_0~(\neq \bot)\neq\ms{uid}_1~(\neq \bot)\wedge \ms{Judge}(\ms{gpk},\ms{uid}_1,\ms{info}_{\tau},\Pi_{\mathsf{Trace},1},M,\Sigma)=1\} \vee \\ \ms{DJudge}(\ms{gpk},\ms{uid}_0,\ms{info}_{\tau},\Pi_{D(\mathsf{uid}_0)},M,\Sigma)=1.$
\vspace{-10pt}
\end{minipage}
\end{center}
\caption{Experiments to define  traceability and tracing soundness  of an  $\mathsf{FDGSwD}$ scheme.}\label{fig:exp-denial-trace-trace-sound}
\end{figure}

\begin{comment}

\begin{figure}%[!htb]

\begin{center}
\begin{minipage}{12cm}

  \vspace{-10pt}
  \end{minipage}
\end{center}
\caption{Experiment to define  tracing soundness of an $\mathsf{FDGSwD}$ scheme.}\label{fig:exp-trace-sound}
\end{figure}
\end{comment}

\begin{definition}
Let $\lambda$ be any security parameter and $\mc{A}$ be any PPT adversary. %correctness and security experiments of a fully dynamic group signature scheme with deniability is formally defined using the above experiments.
For correctness, non-frameability, traceability and tracing soundness, the advantage of the adversary is defined as the probability of outputting~$1$ in the corresponding experiment. For anonymity, the advantage is defined as the absolute difference of probability of outputting~$1$ between experiment $\mb{Exp}_{\ms{FDGSwD},\mcA}^{\ms{anon-1}}$ and experiment $\mb{Exp}_{\ms{FDGSwD},\mcA}^{\ms{anon-0}}$.

A fully dynamic group signature scheme with deniability is said to be correct and secure (i.e., anonymous, non-frameable, traceable and tracing sound) if the advantages of the adversary in all considered experiments are negligible in~$\lambda$.
\end{definition}

%\subsection{Fully Dynamic Group Signature Scheme with Deniability}

\subsection{Our Construction}\label{subsection:denial-scheme}
 As shown by Ishida et al.~\cite{IEHST16}, the main difficulty towards realizing the deniability functionality is to prove in zero-knowledge that a given ciphertext does not decrypt to a particular message. Such a mechanism is non-trivial to realize in general, but it can be done quite easily for our scheme  by using the Stern-like technique for proving inequality.
Let us look at the equation
\[
\mathbf{c}_{1,2} - \mathbf{S}_1^\top\cdot \mathbf{c}_{1,1} = \mathbf{y} + \lfloor q/2 \rfloor\cdot \mathbf{b}' \bmod q,
\]
that appears in the argument of correct decryption from Section~\ref{subsection:correct-decryption}. To realize deniability, the plaintext $\mathbf{b}' \in \{0,1\}^\ell$ should be hidden, and we have to prove that $\mathbf{b}' \neq \mathsf{uid}'$, for some given $\mathsf{uid}'\in \{0,1\}^\ell$. Subtracting both sides of the equation by $\lfloor q/2 \rfloor\cdot \mathsf{uid}'$, we get the following:
\[
\mathbf{c}_{1,2} - \mathbf{S}_1^\top\cdot \mathbf{c}_{1,1} - \lfloor q/2 \rfloor\cdot \mathsf{uid}' = \mathbf{y} + \lfloor q/2 \rfloor\cdot \mathbf{b} \bmod q,
\]
where  $\mathbf{b}=\mathbf{b'}-\mathsf{uid}'\in\{-1,0,1\}^{\ell}$ is non-zero. Therefore, the problem towards realizing the deniability is reduced to proving knowledge of a non-zero vector $\mathbf{b}\in\{-1,0,1\}^{\ell}$. %By using the Stern-like trick for proving inequality suggested in~\cite{LNSW13}, it  can be done quite easily for our scheme.
This can be done by just appending $2\ell-1$ ``dummy'' entries to $\mathbf{b}$ and randomly permuting the extended vector.
We defer the detailed description to Section~\ref{subsection:ZK-denial-decryption}.  %this is done just by extending $\mathbf{b}^*$ to a vector in $\mathsf{B}_\ell^{2\ell-1}$ and randomly permuting the extended vector!
%Let $ \widetilde{\mathsf{R}}_{\ell}^{3}$ be the set of all vectors in $\{-1,0,1\}^{3\ell-1}$ that has $\ell$ coordinates equal to $-1$ and $1$, and $\ell-1$ coordinates equal to $0$.
%We will described this in detail in Section~\ref{subsection:ZK-denial-decryption}.

Now, let us describe and analyze our fully dynamic group signature scheme with deniability. Besides all the algorithms specified in Section~\ref{subsection:scheme-description}, the scheme contains two additional algorithms, $\mathsf{DTrace}$ and $\mathsf{DJudge}$, as follows.

\begin{description}
\item [$\mathsf{DTrace}(\mathsf{gpk}, \mathsf{tsk}, \mathsf{info}_\tau, \mathbf{reg}, \mathsf{uid}', M, \Sigma$).] Let $\mathbf{b}'$ be the decryption result of the $\mathsf{Trace}$ algorithm, denote $\mathbf{b}=\mathbf{b}'-\mathsf{uid}'$. This algorithm generates a $\mathsf{NIZKAoK}$ to demonstrate the possession of $\mathbf{S}_1\in\mathbb{Z}^{n\times \ell}, \mathbf{E}_1\in\mathbb{Z}^{\ell\times m_E}, \mathbf{y}\in \mathbb{Z}^\ell$, and non-zero vector $\mathbf{b}\in\{-1,0,1\}^\ell$ such that
            \begin{eqnarray}\label{equation:pi-denial}
            \begin{cases}
                \| \mathbf{S}_1 \|_\infty \leq \beta; \hspace*{2.8pt}  \| \mathbf{E}_1 \|_\infty \leq \beta; \hspace*{2.8pt} \| \mathbf{y}\|_\infty \leq  \lceil q/5 \rceil; \\
                \mathbf{S}_1^\top \cdot \mathbf{B} + \mathbf{E}_1 = \mathbf{P}_1 \bmod q; \\
                \mathbf{c}_{1,2} - \mathbf{S}_1^\top\cdot \mathbf{c}_{1,1}-\lfloor q/2 \rfloor\cdot \mathsf{uid}'  = \mathbf{y} + \lfloor q/2 \rfloor\cdot \mathbf{b} \bmod q.
            \end{cases}
            \end{eqnarray}
            The detailed description of the protocol is presented in Section~\ref{subsection:ZK-denial-decryption}.
        The protocol is repeated $\kappa = \omega(\log \lambda)$ times to achieve negligible soundness error  and   made non-interactive via the Fiat-Shamir heuristic as a triple $\Pi_{D(\mathsf{uid}')}= (\{\mathrm{CMT}_i\}_{i=1}^\kappa, \mathrm{CH}, \{\mathrm{RSP}\}_{i=1}^\kappa)$, where
          \begin{eqnarray}\label{equation:pi-denial-FS}
            \mathrm{CH} = \mathcal{H}_{\mathsf{FS}}\big(\{\mathrm{CMT}_i\}_{i=1}^\kappa, \mathsf{gpk}, \mathsf{info}_\tau, M, \Sigma, \mathsf{uid}'\big) \in \{1,2,3\}^\kappa.
          \end{eqnarray}
  \item [$\ms{DJudge}(\ms{gpk}, \mathsf{uid}', \ms{info}_{\tau}, \Pi_{D(\mathsf{uid}')}, M, \Sigma)$.] If $\mathsf{Verify}$ algorithm outputs $0$, then this algorithm returns $0$. Otherwise, this algorithm then verifies the argument $\Pi_{D(\mathsf{uid}')}$ with respect to common input $(\mathsf{gpk}, \mathsf{info}_\tau, M, \Sigma, \mathsf{uid}')$, in a similar manner as in algorithm $\mathsf{Verify}$.

    If $\Pi_{D(\mathsf{uid}')}$ does not verify, return $0$. Otherwise, return $1$.
\end{description}
\smallskip

\noindent
{\sc{Efficiency.}} In comparison with the \textsf{FDGS} scheme analyzed in Section~\ref{subsection:analysis-of-scheme}, the extended scheme additionally features the $\mathsf{NIZKAoK}$ $\Pi_{D(\mathsf{uid}')}$, which is of bit-size $\widetilde{\mathcal{O}}(\ell^2+\lambda\cdot\ell)$.  (Recall that $\lambda$ is the security parameter while $\ell=\log N$.)
\smallskip

\noindent{\sc{Correctness.}} The correctness of the scheme is analyzed as in Section~\ref{subsection:analysis-of-scheme}, but here we additionally need to prove that algorithm $\mathsf{DJudge}$ will accept the proof $\Pi_{D(\mathsf{uid}')}$ for $\mathsf{uid}'$ that is not the real signer. Since the tracing manager can correctly identify the signer, he possesses a valid tuple $(\mathbf{S}_1,\mathbf{E}_1,\mathbf{y},\mathbf{b})$ such that the conditions in~(\ref{equation:pi-denial}) hold. Correctness then follows from the perfect completeness of our argument system generating $\Pi_{D(\mathsf{uid}')}$.
\smallskip

\noindent{\sc{Anonymity.}} It is analyzed as in Lemma~\ref{lemma:anonymity}, except that the challenger additionally needs to answer the $\mathsf{DTrace}$ oracle queries. However, this can be done in the same way it answers the $\mathsf{Trace}$ oracle queries. In other words, when the challenger replies with real proof $\Pi_{\ms{trace}}$, it also replies with real proof $\Pi_{D(\cdot)}$; when the challenger replies with simulated proof $\Pi_{\ms{trace}}$, it also replies with simulated proof $\Pi_{D(\cdot)}$. This minor change does not affect anonymity, since the argument system generating $\Pi_{D(\cdot)}$ is statistically zero-knowledge.%Therefore, our fully dynamic group signature scheme with deniability is still anonymous.  %Then, by the statistical zero-knowledge of the argument system generating $\Pi_{D(\cdot)}$, anonymity is  proved.
\smallskip

\noindent{\sc{Non-frameability.}}
We prove non-frameability by contradiction. Suppose that the non-frameability adversary~$\mc{A}$ succeeds with non-negligible advantage $\epsilon$. Then we build a PPT algorithm~$\mc{B}$ that, with non-negligible probability, either breaks the security of our accumulator or solves $\ms{SIS}_{n,q,m,1}^{\infty}$ problem associated with matrix $\mb{A}$.

Given a matrix $\mb{A}$ from the environment that $\mc{B}$ is in,
it first generates all parameters $\ms{pp}$ as we do in $\ms{GSetup}$, then invokes $\mc{A}$ with $\ms{pp}$, and then proceeds as described in the experiment $\mb{Exp}_{\ms{FDGSwD},\mcA}^{\mathsf{non-frame}}(\lambda)$. Here $\mc{B}$ can consistently answer all the oracle queries made by $\mc{A}$. When $\mc{A}$ outputs $(M^*,\Sigma^*,\ms{bin}(j^*),\Pi_{\ms{trace}}^*,\ms{info}_\tau)$ and wins, $\mc{B}$ then proceeds as in Lemma~\ref{lemma:non-frame} until $\mathcal{B}$ extracts  the witness $\zeta'=(\cdots, \mathsf{bin}(j'),\cdots)$. By  the correctness of our encryption scheme, $\mathbf{c}_1^*$ is decrypted to $\mathsf{bin}(j')$, meaning $\mathsf{bin}(j')$ is the actual signer of  $\Sigma^*$. The fact that $\mathcal{A}$ wins implies either case~(i) $\mathsf{Judge}(\mathsf{gpk},\mathsf{bin}(j^*),\mathsf{info}_{\tau}, \Pi_{\ms{trace}}^*, M^*, \Sigma^*)=1$  or case~(ii) $$\mathsf{DJudge}(\mathsf{gpk},\mathsf{bin}(j^*),\mathsf{info}_{\tau}, \Pi_{D(\ms{bin}(j^*))}, M^*, \Sigma^*)=0.$$
Recall that $\mathsf{bin}(j^*)$ is the framed user outputted by the adversary.
We claim that
%\begin{eqnarray}\label{eqn:no-name}
$\mathsf{bin}(j')=\mathsf{bin}(j^*)$
%\end{eqnarray}
in both case (i) and case (ii). Lemma~\ref{lemma:non-frame} proves that the equality holds in case~(i). Now we prove it also holds in case~(ii). Otherwise, $\mathsf{bin}(j')\neq\mathsf{bin}(j^*)$, i.e., $\Pi_{D(\mathsf{bin}(j^*))}$ is a denial proof for user $\mathsf{bin}(j^*)$, who is not the actual signer $\mathsf{bin}(j')$ of  $\Sigma^*$. Therefore, by the completeness of the underlying argument system, $\Pi_{D(\mathsf{bin}(j^*))}$ will be accepted by the $\mathsf{DJudge}$ algorithm, a contradiction of case~(ii). Hence, the equality holds in case~(ii). To this point, $\mathcal{B}$ proceeds again as in Lemma~\ref{lemma:non-frame}. In the end, $\mathcal{B}$ either breaks the security of our accumulator or solves a $\mathsf{SIS}_{n,m,q,1}^{\infty}$ instance. Therefore, our $\mathsf{FDGSwD}$ scheme satisfies non-frameability.
\smallskip

\noindent{\sc{Traceability.}} Except for the analysis of Lemma~\ref{lemma:traceability}, we need to consider an additional case~(iii) in which $\mathcal{A}$ wins and prove  that this case happens with negligible probability. Case~(iii) is the event that $\mathcal{A}$ %generates a valid signature that is traced to an active user $\mathsf{bin}(j')$ together with a proof $\Pi_{\ms{trace}}$ accepted by the $\mathsf{Judge}$ algorithm, but the denial opening proof $\Pi_{D(\ms{bin}(j'))}$ is also accepted by $\mathsf{DJudge}$ algorithm.
outputs $(\mathsf{info}_{\tau}, M,\Sigma)$, where $\Sigma=(\Pi_{\ms{gs}},\mathbf{c}_1,\mathbf{c}_2)$ is a valid signature;  the (honestly generated) opening result is $(\mathsf{bin}(j),\Pi_{\ms{trace}})$, where $\mathsf{bin}(j)$ is an active user and $\Pi_{\ms{trace}}$ is accepted by the $\mathsf{Judge}$ algorithm; and then $\mathcal{A}$ outputs a denial proof $\Pi_{D(\mathsf{bin}(j))}$, which is accepted by the $\mathsf{DJudge}$ algorithm, after seeing the opening result.

To prove this case happens with negligible probability, we proceed as in Lemma~\ref{lemma:traceability} and obtain that $\mathsf{bin}(j')=\mathsf{bin}(j)$ with overwhelming probability, where $\mathsf{bin}(j')$ is part of the extracted witness, meaning $\mathbf{c}_1,\mathbf{c}_2$ are encryptions of $\mathsf{bin}(j')$.  %It then follows  that user $\mathsf{bin}(j)$ is an active user and $\Pi_{\ms{trace}}$ will be accepted by the $\mathsf{Judge}$ algorithm by the security of our accumulator and completeness of the underlying argument system generating $\Pi_{\ms{trace}}$.
The assumption that $\Pi_{D(\mathsf{bin}(j))}$ will be accepted by the $\mathsf{DJudge}$ algorithm  implies that $\mathsf{bin}(j')-\mathsf{bin}(j)\neq \mathbf{0}$ with overwhelming probability, by the soundness of the underlying argument system generating $\Pi_{D(\cdot)}$. This results in a contradiction. Therefore, case~(iii) only happens with negligible probability, which implies that our $\mathsf{FDGSwD}$ scheme is traceable.
\smallskip

\noindent{\sc{Tracing soundness.}} Except for the analysis of Lemma~\ref{lemma:tracing-soundness}, we need to additionally show that it is infeasible for the adversary  to generate a valid signature together with a denial proof (for the actual signer) that is accepted by the $\mathsf{DJudge}$ algorithm. In other words,   we need to prove the event, when $\mathcal{A}$ outputs $(M,\Sigma,\ms{uid}_0,\Pi_{\mathsf{trace},0},\Pi_{D(\mathsf{uid}_0)},\ms{uid}_1,\Pi_{\mathsf{trace},1},\ms{info}_{\tau})$, both $\Pi_{\ms{trace},0}$ and $\Pi_{D(\mathsf{uid}_0)}$ are accepted, happens with negligible probability. Parse $\Sigma$ as $(\Pi_{\ms{gs}},(\mathbf{c}_{1,1},\mathbf{c}_{1,2}),\mathbf{c}_2)$. Since $\Pi_{\ms{trace},0}$ is a valid opening proof, as in Lemma~\ref{lemma:tracing-soundness}, we can extract witness $\mathbf{S}_1,\mathbf{E}_1,\mathbf{y}$ such that $\mathbf{S}_1^\top \cdot \mathbf{B} + \mathbf{E}_1 = \mathbf{P}_1 \bmod q$ and
\begin{eqnarray}\label{equation:no-name-again}
\mathbf{c}_{1,2} - \mathbf{S}_1^\top\cdot \mathbf{c}_{1,1} = \mathbf{y} + \lfloor q/2 \rfloor\cdot \mathsf{uid}_0 \bmod q.
\end{eqnarray}
The fact that $\Pi_{D(\mathsf{uid}_0)}$ is accepted by the $\mathsf{DJudge}$ algorithm implies decryption of $(\mathbf{c}_{1,1},\mathbf{c}_{1,2})$ is not equal to $\mathsf{uid}_{0}$, due to the soundness of the underlying argument system. However, from equation~(\ref{equation:no-name-again}), we know decryption of $(\mathbf{c}_{1,1},\mathbf{c}_{1,2})$ is $\mathsf{uid}_{0}$, by correctness of our encryption scheme. This results in a contradiction. Therefore, the considered event occurs only with negligible probability, and the \textsf{FDGSwD} scheme is tracing sound.

\subsection{Zero-Knowledge Argument of ``Denial Opening''}\label{subsection:ZK-denial-decryption}

\begin{comment}
As shown by Ishida et al.~\cite{IEHST16}, the main technical challenge towards realizing the deniability functionality consists of constructing a \textsf{ZK} proof/argument that a given ciphertext does not decrypt to a particular message. Such a mechanism is non-trivial to realize in general, but as we mentioned earlier, by using the Stern-like trick for proving inequality, it can be done quite easily for our scheme.
Let us look at the equation
\[
\mathbf{c}_{1,2} - \mathbf{S}_1^\top\cdot \mathbf{c}_{1,1} = \mathbf{y} + \lfloor q/2 \rceil\cdot \mathbf{b}' \bmod q,
\]
that appears in the argument of correct decryption from Section~\ref{subsection:correct-decryption}. In the context we are considering, the plaintext $\mathbf{b}' \in \{0,1\}^\ell$ should be hidden, and we have to prove that $\mathbf{b}' \neq \mathbf{b}$, for some given $\mathbf{b} \in \{0,1\}^\ell$. Subtracting both sides of the equation by $\lfloor q/2 \rceil\cdot \mathbf{b}$, we get the equation:
\[
\mathbf{c}_{1,2} - \mathbf{S}_1^\top\cdot \mathbf{c}_{1,1} - \lfloor q/2 \rceil\cdot \mathbf{b} = \mathbf{y} + \lfloor q/2 \rceil\cdot \mathbf{b}^* \bmod q,
\]
where we have to prove that $\mathbf{b}^* \neq \mathbf{0}^\ell$. This is done just by extending $\mathbf{b}^*$ to a vector in $\mathsf{B}_\ell^{2\ell-1}$ and randomly permuting the extended vector!

Moreover, the technique is also applicable for several other existing lattice-based group signatures~\cite{LNW15,LLNW16,LLMNW16-dgs} which employ similar tracing mechanisms as ours.

\end{comment}
We now present the underlying $\mathsf{ZKAoK}$ used by the tracing manager to generate $\Pi_{D(\cdot)}$. The protocol allows the prover to show that a given ciphertext does not decrypt to a particular plaintext. The only difference of this protocol with the one described in Section~\ref{subsection:correct-decryption} is that we need to additionally prove knowledge of a \emph{non-zero} vector $\mathbf{b}\in\{-1,0,1\}^{\ell}$. To this end, we will use the following extending-permuting  technique for Stern-like protocols, which is developed based on Ling et al.'s idea~\cite{LNSW13}.

For any integer $\mathfrak{m}\in\mathbb{Z}_{+}$, define $\widetilde{\mathsf{R}}_{\mathfrak{m}}^{3}$ to be the set of all vectors in $\{-1,0,1\}^{3\mathfrak{m}-1}$ that has $\mathfrak{m}$ coordinates equal to $-1$ and $1$, and $\mathfrak{m}-1$ coordinates equal to $0$.  To prove that $\mathbf{b}\in\{-1,0,1\}^{\ell}$ and $\mathbf{b}\neq \mathbf{0}^{\ell}$, we append $2\ell-1$ ``dummy'' entries to $\mathbf{b}$ to obtain $\mathbf{b}^*\in \widetilde{\mathsf{R}}_{\ell}^3$.
Observe that $\mathbf{b}\in\{-1,0,1\}^\ell$ is non-zero if and only if  the extended vector $\mathbf{b}^*\in\widetilde{\mathsf{R}}_{\ell}^3$. Furthermore, for any $\pi_{b}\in\mathcal{S}_{3\ell-1}$, the following equivalence holds:  \[\mathbf{b}^*\in \widetilde{\mathsf{R}}_{\ell}^3\Leftrightarrow \pi_{b}(\mathbf{b}^*)\in \widetilde{\mathsf{R}}_{\ell}^3.\]

Now, by using the same technique described  in Section~\ref{subsection:correct-decryption}, we can unify equations in~(\ref{equation:pi-denial}) into \[\widetilde{\mathbf{M}}\cdot \widetilde{\mathbf{z}}=\widetilde{\mathbf{u}}\bmod q,\] where $\widetilde{\mathbf{M}}$ and $\widetilde{\mathbf{u}}$ are built from public input while $\widetilde{\mathbf{z}}=(\mathbf{z}\|\mathbf{b}^*)\in\{-1,0,1\}^{\widetilde{D}}$, where $\mathbf{z}\in \mathsf{R}_{D'}^3$ as defined in Section~\ref{subsection:correct-decryption}, $\mathbf{b}^*\in\widetilde{\mathsf{R}}_{\ell}^{3}$, and $\widetilde{D}=D+3\ell-1$.

Define $\widetilde{\mathsf{VALID}}=\mathsf{R}_{D'}^3\times \widetilde{\mathsf{R}}_{\ell}^{3}$ and $\widetilde{\mathcal{S}}=\mathcal{S}_{D}\times \mathcal{S}_{3\ell-1}$. For every $\widetilde{\eta}=(\eta,\pi_{b})\in \widetilde{\mathcal{S}}$, define $\Gamma_{\widetilde{\eta}}=\widetilde{\eta}$. It can be seen that our secret vector $\widetilde{\mb{z}}$ belongs to this tailed set $\widetilde{\ms{VALID}}$ and the following equivalence holds:
\[\widetilde{\mb{z}}\in\widetilde{\ms{VALID}}\Longleftrightarrow\Gamma_{\widetilde{\eta}}(\widetilde{\mb{z}})\in\widetilde{\ms{VALID}},\]
for any $\widetilde{\eta}\in\widetilde{\mathcal{S}}$.
At this point, we can see that the desired statistical $\mathsf{ZKAoK}$  protocol can be obtained from the one in Fig.~\ref{Figure:Interactive-Protocol}.

We remark that the above technique for obtaining zero-knowledge argument of ``denial opening'' is also applicable for several other existing lattice-based group signatures~\cite{LNW15,LLNW16,LLMNW16-dgs} which employ similar tracing mechanisms as ours.

\medskip

\noindent
{\sc Acknowledgements. } The authors would like to thank Beno\^{\i}t Libert for helpful comments and discussions.
The research is supported by Singapore Ministry of Education under Research Grant
MOE2016-T2-2-014(S).

\appendix

\section{Proof of Theorem~\ref{Theorem:zk-protocol}}\label{appdendix:proof-for-stern}

We first restate Theorem~\ref{Theorem:zk-protocol}.
\begin{theorem}%\label{Theorem:zk-protocol}
The protocol in Fig.~\ref{Figure:Interactive-Protocol} is a statistical \emph{\textsf{ZKAoK}} with perfect completeness, soundness error~$2/3$, and communication cost~$\widetilde{\mathcal{O}}(D\log q)$. Namely:
\begin{itemize}
\item There exists a polynomial-time simulator that, on input $(\mathbf{M}, \mathbf{u})$, outputs an accepted transcript statistically close to that produced by the real prover.
\item There exists a polynomial-time knowledge extractor that, on input a commitment $\mathrm{CMT}$ and $3$ valid responses $(\mathrm{RSP}_1,\mathrm{RSP}_2,\mathrm{RSP}_3)$ to all $3$ possible values of the challenge $Ch$, outputs $\mathbf{z}' \in \mathsf{VALID}$ such that $\mathbf{M}\cdot \mathbf{z}' = \mathbf{u} \bmod q.$
\end{itemize}
\end{theorem}
%\end{comment}
%\begin{comment}
\begin{proof}
It can be checked that the protocol has perfect completeness: If an honest prover follows the protocol, then he always gets accepted by the verifier. It is also easy to see that the communication cost is bounded by $\widetilde{\mathcal{O}}(D \log q)$.

%We  now prove that the protocol is a statistical zero-knowledge argument of knowledge. \smallskip

\smallskip
\noindent
{\bf Zero-Knowledge Property. } We construct a \textsf{PPT} simulator $\mathsf{SIM}$ interacting with a (possibly dishonest) verifier $\widehat{\mathcal{V}}$, such that, given only the public input, $\mathsf{SIM}$ outputs with probability negligibly close to $2/3$ a simulated transcript that is statistically close to the one produced by the honest prover in the real interaction.

The simulator first chooses a random $\overline{Ch} \in \{1,2,3\}$ as a prediction of the challenge value that $\widehat{\mathcal{V}}$ will \emph{not} choose.

\noindent
\textbf{Case }$\overline{Ch}=1$: Using basic linear algebra over $\mathbb{Z}_q$, $\mathsf{SIM}$ computes a vector $\mathbf{z}' \in \mathbb{Z}_q^D$ such that $\mathbf{M}\cdot \mathbf{z}' = \mathbf{u} \bmod q.$
Next, it samples $\mathbf{r}_z \xleftarrow{\$} \mathbb{Z}_q^D$, $\eta \xleftarrow{\$} \overline{\mathcal{S}}$, and randomness $\rho_1, \rho_2, \rho_3$ for $\mathsf{COM}$. Then, it sends $\mathrm{CMT}= \big(C'_1, C'_2, C'_3\big)$ to $\widehat{\mathcal{V}}$, where
    \begin{gather*}
        C'_1 =  \mathsf{COM}(\eta, \mathbf{M}\cdot \mathbf{r}_z; \rho_1), \\
        C'_2 =  \mathsf{COM}(\Gamma_{\eta}(\mathbf{r}_z); \rho_2), \quad
        C'_3 =  \mathsf{COM}(\Gamma_{\eta}(\mathbf{z}' + \mathbf{r}_z); \rho_3).
    \end{gather*}
Receiving a challenge $Ch$ from $\widehat{\mathcal{V}}$, the simulator responds as follows:
\begin{itemize}
\item If $Ch=1$: Output $\bot$ and abort.
\item If $Ch=2$: Send $\mathrm{RSP} = \big(\eta, \mathbf{z}' + \mathbf{r}_z, \rho_1, \rho_3 \big)$.
\item If $Ch=3$: Send $\mathrm{RSP} = \big(\eta, \mathbf{r}_z, \rho_1, \rho_2\big)$.
\end{itemize}

\noindent
\textbf{Case }$\overline{Ch}=2$: $\mathsf{SIM}$ samples $\mathbf{z}' \xleftarrow{\$} \mathsf{VALID}$, $\mathbf{r}_z \xleftarrow{\$} \mathbb{Z}_q^D$, $\eta \xleftarrow{\$}  \overline{\mathcal{S}}$, and randomness $\rho_1, \rho_2, \rho_3$ for $\mathsf{COM}$. Then it sends $\mathrm{CMT}= \big(C'_1, C'_2, C'_3\big)$ to $\widehat{\mathcal{V}}$, where
    \begin{gather*}
        C'_1 =  \mathsf{COM}(\eta, \mathbf{M}\cdot \mathbf{r}_z; \rho_1), \\
        C'_2 =  \mathsf{COM}(\Gamma_{\eta}(\mathbf{r}_z); \rho_2), \quad
        C'_3 =  \mathsf{COM}(\Gamma_{\eta}(\mathbf{z}' + \mathbf{r}_z); \rho_3).
    \end{gather*}
Receiving a challenge $Ch$ from $\widehat{\mathcal{V}}$, the simulator responds as follows:
\begin{itemize}
\item If $Ch=1$: Send $\mathrm{RSP} = \big(\Gamma_\eta(\mathbf{z}'), \Gamma_\eta(\mathbf{r}_z), \rho_2, \rho_3\big)$.
\item If $Ch=2$: Output $\bot$ and abort.
\item If $Ch=3$: Send $\mathrm{RSP} = \big(\eta, \mathbf{r}_z, \rho_1, \rho_2\big)$.
\end{itemize}

\smallskip

\noindent
\textbf{Case }$\overline{Ch}=3$: $\mathsf{SIM}$ samples $\mathbf{z}' \xleftarrow{\$} \mathsf{VALID}$, $\mathbf{r}_z \xleftarrow{\$} \mathbb{Z}_q^D$, $\eta \xleftarrow{\$}  \overline{\mathcal{S}}$, and randomness $\rho_1, \rho_2, \rho_3$ for $\mathsf{COM}$. Then it sends $\mathrm{CMT}= \big(C'_1, C'_2, C'_3\big)$ to $\widehat{\mathcal{V}}$, where
$C'_2 =  \mathsf{COM}(\Gamma_{\eta}(\mathbf{r}_z); \rho_2)$, $C'_3 =  \mathsf{COM}(\Gamma_{\eta}(\mathbf{z}' + \mathbf{r}_z); \rho_3)$ as in the previous two cases, while
    \begin{eqnarray*}
        C'_1 =  \mathsf{COM}(\eta, \mathbf{M}\cdot (\mathbf{z}'+ \mathbf{r}_z) - \mathbf{u}; \rho_1).
    \end{eqnarray*}
Receiving a challenge $Ch$ from $\widehat{\mathcal{V}}$, it responds as follows:
\begin{itemize}
  \item If $Ch=1$: Send $\mathrm{RSP}$ computed as in the case $(\overline{Ch}=2, Ch=1)$.
  \item If $Ch=2$: Send $\mathrm{RSP}$ computed as in the case $(\overline{Ch}=1, Ch=2)$.
 \item If $Ch=3$: Output $\bot$ and abort.
\end{itemize}
%\smallskip

\noindent
We observe that, in every case, since $\mathsf{COM}$ is statistically hiding, the distribution of the commitment $\mathrm{CMT}$ and the distribution of the challenge~$Ch$ from~$\widehat{\mathcal{V}}$ are statistically close to those in the real interaction. Hence, the probability that the simulator outputs~$\bot$ is negligibly close to~$1/3$. Moreover, one can check that whenever the simulator does not halt, it will provide an accepted transcript, the distribution of which is statistically close to that of the prover in the real interaction. In other words, we have constructed a simulator that can successfully impersonate the honest prover with probability negligibly close to~$2/3$.

\smallskip

\noindent
{\bf Argument of Knowledge.} Suppose that $\mathrm{RSP}_1 = (\mathbf{t}_z, \mathbf{t}_r, \rho_{2}, \rho_{3})$, $\mathrm{RSP}_2 = (\eta_2, \mathbf{z}_2, \rho_{1}, \rho_{3})$, $\mathrm{RSP}_3 = (\eta_3, \mathbf{z}_3, \rho_{1}, \rho_{2})$ are $3$ valid responses to the same commitment $\mathrm{CMT} = (C_1, C_2, C_3)$, with respect to all $3$ possible values of the challenge. The validity of these responses implies that:
\[
\begin{cases}
\mathbf{t}_z \in \mathsf{VALID}; \hspace*{5pt}
C_1 = \mathsf{COM}(\eta_2, \mathbf{M}\cdot \mathbf{z}_2 - \mathbf{u};\rho_1) = \mathsf{COM}(\eta_3, \mathbf{M}\cdot \mathbf{z}_3; \rho_1); \\
C_2 = \mathsf{COM}(\mathbf{t}_r; \rho_2) = \mathsf{COM}(\Gamma_{\eta_3}(\mathbf{z}_3); \rho_2); \\
{C}_3 = \mathsf{COM}(\mathbf{t}_z + \mathbf{t}_r; \rho_3) = \mathsf{COM}(\Gamma_{\eta_2}(\mathbf{z}_2); \rho_3).
\end{cases}
\]
Since \textsf{COM} is computationally binding, we can deduce that:
\[
\begin{cases}
\mathbf{t}_z \in \mathsf{VALID}; \hspace*{2.8pt}
\eta_2 = \eta_3; \hspace*{2.8pt}
\mathbf{t}_r = \Gamma_{\eta_3}(\mathbf{z}_3);
\hspace*{2.8pt} \mathbf{t}_z + \mathbf{t}_r = \Gamma_{\eta_2}(\mathbf{z}_2) \bmod q; \\[2.5pt]
\mathbf{M}\cdot \mathbf{z}_2 - \mathbf{u} = \mathbf{M}\cdot \mathbf{z}_3 \bmod q.
\end{cases}
\]
Since $\mathbf{t}_z \in \mathsf{VALID}$, if we let $\mathbf{z}' = [\Gamma_{\eta_2}]^{-1}(\mathbf{t}_z)$, then $\mathbf{z}' \in \mathsf{VALID}$. Furthermore, we have $$\Gamma_{\eta_2}(\mathbf{z}') + \Gamma_{\eta_2}(\mathbf{z}_3) = \Gamma_{\phi_2}(\mathbf{z}_2) \bmod q,$$
which implies that $\mathbf{z}' + \mathbf{z}_3 = \mathbf{z}_2 \bmod q$, and that
$\mathbf{M}\cdot \mathbf{z}' + \mathbf{M}\cdot \mathbf{z}_3 = \mathbf{M}\cdot \mathbf{z}_2 \bmod q$. As a result, we have $\mathbf{M}\cdot \mathbf{z}' = \mathbf{u} \bmod q$.
This concludes the proof.

\qed
\end{proof}

\end{document}